\documentclass{acmsmall} 

\usepackage{amsmath,amsfonts,amssymb}

\usepackage{bm}

\usepackage{epsfig}

\usepackage{xspace}

\usepackage{paralist}

\usepackage{graphicx}

\usepackage{subfig}

\usepackage{boxedminipage}

\usepackage{algorithm}
\usepackage{algpseudocode}

\usepackage{xspace}

\usepackage{paralist}

\usepackage{graphicx}

\usepackage{epstopdf}

\usepackage{boxedminipage}

\usepackage{subfig}

\def\submit{0}   
\ifnum\submit=1
	\newcommand{\full}[1]{}
	\newcommand{\confer}[1]{#1}
\else
	\newcommand{\full}[1]{#1}
	\newcommand{\confer}[1]{}
\fi

\newtheorem{claim}[theorem]{Claim}
\newtheorem{prop}[theorem]{Proposition}

\newcommand{\cA}{{\cal A}}
\newcommand{\cB}{{\cal B}}
\newcommand{\cC}{{\cal C}}
\newcommand{\cD}{{\cal D}}
\newcommand{\cE}{{\cal E}}
\newcommand{\cF}{{\cal F}}

\newcommand{\cH}{{\cal H}}

\newcommand{\cL}{{\cal L}}

\newcommand{\cP}{{\cal P}}

\newcommand{\cR}{{\cal R}}
\newcommand{\cS}{{\cal S}}
\newcommand{\cT}{{\cal T}}
\newcommand{\cU}{{\cal U}}
\newcommand{\cV}{{\cal V}}
\newcommand{\cW}{{\cal W}}

\newcommand{\Sec}[1]{\hyperref[sec:#1]{\S\ref*{sec:#1}}} 
\newcommand{\Eqn}[1]{\hyperref[eq:#1]{(\ref*{eq:#1})}} 
\newcommand{\Fig}[1]{\hyperref[fig:#1]{Fig.\,\ref*{fig:#1}}} 
\newcommand{\Tab}[1]{\hyperref[tab:#1]{Tab.\,\ref*{tab:#1}}} 
\newcommand{\Thm}[1]{\hyperref[thm:#1]{Thm.\,\ref*{thm:#1}}} 
\newcommand{\Lem}[1]{\hyperref[lem:#1]{Lem.\,\ref*{lem:#1}}} 
\newcommand{\Prop}[1]{\hyperref[prop:#1]{Prop.~\ref*{prop:#1}}} 
\newcommand{\Cor}[1]{\hyperref[cor:#1]{Cor.~\ref*{cor:#1}}} 
\newcommand{\Def}[1]{\hyperref[def:#1]{Defn.~\ref*{def:#1}}} 
\newcommand{\Alg}[1]{\hyperref[alg:#1]{Alg.\,\ref*{alg:#1}}} 
\newcommand{\Ex}[1]{\hyperref[ex:#1]{Ex.~\ref*{ex:#1}}} 
\newcommand{\Clm}[1]{\hyperref[clm:#1]{Claim~\ref*{clm:#1}}} 
\newcommand{\Step}[1]{\hyperref[step:#1]{Step~\ref*{step:#1}}} 

\makeatletter
\newcommand{\superimpose}[2]{%
  {\ooalign{$#1\@firstoftwo#2$\cr\hfil$#1\@secondoftwo#2$\hfil\cr}}}
\makeatother

\newcommand{\searrowcirc}{\mathpalette\superimpose{{\searrow}{*}}}

\newcommand{\RR}{\mathbb{R}}

\newcommand{\bx}{Box}

\newcommand{\lis}{{\tt lis}}
\newcommand{\est}{{\tt est}}
\newcommand{\loss}{{\tt loss}}
\newcommand{\alis}{{\tt alis}}
\newcommand{\aloss}{{\tt aloss}}

\newcommand{\inside}{{\tt in}}
\newcommand{\outside}{{\tt out}}
\newcommand{\phase}{{\tt ph}}
\newcommand{\In}{In}
\newcommand{\Out}{Out}

\newcommand{\terminalbox}{{\rm{\bf TerminalBox}}}

\newcommand{\Good}{Good}

\newcommand{\deltapar}{\overline{\delta}}
\newcommand{\deltaval}{1/\errorcont}
\newcommand{\err}{\xi}
\newcommand{\gridprecval}{\frac{1}{(C_2\errorcont)^4}}
\newcommand{\kappaest}{\widehat{\kappa}}
\newcommand{\suff}{{21}}
\newcommand{\sszval}{100\errorcont^3}
\newcommand{\taint}{\eta}
\newcommand{\taintval}{1/10\errorcont}
\newcommand{\taintprob}{\phi}
\newcommand{\taupar}{\overline{\tau}}

\newcommand{\pr}{{\rm Pr}}

\newcommand{\poly}{\textrm{poly}}

\newcommand{\eps}{\varepsilon}

\newcommand{\EX}{\hbox{\bf E}}
\newcommand{\prob}{{\rm Prob}}

\newcommand{\strip}[2]{#1|#2}
\newcommand{\width}{{\bf {\rm  w}}}

\newcommand{\errorprob}{\eps}

\newcommand{\griderror}{\xi}
\newcommand{\gridprec}{\alpha}

\newcommand{\errorcont}{\Psi}

\newcommand{\neterror}{\xi}
\newcommand{\netprec}{\alpha}

\newcommand{\errorterm}{\zeta}

\newcommand{\algcomment}[1]{$\backslash * \backslash$ {#1}}

\newcommand{\ssz}{\sigma}

\newcommand{\classify}{{\rm {\bf Classify}}}

\newcommand{\basicmain}{{\rm {\bf BasicMain}}}
\newcommand{\improvedmain}{{\rm {\bf ImprovedMain}}}
\newcommand{\approxlis}{{\rm {\bf ApproxLIS}}}

\newcommand{\critbox}{{\rm {\bf CriticalBox}}}

\newcommand{\findsplitter}{{\rm {\bf FindSplitter}}}
\newcommand{\splitterfound}{splitter\_found}
\newcommand{\buildgrid}{{\rm {\bf BuildGrid}}}
\newcommand{\buildnet}{{\rm {\bf BuildNet}}}
\newcommand{\lischain}{{\rm {\bf LIS.Chain}}}

\newcommand{\gridchain}{{\rm {\bf GridChain}}}
\newcommand{\point}[2]{\langle #1,#2 \rangle}

\newcommand{\maxval}{{\tt valbound}}
\newcommand{\liserror}{\nu}

\title{Estimating the longest increasing sequence in polylogarithmic time}
\date{}

\begin{document} 


\author{M. SAKS
\affil{Rutgers University} 
C. SESHADHRI
\affil{Sandia National Laboratories, Livermore} 
}

\begin{abstract}
Finding the length of the longest increasing subsequence (LIS) is
a classic algorithmic problem. Let $n$ denote the size
of the array. Simple $O(n\log n)$ algorithms
are known for this problem.  We develop a polylogarithmic time
randomized algorithm that for any constant $\delta > 0$,
estimates the length of the LIS of an array to within an additive error of $\delta n$.
More precisely, the running time of the algorithm
is $(\log n)^c (1/\delta)^{O(1/\delta)}$ where the exponent $c$
is independent of $\delta$.
Previously, the
best known polylogarithmic time algorithms could only
achieve an additive $n/2$ approximation.  With a suitable choice of parameters,
our algorithm also gives, for any fixed $\tau>0$,  a multiplicative
$(1+\tau)$-approximation to the distance to monotonicity $\varepsilon_f$ (the fraction of entries not in the LIS),
whose running time is polynomial in $\log(n)$ and $1/varepsilon_f$.  The best previously known algorithm
could only guarantee an approximation within a factor (arbitrarily close to) 2.
%
\end{abstract}




\begin{bottomstuff}
A preliminary version of this result appeared as \cite{SaSe10}.

This work was supported in part by NSF under grants CCF 0832787 and CCF 1218711.
This work was mainly performed when C. Seshadhri was at IBM Almaden.
It was also
performed at Sandia National
Laboratories, a multiprogram laboratory operated by Sandia
Corporation, a wholly owned subsidiary of Lockheed Martin Corporation,
for the United States Department of Energy's National Nuclear Security
Administration under contract DE-AC04-94AL85000.
\end{bottomstuff}

\maketitle

\section{Introduction}

Finding the length of longest increasing subsequence (LIS) of an array
is a classic algorithmic problem. We are given
a function $f:[n] \rightarrow \mathbb{R}$, which we think of as an array.
An \emph{increasing subsequence} of this array is a sequence
of indices $i_1 < i_2 < \cdots < i_k$ such that 
$f(i_1) \leq f(i_2) \leq \cdots \leq f(i_k)$. An LIS is an 
increasing subsequence
of maximum size. The LIS problem is a standard elementary application
of dynamic programming used in basic algorithms textbooks (e.g. ~\cite{CLRS}).
The obvious dynamic program yields
an $O(n^2)$ algorithm. Fredman~\cite{F75} gave a clever way of maintaining the dynamic program,
leading to an $O(n\log n)$ algorithm.
Aldous and Diaconis~\cite{AD99} use the elegant algorithm of 
\emph{patience sorting} to find
the LIS.  

The size of the \emph{complement} of the LIS is called
the \emph{distance to monotonicity}, and  is equal to the minimum number of values that need to be changed to make
$f$ monotonically nonde
creasing.  We write $\lis_f$ for the length of the LIS
and set $\loss_f=n-\lis_f$. The distance to monotonicity is conventionally defined
as $\eps_f = \loss_f/n$.
For exact algorithms, of course,
finding $\lis_f$ is equivalent to finding $\loss_f$. Approximating these quantities can be 
very different problems.

In recent years, motivated by the increasing ubiquity of massive
sets of data, there has been considerable attention given to
the study of approximate solutions of computational problems on huge
data sets by judicious sampling of the input.
In the context of property testing it was
shown in~\cite{EKK+00,DGLRRS99,FischerSurvey,ACCL1} 
that for any $\eps > 0$,
$O(\eps^{-1}\log n)$ random samples are necessary and sufficient to
distinguish the case that $f$ is increasing ($\loss_f=0$)
from the case that $\loss_f \geq \eps n$. 

In
of~\cite{PRR04,ACCL1}, algorithms for estimating  
distance 
to monotonicity were given.  Both of these algorithms  gave a $2+o(1)$-approximation to $\loss_f$  and
had running time $(\loss_f (\log n)/n)^{O(1)}$, and were the best such algorithms known prior to the present work.
These algorithms 
provide little information about the LIS if $\lis_f$ is between $0$ and $n/2$. 
In this case $\loss_f \geq n/2$, so a 2-approximation to $\loss_f$ may produce the (trivial) estimate $n$ to $\loss_f$.
Indeed, there are simple examples where $\loss_f=n/2$
and the algorithms of ~\cite{PRR04,ACCL1} do exactly this.

Note that for small $\varepsilon>0$,
the situation that $\loss_f=\varepsilon n$ and $\lis_f=(1-\varepsilon) n$
is qualitatively different than the situation $\loss_f=(1-\varepsilon)n$ and $\lis_f = \varepsilon n$.
In the former case the array is ``nearly"
increasing and the known algorithms exploit this structure. 


In this paper, we show how to get $\delta n$-additive approximations
to $\lis_f$ in time polylogarithmic in $n$ for any $\delta >0$.
With high probability, our algorithm outputs an estimate $\est$ such that $|\est - \lis_f| \leq \delta n$.
This is equivalent to getting an additive $\delta n$-approximation
for $\loss_f$.
The existing multiplicative 2-approximation algorithm for $\loss_f$ 
gives the rather weak consequence of an additive $n/2$-approximation for $\lis_f$.
Prior to the present paper, this was the best additive error guarantee that was known.  Here we prove:

\begin{theorem} \label{thm:main} Let $f$ be an array of size $n$
and $\lis_f$ the size of the LIS.
There is a randomized algorithm 
which takes as input an array $f$ and  parameter $\delta > 0$, and outputs
a number $\est$ such that
$|\est - \lis_f| \leq \delta n$ with probability at least 3/4. The running time
is $(1/\delta)^{O(1/\delta)} (\log n)^c$, for some
absolute constant $c$ independent of $n$ and $\delta$.
\end{theorem}

The algorithm of this theorem is obtained from a specific choice of parameters within
our main algorithm.  By using a different  choice of parameters, the same algorithm
provides a
multiplicative $(1+\tau)$-approximation to $\loss_f$ for any $\tau>0$ (improving
on the  $(2+\tau)$-approximations
of~\cite{PRR04,ACCL1}).

\begin{theorem} \label{thm:dist} Let $0<\tau < 1$ and $\eps_f=\loss_f/n$. 
There exists an algorithm with running time
$(1/(\eps_f\tau))^{O(1/\tau)}(\log n)^c$ (where $c$ is an absolute constant) that computes a real number $\eps$
such that with probability at least 3/4,  $\eps_f \in [\eps ,(1+\tau)\eps]$.
\end{theorem}

The error probability of 1/4  in each of these theorems can be reduced to any desired value $\errorprob>0$ by the following standard
method: for an appropriate constant $C$,  repeat the algorithm $C \log(1/\errorprob)$ times and output the median output of the trials.  This output will be outside
the desired estimation interval only if at least half of the trials produce an output outside of the desired estimation interval.    Since for each trial this happens
with probability at most 1/4, we can use a binomial tail bound
 (e.g. \Prop{hoeff1} below), to conclude that the probability that the median lies outside the desired interval
is at most $\errorprob$.

\subsection{Related work and relation to other models}

The field of \emph{property testing}~\cite{RS96,GGR98} deals with finding sublinear, or even constant,
time algorithms for distinguishing whether an input has a property,
or is far from the property (see surveys~\cite{FischerSurvey,RonSurvey,GoldreichSurvey}). The property
of monotonicity has been studied over a various partially ordered domains, especially
the boolean hypercube and the set $[n]$~\cite{GGLRS00,DGLRRS99,ELNRR02,HK03,ACCL1,PRR04,BGJRW09}. Our result
can be seen as a \emph{tolerant tester}~\cite{PRR04}, which can closely
approximate the distance to monotonicity.

The LIS has been studied in detail in the streaming model~\cite{GJKK07,SW07,GG07,EJ08}. Here,
we are allowed a small number (usually, just a single) of passes over
the array and we wish to estimate either the LIS or the
distance to monotonicity. The distance approximations
in the streaming model
are based on the sublinear time approximations.
The technique of counting inversions used in the property
testers and sublinear time distance approximators is a major
component of these algorithms. This problem has
also been studied in the communication
models where various parties may hold different
portions of the array, and the aim is to compute
the LIS with minimum communication. This is usually
studied with the purpose of proving streaming lower bounds~\cite{GJKK07,GG07,EJ08}.
Subsequent work of the authors use dynamic programming methods to design
a streaming algorithm for LIS, giving estimates similar to those obtained here~\cite{SaSe12}.
This does not use the techniques from this result, and is a much easier problem
than the sampling model.

There has been a body of work on studying the Ulam
distance between strings~\cite{AK07,AK08,AIK09,AN10}. For permutations, the Ulam
distance is twice the size of the complement
of the longest common subsequence. Note that Ulam
distance between a permutation and the identity permutation
is basically the distance to monotonicity. There has
been a recent sublinear time algorithm for approximating
the Ulam distance between two permutations~\cite{AN10}. We again note
that the previous techniques for distance approximation
play a role in these results. Our results
may be helpful in getting better approximations
for these problems.

\subsection{Obstacles to additive estimations of the LIS} 

A first approach to estimating the length of the LIS is to take a small random sample $S$ of entries of the array, and
exactly compute the length of the LIS of the sample, $\lis_f(S)$.  Scaling this up to
$n\frac{\lis_f(S)}{|S|}$ gives a natural estimator for $\lis_f$.
A little consideration shows that this estimator can be very inaccurate.
Consider the following example.  Let $K$ be a large constant
and $n=Kt$. For $0 \leq i \leq t-1$ and $0 \leq j < K$,
set $f(iK+j+1)=iK-j$. Refer to \Fig{func1}.  The LIS of this function has size $t=n/K$, 
but a small random sample will almost certainly be completely
increasing and so the estimator is likely to equal $n$.

\begin{figure}[tb]
  \centering
 \includegraphics[width=0.25\textwidth]{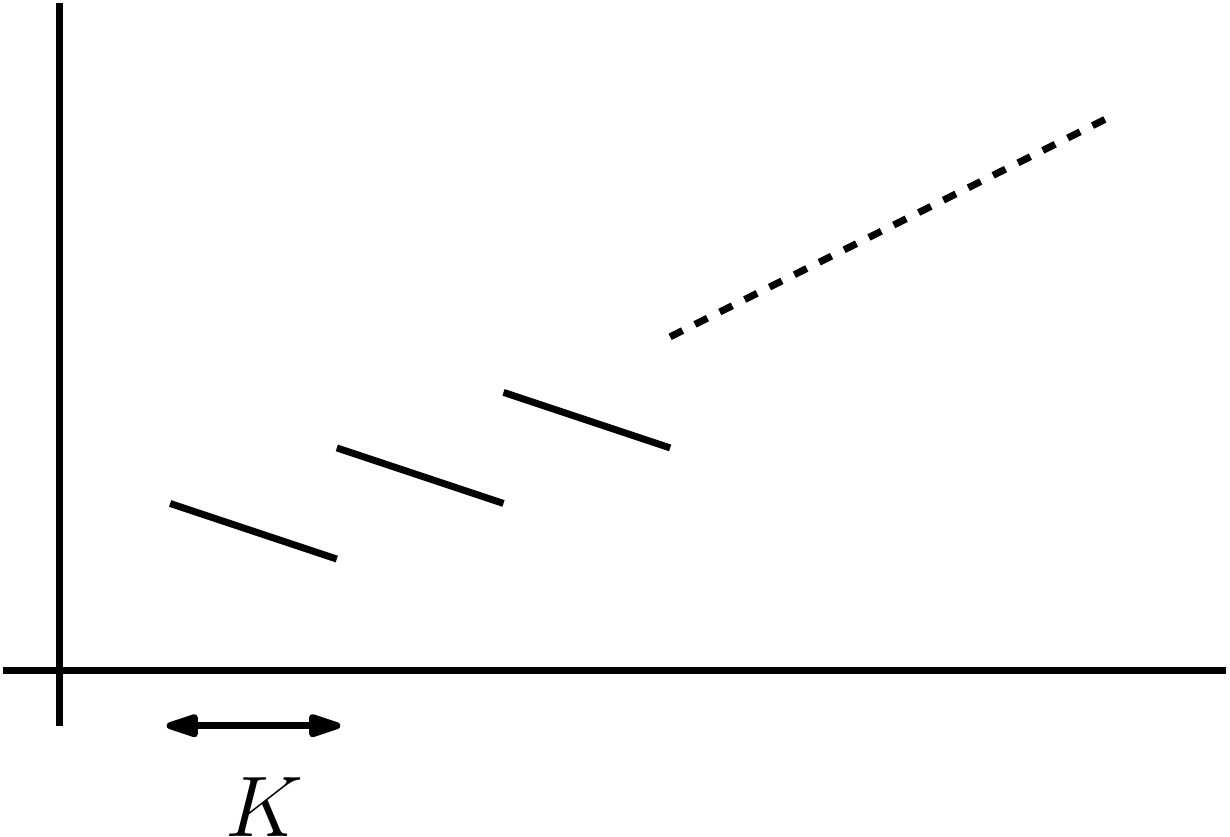}
    \caption{Small random samples will almost always be totally increasing}
  \label{fig:func1}
\end{figure}

An alternative approach to estimating the LIS
is to give an algorithm which, given an index $i \in [n]$, classifies
$i$ as {\em good} or {\em bad} in such a way that:

\begin{itemize}
\item The good indices form an increasing sequence.
\item The number of good indices is close to the size of the LIS, so the number of bad indices
can be bounded.
\end{itemize}

This approach was used in \cite{ACCL1} and \cite{PRR04}.
The classification algorithms presented in those papers essentially work as follows.
Given an index $i$,
for each $k$ between $1$ and $O(\log(n))$ consider each of the index intervals of
the form $[i-(1+\gamma)^k,i+(1-\gamma)^k]$, for all $k$ (for a suitably small $\gamma$), and for each such interval, examine a randomly chosen
subset of indices of polylogarithmic size.   
If, for any one of these samples, 
$i$ is in violation with at least half of the sample then $i$ should be classified as bad.
The analysis of this algorithm shows that the fraction of indices that
are delared bad is at most $(2+o(1))\loss_f/n$ which gives a multiplicative $2+o(1)$-approximation
to $\loss_f$.    This analysis is essentially tight since for the 
function shown in \Fig{func1} where  $\lis_f = n/K$),
these algorithms will classify all indices as bad.  In particular, when $K=2$, the distance to monotonicity
is $n/2$ but the algorithm returns an estimate of $n$.  

If we abstract away the details from this algorithm we see that an index $i$ is classified based on
an $log(n)^{O(1)}$ size sample of indices where the probability that an index $j$ is in the sample is roughly proportional
to $1/|j-i|$.   Call this a {\em sparse proximity-based sample}. It is natural to ask whether there is a better way to use this sample to classify $i$.
The following example shows that there is a strong limitation on the quality of approximation that can be provided by
a classification algorithm based on a sparse proximity-biased sample.
Set $n=64$ and divide the indices
into three contiguous blocks, where the first has size $r$, the second
has size $2r$ and the third has size $3r$.  Consider the sequence $f$ whose
first block is $100r+1,\ldots,101r$,
whose second block is
$1,101r+1,2,101r+2,\ldots, r,101r$ 
and whose third block is some increasing 
subsequence  of
$r+1,\ldots,99r$.  Let $f'$ be a sequence that
agrees with $f$ on the first two blocks. The
final $3r$ positions is some sequence with values in the range
$r+1,\ldots,99r$ but looks like the function in \Fig{func1}.
Refer to \Fig{funcprob} for a pictorial representation
of these sequences.

Notice that in classifying an index $i$ in the first block, a sparse proximity-based is unlikely to be able
to distinguish $f$ from $f'$ and so it will classify $i$ as good or bad the same in both cases.
The LIS of $f$ has size $4r$ (and excludes the first block
of elements) and an increasing sequence that uses any element from the first block has size
at most $2r$.  Hence the algorithm must classify indices $i$ as bad, or incur an additive $2r=n/3$ error.
On the other hand,  the LIS for $f'$ has size $2r$ (and includes all indices in the first block),
and if the algorithm classifies such indices as bad then the algorithm will classify at most $r$ indices as good
and the additive error will be at least $r=n/6$.

Roughly speaking, this example shows that for a classification algorithm that provides better than an $n/6$
approximation, the classification of an index $i$ may involve small scale properties of the sequence
far away from $i$.
Since one can build
many variants of this example, where the size and location of the
critical block is different, and the important scale within the critical
block may also vary, it seems that very global information at all scales
may be required to make a satisfactory decision about any particular index.

\begin{figure}[tb]
  \centering
 \includegraphics[width=1\textwidth]{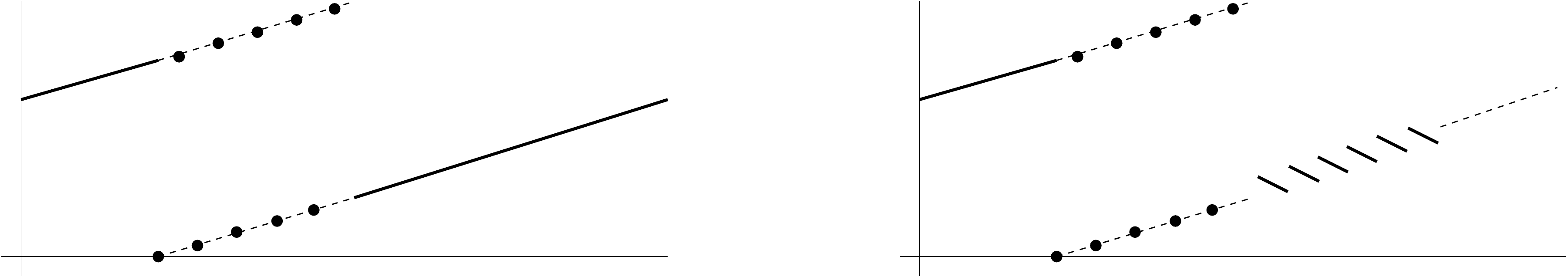}
    \caption{{The sequences/functions $f$ and $f'$ where small scale properties in one block affect
	points of a distant block}\label{fig:funcprob}}\hfill\break
\end{figure}
%
%
%
%
Another perspective is to consider the dynamic program that computes
the LIS.  The dynamic program starts by building and storing
small increasing sequences. Eventually it tries to join them to build
larger and larger sequences. Any one of the currently stored increasing
sequences may extend to
the LIS, while the others may turn out to be incompatible with any
increasing sequence close in size to the LIS.
Deciding among these alternatives requires accurate knowledge of
how partial sequences all over the sequence fit together.
Any sublinear time algorithm that attempts to 
approximate the LIS arbitrarily well has to be 
able to (in some sense) mimic this.

\section{The algorithmic idea and intuition} \label{sec-intuit}

In this section we give an overview of the algorithm.
It is convenient to identify the array/function $f$ with the set
of points $\{\langle i,f(i) \rangle: i \in [n]\}$ in $\RR^2$. We take the natural partial order where $\langle a_1,a_2 \rangle \preceq \langle b_1,b_2 \rangle$ 
if and only if $a_1 \leq b_1$ and $a_2 \leq b_2$.
The LIS corresponds to the longest chain in this partial order.
The axes of the plane will
be denoted, as usual, by $x$ and $y$. We use \emph{index interval} to denote
an interval of indices, and typically denote such an interval by the notation $(x_L,x_R]$ 
which is the set of indices $x$ satisfying $x_L<x \leq x_R$.  The {\em width} of
an interval $I=(x_L,x_R]$, denoted $\width(I)$ is the number $x_L-x_R$ of indices it contains.
We use \emph{value} to denote $y$-coordinates.  Intervals of values
are denoted by closed intervals $[y_L,y_R]$.   A {\em box} $\cB$ is a Cartesian product
of an index interval and value interval.  The width of box, $\width(\cB)$, 
is the width of the corresponding index interval.
We write $X(\cB)$ for the index set of $\cB$.

\subsubsection*{An interactive protocol}

The first idea, which takes its inspiration from complexity theory,	
is to consider an easier problem, that of giving an  
\emph{interactive protocol} for proving a lower bound on $\lis_f$. (Note that 
we will not make any mention of these protocols in the
actual algorithm or in any proof but they provide a useful intuition to keep in mind.)
Suppose that we have a sequence $f$
and two players, a prover and verifier. The prover
has complete knowledge of $f$  and the verifier has query access to $f$.
The prover makes a claim of the form $\lis_f \geq b n$ for
some $b \in (0,1)$.
The verifier wishes to check this claim
by asking the prover questions and querying $f$ on a small
number of indices.
At the end of the interaction, the verifier either accepts or rejects and we require
the following (usual) properties. If $\lis_f \geq b n$, then
there is a strategy of the prover that makes the verifier accept with high
probability. If the prover is lying and $\lis_f < (b -\delta)n$, then for any strategy of the
prover it is unlikely that the verifier will accept.  

The protocol consists of $R$ rounds.  In each round the verifier either accepts or rejects the round, and the round
is designed to have the following properties: (1) If the LIS has size at least $bn$ then the prover can make the verifier
accept with probability at least $b$.  (2) If the LIS has size less than $(b-\delta)n$, then no matter how the prover
behaves the verifier will accept the round with probability less than $(b-\delta)$.  
After performing the $r$ rounds, the verifier will then accept
the prover's claim if the number of accepted rounds is at least $(b-\frac{\delta}{2})R$.  A standard application of the
Chernoff-Hoeffding bound then gives that for $R=\Omega(\log n)$, with high probability the verifier will accept if
the LIS size is at least $bn$ (and the prover follows the protocol) and will reject if the LIS size is at most $(b-\delta)n$.

\begin{figure}[tb]
  \centering
 \includegraphics[width=0.25\textwidth]{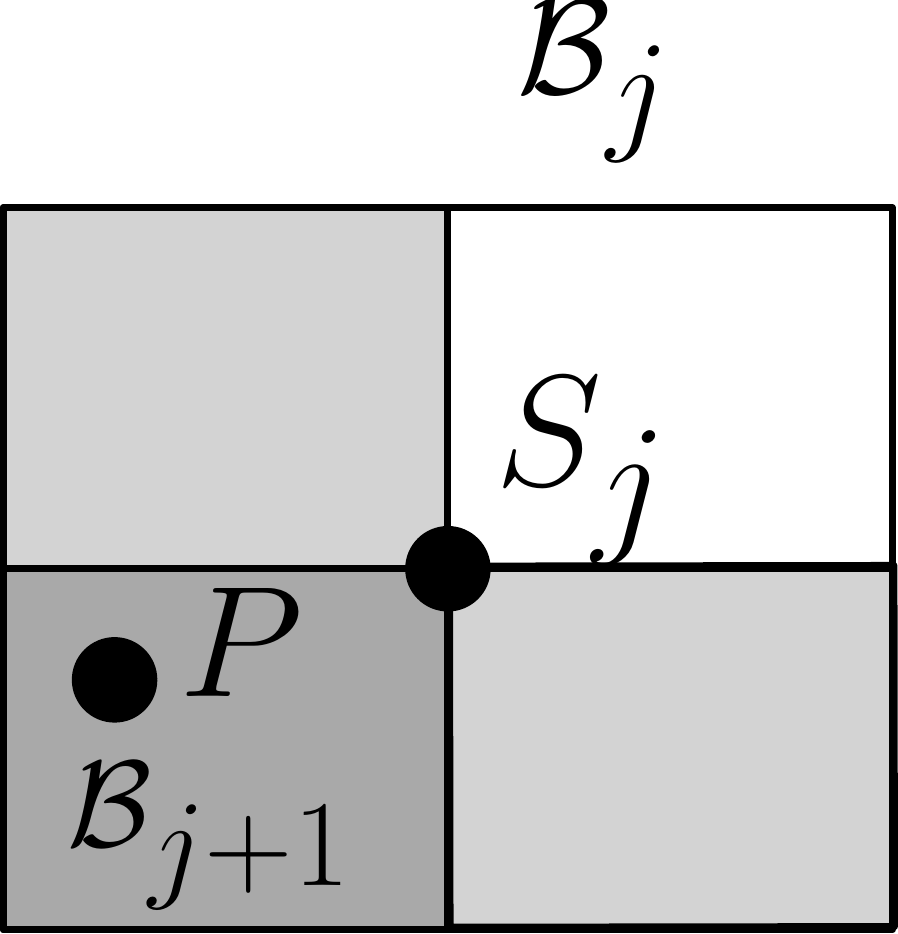}
    \caption{The interactive protocol: if $P$ is in the light gray regions, the verifier rejects. In this case,
    $P \prec S_j$, leading to $\cB_{j+1}$ as shown in dark gray.}
  \label{fig:prot}
\end{figure}

So now we describe how a single round works.  the verifier (secretly) selects an index $i$ uniformly at random from  $(0,n]$.
Let $F(i) = \langle i,f(i) \rangle$.
The prover and verifier jointly generate a nested sequence $\cB_1 \supset \cdots \supset \cB_k$ of boxes
all containing $F(i)$
and a sequence $S_1,\ldots,S_k$ of points. For all $j$, $S_j$ is contained in $\cB_j$.
At the beginning of the $j$th round, $\cB_1,\ldots,\cB_j$ and $S_1,\ldots,S_{j-1}$
are already determined. (We initialize with $\cB_1$ as a box
containing all input points.) Now, the verifier selects some point $S_j$ inside $\cB_j$.
This is called a {\em splitter} for $\cB_j$. The verifier compares $F(i)$ with $S_j$
and rejects if $F(i)$ is incomparable to $S_j$. (Refer to \Fig{prot}.) If $F(i) \preceq S_j$,
then the verifier declares $\cB_{j+1}$ to be the box formed by the bottom-left corner of $\cB_j$ and $S_j$.
If $F(i) \succ S_j$, then $\cB_{j+1}$ is the box formed by $S_j$ and the top-right corner of $\cB_j$.
This ends the $j$th round. 
If this process finally leads to a box $\cB_k$ containing the single point $F(i)$, then
the verifier accepts the index $i$.
%
%
%
%
%

Suppose $\cL$ is some LIS. The prover can make the verifier accept whenever $F(i)$ belongs to $\cL$
by always selecting $S_i$ so that $S_i \in \cL \cap \cB_i$.  (Note that at each step the set $\cL \cap \cB_i$
is an LIS for the box $\cB_i$.) We leave it to the reader to show that
the prover cannot make the verifier to accept with probability higher than $|\cL|/n = \lis_f/n$.

The number of rounds of the protocol depends on how \emph{balanced} the splitters
selected by the prover are.   A splitter $S_i$ is $\rho$-balanced in $\cB_i$ if
it does not belong to the leftmost or rightmost $\rho$ fraction of points of $\cB_i$,
which ensures that each successive box is reduced in size by a factor $1-2\rho$. This would bound the number of rounds
by $O(\log n/\rho)$.  We require that the prover select a $\rho$-balanced splitter for some $\rho=\Theta(\delta)$ (if there
is such a splitter) and halt otherwise.  Iimposing the $\rho$-balance requirement on the prover can decrease the probability
of convincing the verifier by at most $\Theta(\rho)$).

\subsubsection*{Algorithmically searching for a splitter}

Can we  simulate this protocol by an algorithm?  At first glance this seems impossible, since the prover has complete knowledge
 of the array, while the algorithm must ``pay" for any information.

At each step, the all-knowing prover selects a $\rho$-balanced splitter in the current box that belongs to the LIS within that box.
Since our algorithm does not know the LIS, it seems impossible to simulate this.  But since we are only
doing an approximation, we do not need the splitter to be on the LIS, we only need that the splitter
belongs to an increasing sequence that is close to the LIS.  How can we recognize such a splitter?

%
Let $\cB_j$ be the current box being split and $\cL_j$ be the LIS of $\cB_j$.
We say that an input point $P$ is a \emph{violation} with a splitter $S$
if they are incomparable. A good splitter is one which is a violation
with few points in $\cL_j$.
%
How few?  If
we can guarantee that the number of violations of the LIS with the splitter
is at most $\gamma \width(\cB)$, then the total error will be $\gamma \width(\cB)$ times the number
of rounds which is $O(\log n/\rho)$.  The overall error will be a small fraction of $\width(\cB)$ if 
$\gamma$ is a small fraction of $\rho/\log n$.

It will be convenient in this informal discussion to assume that splitters are not necessarily input point (although it turns
out that our algorithm will restrict its search for splitters to input points).
We do not know what $\cL_j$ is, so we look for a splitter $S$ with
a much stronger property: $S$ has a total of $\gamma \width(\cB)$ (input point) violations in $\cB_j$.
We call this a {\em conservative} splitter.  This is a stronger requirement than what is needed for a good splitter because it includes
violations of the splitter with points that are not on the LIS.    If we can find a conservative splitter, it is safe to use it
as the splitter in the interactive proof.  Whether a candidate point is a conservative splitter for box $\cB_j$ can be quickly (approximately) checked by estimating
the fraction of violations that $i$ has within $\cB_j$ by examining a small random sample of
indices from $\cB_j$.    Furthermore, if a non-trivial fraction of points are 
conservative splitters, then one can be found and identified quickly by random sampling. 

Thus conservative splitters are easy to find (if there are enough of them)
and if one is found then we can use it to simulate the prover.  
Of course, there may not be enough conservative splitters for the random search algorithm to succeed; indeed there
may be no conservative splitters.   
A new idea is required to deal with this problem:
boosting the quality of the approximation.  
%

\subsubsection*{Boosting the quality of approximation}

Let us restart our quest for an algorithm from a different starting point:
Given an algorithm that guarantees an additive $\delta n$-approximation to $\lis_f$,
can we use it to get an additive $\delta'$-approximation for some
smaller $\delta'$?  If we could, then by using the known
additive $1/2$-approximation algorithm as a starting point,
we might be able to apply this error reduction method recursively to
achieve any desired error.

Consider the following divide-and-conquer dynamic program for
estimating the LIS of points in a box $\cB$.  
Fix $s$ to be some (balanced) index in $\cB$.
Let $y$ be the value of some point in $\cB$ and set $P = \langle s,y\rangle$.
By varying $y$, we get different points $P$.
Define the box $\cB_{L}(P)$ formed by the bottom-left corner of $\cB$ and $P$.
Analogously, $\cB_{R}(P)$ is formed by $P$ and the top-right corner of $\cB$.
Observe that $\lis(\cB) = \max_P \{\lis(\cB_{L}(P)) + \lis(\cB_{R}(P))\}$.
%
%
This recurrence can be viewed as a  dynamic program for $\lis(\cB)$.
Let $\alis(\cB_{L}(P))$ and $\alis(\cB_{R}(P))$ denote estimates obtained by running our
base approximation algorithm on $\cB_{L}(P)$ and $\cB_{R}(P)$. We can maximize $\alis(\cB_{L}(P))+\alis(\cB_{R}(P))$ over $P$
to get 
an approximation for $\lis(\cB)$ 
While it is too costly to search over all points $P$,
a good approximation can be obtained by maximizing over a polylogarithmic
sample from the points in $\cB$. For the sake of this discussion, we will
assume that the true maximizer $P^*$ can be determined.
We denote the corresponding boxes $\cB_{L}(P^*)$ and $\cB_{R}(P^*)$ by $\cB_L$ and $\cB_R$.

An initial analysis suggests that we gain nothing from this.
If $\alis(\cB_L)$ is a $\delta \width(\cB_L)$-additive approximation and $\alis(\cB_R)$
is a $\delta \width(\cB_R)$-additive approximation then the best we can say is
that the sum is an additive $\delta (\width(\cB_L)+\width(\cB_R))=\delta \width(\cB)$
approximation to $\lis(\cB)$, so we get no advantage.

However, if we make a subtle change in the notion of additive error, then 
an advantage emerges.  
Instead of measuring the additive approximation error
as a multiple of $\width(\cB)$, we measure it as a multiple of $\loss(\cB)=|\cB \cap \cF|-\lis(\cB)$, where $\cF=
\{F(i):i \in [n]\}$ is the set of input points.
Note that in general $\loss(\cB)$
may be much smaller than $\width(\cB)$.  So suppose we have an algorithm that
whose approximation error is at most $\tau \loss(\cB)$.  
Then the additive error of $\alis(\cB_L)+\alis(\cB_R)$ will be at most
$\tau (\loss(\cB_L)+\loss(\cB_R))$, which may be significantly  less than $\tau \loss(\cB)$.
If this can be bounded above by $\tau' \loss(\cB)$ for some $\tau'<\tau$,
the quality of approximation can be boosted.

The nice surprise is that 
$\loss(\cB)-(\loss(\cB_L)+\loss(\cB_R))|
= |\cB\cap \cF| - (|\cB_L \cap \cF| + |\cB_R\cap \cF|)$. This quantity is precisely the number of
points in $\cB \cap \cF$ that are in violation with $P^*$.  This gives the following dichotomy: if the number of such violations is at least
$\mu |\cB \cap \cF|$ then we can take $\tau'$ to be $(1-\mu)\tau$ and boost the quality of approximation from $\tau$
to $\tau'$. Otherwise, the number of violations is less than $\mu |\cB|$ and $P^*$
is  a conservative splitter!  

It is not clear how to apply this dichotomy to interleave the simulation of the interactive protocol
and the boosting idea to  get an algorithm.
But there is a more significant difficulty.
In the search for a good splitter we measured the number of violations
as $\gamma \width(\cB)$ and argued that $\gamma$ should be $O(1/\log n)$.
For the recursive boosting to
work efficiently, we measured the quality of the splitter by $\mu |\cB|$
and for this we need $\mu$ to be at least $\Omega(1/\log\log n)$.
Why? For each level of recursion, we can improve the 
additive approximation from $\tau$ to $\tau(1-\mu)$.
So we will need $1/\mu$ levels of
recursion to improve from $\delta$ to $\delta/2$.
At each level of the recursion, we make at least $2$
recursive calls, for the left and right subproblems
generated by any choice of a splitter. So the total number of iterated recursive calls
is exponential in $1/\mu$.  Since we want the running time to be $\log(n)^{O(1)}$, this  leads to $\mu = \Omega(1/\log\log n)$.

For the interactive protocol simulation, splitters can have at most
$O(|\cB \cap \cF|/\log n)$ violations inside $\cB$, and for the boosting algorithm
every (or nearly every) splitter has $\Omega(|\cB \cap \cF|/\log\log n)$
inside $\cB$.
So while the dichotomy seems promising, we have
a huge gap from an algorithmic perspective.

\subsubsection*{Closing the dichotomy gap}

We seek ways to close the gap in this dichotomy.
Upon further consideration (and using past work in the area as a guide),
it seems fruitful to modify the criterion for a good splitter.
For a candidate splitter $P$ we relax the condition on the maximum number
of violations $\gamma \width(\cB)$ to $\mu|\cB \cap \cF|+ \gamma \width(\cB)$,
where $\gamma$ will be $1/\poly\log n$ and $\mu$ will be a small constant.
(We strengthen this requirement by requiring that a similar
condition holds for various subboxes of $\cB$.)
Note that this condition is weaker than the original condition on splitters.
But we prove that it is good enough to simulate the
interactive protocol.

Now for the other side of the dichotomy. Even if there is no good splitter
satisfying this weaker condition, the above divide-and-conquer
scheme might still fail to boost the quality of approximation.
The  boosting algorithm uses an index $s$ 
to divide the box $\cB$ into two parts and then searches over different $y$ to maximize
$\alis(\cB_L)+\alis(\cB_R)$.  This is essentially solving a longest path problem on
a 3-layer DAG.  In the modified boosting algorithm we do analogous thing,
but where we divide the box
into a larger number of parts and solve a longest path problem on a DAG with more layers.
The number of layers turns out to be $1/\gamma$, where $\gamma$ is
the parameter in the relaxed definition of splitter.

These two ingredients - the simulation of the interactive protocol
and the modified boosting algorithm - are combined together
to give our algorithm.  There are various parameters such as $\gridprec,\mu,\gamma$ involved
in the algorithm, and we must choose their values carefully. We also need to determine how many
levels of boosting are required to get a desired approximation.
A direct
choice of parameters leads to a $(\log n)^{1/\delta}$
approximation algorithm.  
The better algorithm
claimed in \Thm{main} is obtained by a more delicate
version of the algorithm, which involves modifying the various parameters
as the algorithm proceeds.  This reduces
the number of recursive calls needed for boosting
from polylogarithmic in $n$ to a constant depending on $\delta$.
%
%

\section{Preliminaries}
\label{sec:prelim}
\subsection{Basic definitions}
\label{subsec:basic defs}

We write $\mathbb{N}$ for the set of nonnegative integers and $\mathbb{N}^+$
for the set of positive integers.   We typically use interval notation $[a,b]$ and $(a,b]$
to denote intervals of nonnegative integers.  Occasionally we also use interval notation
to denote intervals of real numbers; the context should make it clear which meaning is intended.

Throughout this paper $n$ is an arbitrary but fixed positive integer and we refer to the
set $(0,n]=\{1,\ldots,n\}$ as the {\em index set}.  We let $f$ denote a fixed arbitrary
function mapping $(0,n]$ to $\mathbb{N}^+$.  
Occasionally,  we abuse notation and 
view 0 as an index. 
Also for convenience,
we assume that we are given an upper bound $\maxval$ on the maximum of $f$.    
and define ${\rm range}(f)$ to be the set $[1,\maxval]$.  

For $X \subseteq (0,n]$ we write $f(X)=\{f(x):x \in X\}$, and  
for $Y \subseteq {\rm range}(f)$ we
write $f^{-1}(Y)=\{x \in [n]:f(x) \in Y\}$.

\begin{description}
\item [{\bf Points}]
As usual, $\mathbb{N}^2$ denotes the set of ordered 
pairs of nonnegative integers, which we call {\em points}.  
A point is denoted by $\point{a}{b}$ (rather than $(a,b)$, to avoid confusion
with interval notation).   
The first coordinate of point $P$ is denoted
$x(P)$ and is called the {\em index of $P$}.  The second
coordinate of $P$ is denoted $y(P)$ is called the {\em value of $P$}. 
We denote points by upper case letters, and 
sets of points
by calligraphic letters.

\item[{\bf The sets  $X(\cS)$ and $Y(\cS)$}] For a set $\cS$ of points, we write
$X(\cS)$ for the set $\{x(P):P \in \cS\}$ of indices of points
of $\cS$, and 
$Y(\cS)$ for the set $\{y(P):P \in \cS\}$
of values of points in $\cS$.

\item[{\bf The point $F(x)$, the set $\cF$, and $\cF$-points}]
For $x \in (0,n]$, $F(x)$ denotes the point $\langle x,f(x)\rangle$.
We define $\cF=\{F(x):x \in (0,n]\}$. 
We refer to points in $\cF$ as $\cF$-points.

Observe that for 
a set $\cS$ of points, $F^{-1}(\cS) \subseteq X(\cS)$ is the set of
indices $x$ for which $F(x) \in \cS$.  Since $F$ is a 1-1 function
it follows that the sets $F^{-1}(\cS)$ and $\cF \cap \cS$ are in 1-1 correspondence.

\item [{\bf Relations $P \leq Q$, $P \prec Q$, $P \searrow Q$, $P \searrowcirc Q$ and   sets  $P^{NE}$, $P^{NW}$, $P^{SE}$,$P^{SW}$}]

For points $P,Q$, 
\begin{description}
\item[$P \leq Q$] means
$x(P) \leq x(Q)$ and $y(P) \leq y(Q)$
\item[$P \prec Q$] means $x(P)<x(Q)$ and $y(P) \leq y(Q)$.
\item[$P \searrow Q$] means $x(P) < x(Q)$ and $y(P)>y(Q)$.
\item[$P \searrowcirc Q$] means $x(P) \leq x(Q)$ and $y(P) > y(Q)$.
\end{description}

Observe that for points $P,Q$ with $x(P) \neq x(Q)$, for example, if $P,Q$ are distinct points of $\cF$,
then the conditions $x(P) \leq x(Q)$ and $x(P) \prec x(Q)$ are equivalent,
and the conditions $P \searrow Q$ and $P \searrowcirc Q$ are equivalent. 

For a point $P$, we define the following sets. Refer to \Fig{position}.
\begin{description}
\item
[$P^{NW}$] (for northwest) The set of points $Q$ such that $Q \searrowcirc P$.
\item
[$P^{SW}$] (for southwest) The set of points $Q$ such that $Q \leq P$.
\item
[$P^{NE}$] (for northeast) The set of points $Q$ such that $P \prec Q$.
\item
[$P^{SE}$] (for southeast) The set of points $Q$ such that $P \searrow Q$.

\end{description}

\begin{figure}[tb]
  \centering
 \includegraphics[width=0.25\textwidth]{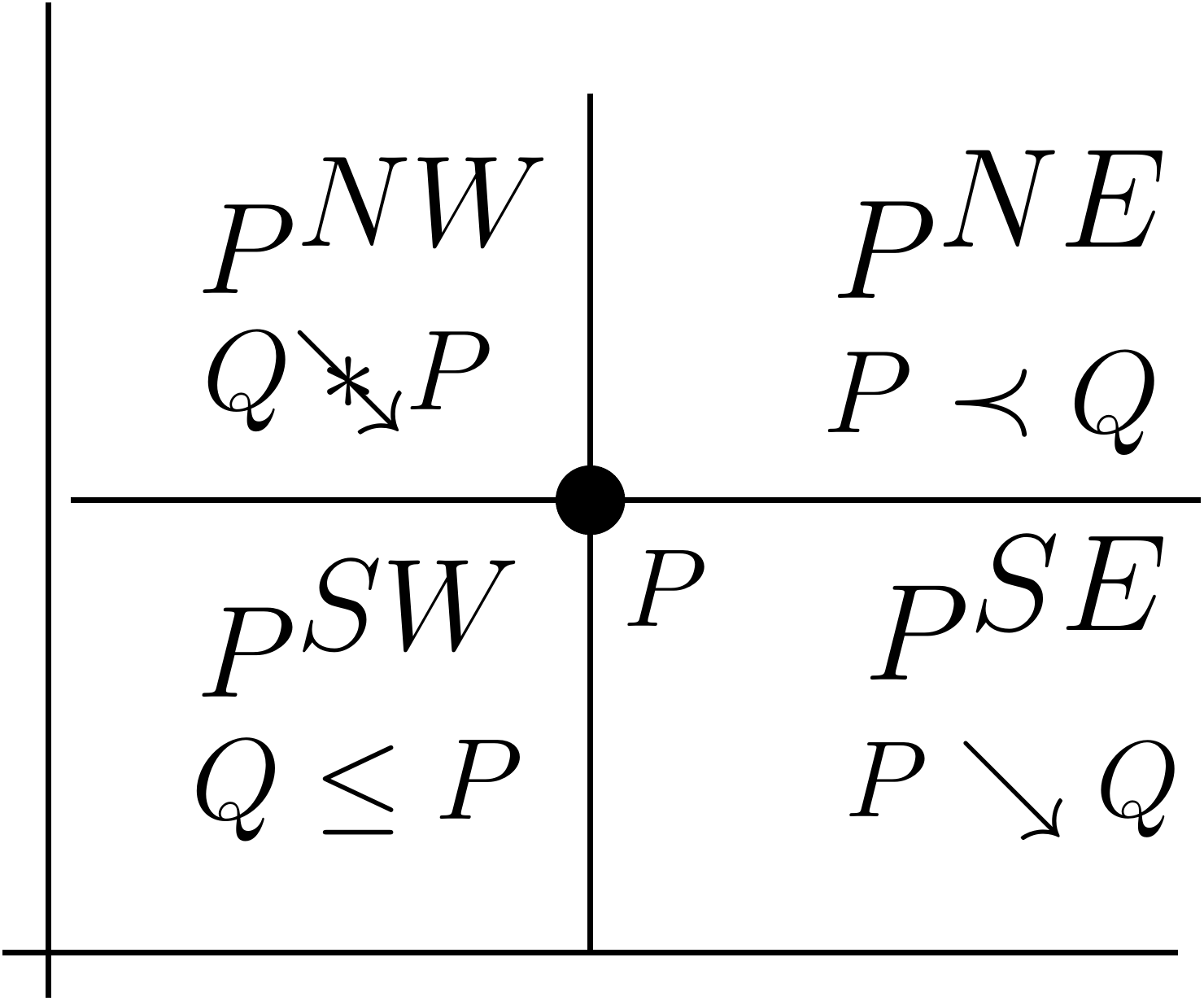}
    \caption{The various positions and sets relative to point $P$. $Q$ denotes a hypothetical point
    in the respective region.}
  \label{fig:position}
\end{figure}

There is an asymmetry in these definitions: For example one might expect that $Q \in P^{NE}$ would
be equivalent to $P \in Q^{SW}$ but this is not the case.  
Notice that this asymmetry disappears if $x(P) \neq x(Q)$.  
One reason to define the sets this way is so that for each point $P$, the sets
 $P^{NW}$,$P^{NE}$, $P^{SW}$, $P^{SE}$ partition $\mathbb{N}^2$.

\item[{\bf Relations $P \sim Q$ ($P$ is comparable to $Q$) and $P\not\sim Q$ ($P$ is a violation with $Q$)}]

For $P,Q \in \cF$ we define:

\begin{description}
\item[$P \sim Q$] means that $P$ and $Q$ are {\em comparable}, i.e., $P \leq Q$ or $Q \leq P$ .
\item[$P \not\sim Q$] means that $P$ and $Q$ are a {\em  incomparable} or a {\em violation}, i.e., either $P \searrow Q$  
or $Q \searrow P$ .  
\end{description}
We emphasize that we only use these terms in the case that $P$ and $Q$ are $\cF$-points.  Since $\cF$-points have distinct indices,
the distinctions between $\leq$ and $\prec$, and between $\searrow$ and $\searrowcirc$ disappear.

\item [{\bf Index  intervals}]   A set of consecutive indices  is called
an {\em index interval}.    An index interval is usually written using the notation
$(a,b]=\{a+1,\ldots,b\}$.  
For an index interval $I$, we define indices $x_L(I)$
and $x_R(I)$, the {\em left and right endpoints of $I$} so that $I=(x_L(I),x_R(I)]$.
Note $x_R(I) \in I$ but $x_L(I) \not\in I$. 

\item [{\bf Value intervals}]
A {\em value interval}
refers to an integer subinterval $[c,d]$ of ${\rm range}(f)$.
For value intervals we always use closed interval notation.  
For a value interval $J$ we define indices $y_B(J)$
and $y_T(J)$, the {\em top and bottom endpoints of $J$}  
such that $J=[y_B(J),y_T(J)]$.  Thus, in contrast with index intervals,
$J$ contains both of its endpoints.

\item[{\bf Box}]
A box $\cB$ is the Cartesian product of a nonempty index interval
$I$ and a nonempty value interval $J$.

Using the notation described above, we have $\cB = X(\cB) \times Y(\cB)$,
$X(\cB)=(x_L(X(\cB)),x_R(X(\cB))]$, and $Y(\cB)=[y_B(Y(\cB),y_T(Y(\cB))]$. To simplify notation we define
$x_L(\cB)=x_L(X(\cB))$ and analogously $x_R(\cB), y_B(\cB), y_T(\cB)$.
%
%
Thus $X(\cB)=(x_L(\cB),x_R(\cB)]$ and $Y(\cB)=[y_B(\cB),y_T(\cB)]$.
We also define the {\em bottom-left point of $\cB$}, 
$P_{BL}(\cB)=\point{x_L(\cB)}{y_B(\cB)}$
and the {\em top-right point of $\cB$},  $P_{TR}(\cB)=\point{x_R(\cB)}{y_T(\cB)}$.
(Note that under these definitions $P_{BL}(\cB) \not\in \cB$ since
$x_L(\cB) \not\in X(\cB)$.)

If $Q,R$ are points with $Q \prec R$, the 
{\em box spanned by $Q,R$}, denoted $\bx(Q,R)$ is the box having $P_{BL}(\cB)=Q$ and $P_{TR}(\cB)=R$.
It is easy to see that 
$\bx(Q,R)=
\{P:Q \prec P \leq R\}$

\item[{\bf Grids}]
A {\em grid} $\Gamma$ is any Cartesian product $I \times J$ where
$I$ is a set of indices and $J$ is a set of values.  Thus 
$\Gamma$ is a grid if and only if $\Gamma = X(\Gamma) \times Y(\Gamma)$.
A box is the special case of a grid in which both $X(\Gamma)$
and $Y(\Gamma)$ are intervals.  

We refer to the sets of the form $\{x\} \times Y(\Gamma)$ for $x \in X(\Gamma)$ as
{\em columns of $\Gamma$} and sets of the form $X(\Gamma) \times \{y\}$ for $y \in Y(\Gamma)$
as {\em rows of $\Gamma$}.  

\item[{\bf The universe $\cU$ and the universal grid $\Gamma(\cU)$}]  The {\em universe $\cU$}  is  the box $(0,n] \times {\rm range}(f)$.
The {\em universal grid} is the grid $(0,n] \times F(0,n]$.  This is the smallest grid that contains $\cF$, and its
size is at most $n^2$.  Every point that is encountered in any of our algorithms belongs to the universal grid.

\item[{\bf width}] The {\em width} of an interval $I$, denoted $\width(I)$, is $|I|=x_R(I)-x_L(I)$.
Similarly, for a box $\cB$, 
$\width(\cB)$ is equal to $|X(\cB)|=
x_R(\cB)-x_L(\cB)$.

\item [{\bf Relation $\triangleleft$ on index intervals and boxes}]
\begin{itemize}
\item For index intervals $I_1,I_2$ we write $I_1 \triangleleft I_2$
if $x_R(I_1)=x_L(I_2)$.  It follows that $I_1 \cap I_2=\emptyset$
and $I_1 \cup I_2$ is
equal to the index interval $(x_L(I_1),x_R(I_2)]$.  
In particular
this implies that $x_1 < x_2$ for all $x_1 \in I_1$ and $x_2 \in I_2$.
\item For boxes $\cB_1,\cB_2$ we write $\cB_1 \triangleleft \cB_2$ to mean
$P_{TR}(\cB_1)=P_{BL}(\cB_2)$.  This is equivalent to
$X(\cB_1) \triangleleft X(\cB_2)$ and $y_T(\cB_1)=y_B(\cB_2)$.
In particular,
this implies  $\cB_1 \cap \cB_2=\emptyset$ and $P_1 \prec P_2$ for all $P_1 \in \cB_1$ and $P_2 \in \cB_2$.
\end{itemize}

\item [{\bf Box sequences}]

A {\em box sequence} is a list of boxes such that each successive box is entirely to the right
of the previous.
We use the notation  $\vec{\cB}$ to denote a box sequence.
We write $\cB \in \vec{\cB}$ to mean that the box $\cB$ appears
in $\vec{\cB}$. 

\begin{figure*}[tb]
  \centering
  \subfloat[$\cB$-strip]{\includegraphics[width=0.35\textwidth]{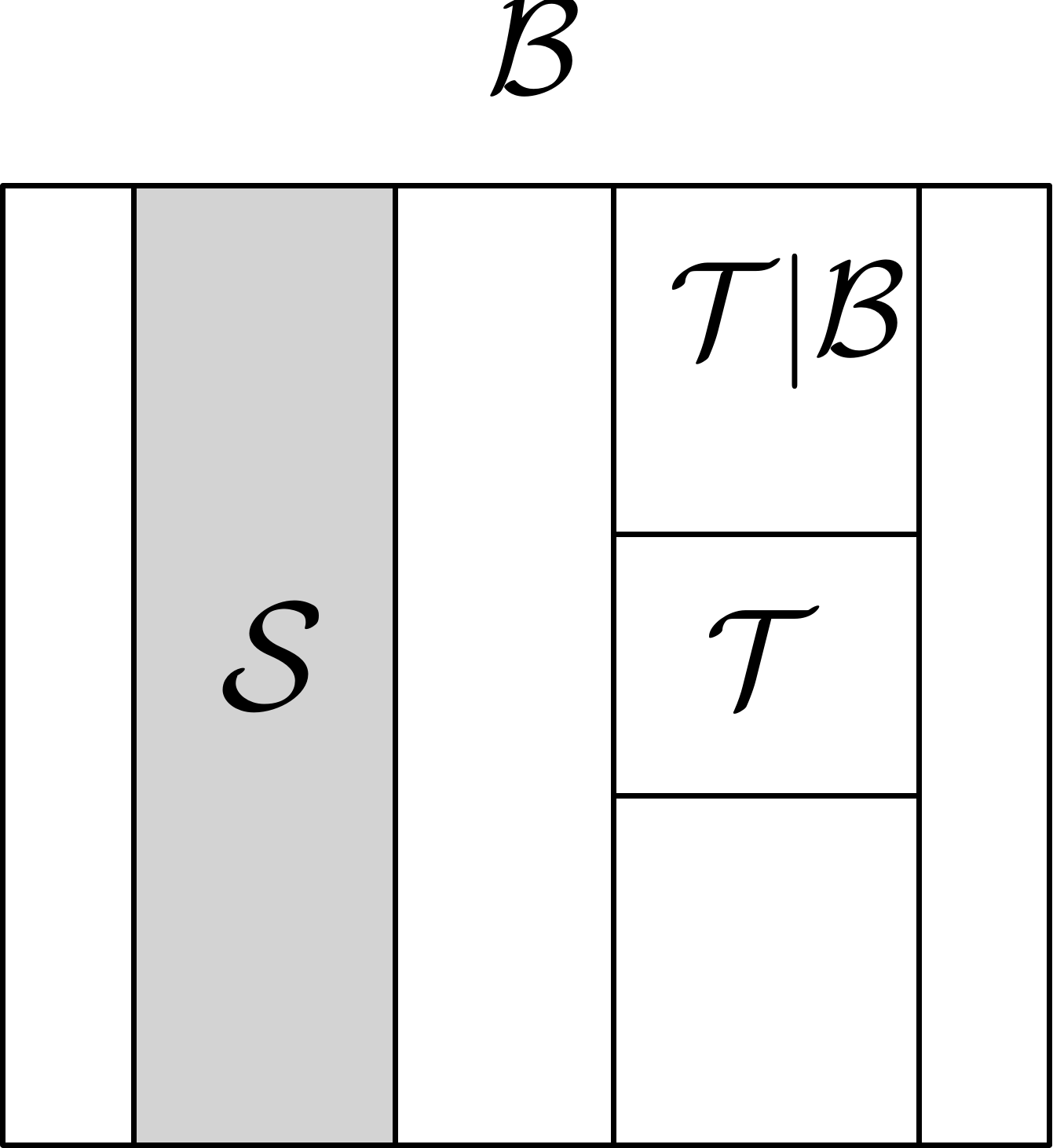} \label{fig:strip}} 
  $\qquad \qquad$
  \subfloat[Box chains]{\includegraphics[width=0.35\textwidth]{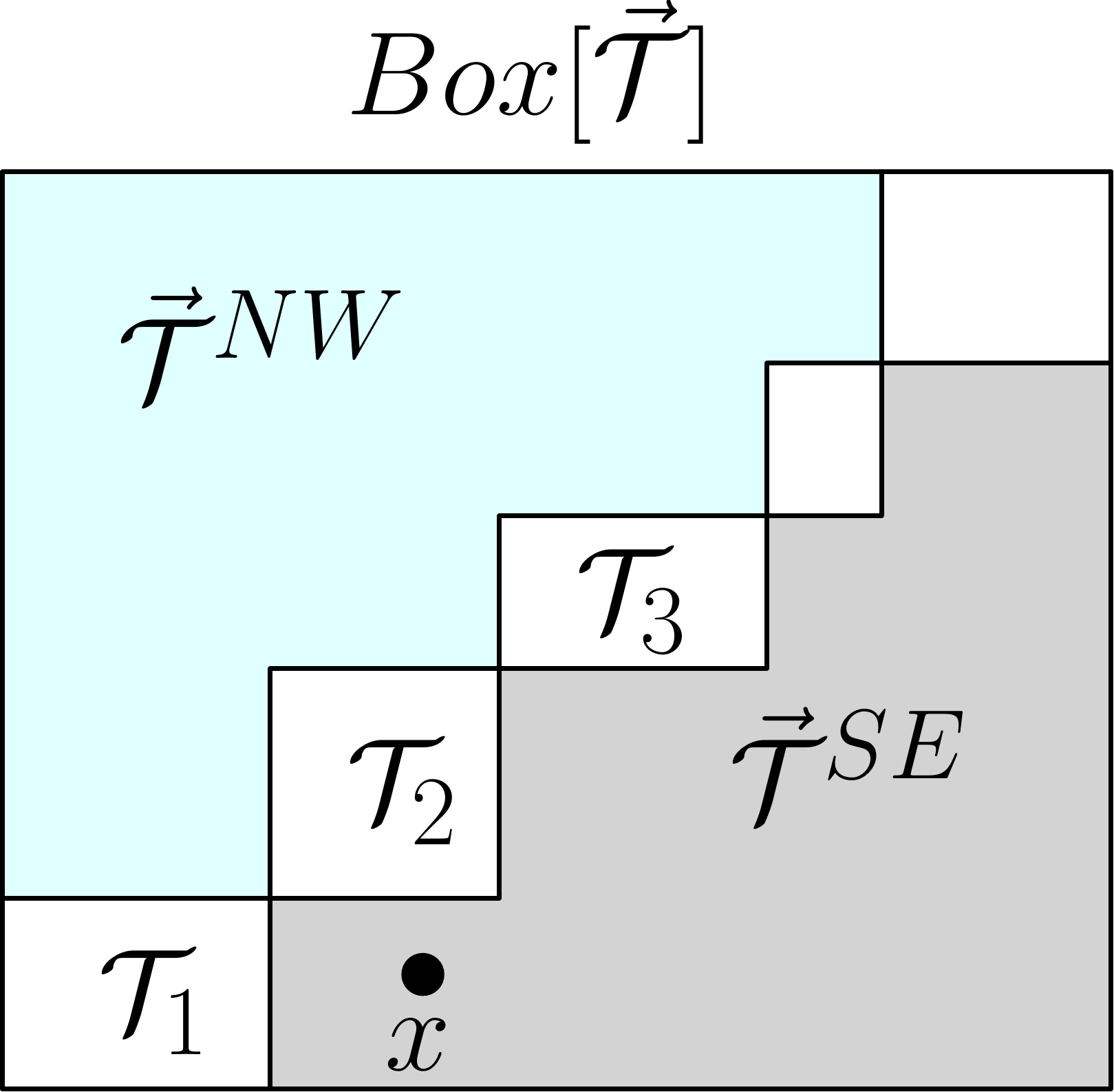} \label{fig:boxchain}}
    \caption{(a) $\cS$ is a $\cB$-strip. The figure also shows the strip $\cT | \cB$ for box $\cT$. 
    (b) The box chain $\vec{\cal T}$: The boxes $\cT_1, \cT_2, \ldots$ form a chain. The box $\vec{\cT}[x]$
    is $\cT_2$. The regions $\vec{\cT}^{NW}$ and $\vec{\cT}^{SE}$ are shaded.}
\end{figure*}

\item [{\bf $\cB$-strips and $\cB$-strip decompositions}] 

If $\cB$ is a box, a $\cB$-strip is a subbox $\cS$ of $\cB$ such that $Y(\cS)=Y(\cB)$.  Thus
a $\cB$-strip has the same vertical extent as $\cB$ and is specified relative to $\cB$
by its index set $X(\cS)$.
If $I \subseteq X(\cB)$ is an index interval then $\strip{I}{\cB}$
denotes the $\cB$-strip with index set $I$.  Similarly if $\cT$ is a subbox of $\cB$
then $\strip{\cT}{\cB}$ denotes the strip $\strip{X(\cT)}{\cB}$.
Refer to \Fig{strip}.

A $\cB$-strip decomposition is a partition of $\cB$ into strips.
A $\cB$-strip decomposition into $r$ strips is specified by a sequence $x_0=x_L(\cB) < x_1 < \cdots < x_r = x_R(\cB)$,
where the $j$th strip is $\strip{(x_{j-1},x_j]}{\cB}$.   We use the sequence notation $\vec{\cS}$ to denote a $\cB$-strip decomposition
in the  natural left-to-right order.

In particular if $\Gamma \subset \cB$ is a grid then $\Gamma$ naturally defines
a strip decomposition of $\cB$ obtained by taking the strips that end at successive columns
of $\Gamma$ with the final strip ending at $x_R(\cB)$.

\item [{\bf Increasing point sequences} and {\bf Box Chains}]

A set of points $\cP$ that can be ordered so that
$P_0 \prec \cdots \prec P_k$ is called an {\em increasing point sequence} or simply {\em increasing sequence}.
If $\cS$ is a set of points then an $\cS$-increasing point sequence is an increasing sequence
whose points all belong to $\cS$.  

A {\em box chain} $\vec{\cT}$ is a box sequence $\vec{\cT}=(\cT_1,\ldots,\cT_k)$
satisfying $\cT_1 \triangleleft \cdots \triangleleft \cT_k$. Refer to \Fig{boxchain} for
the following definitions.
\begin{itemize}
\item
There is a one-to-one correspondence between increasing point sequences 
and box chains which maps the increasing sequence $\cP$ with points 
$P_0 \prec P_1 \prec \ldots \prec P_k$  to the box chain
$\vec{\cT}$ given by $\bx(P_0,P_1) \triangleleft \cdots \triangleleft \bx(P_{k-1},P_k)$.
We refer to  $\cP$ as the {\em increasing sequence associated to box chain $\vec{\cT}$}.
If $\vec{\cT}$ is a box chain, we write $\cP(\vec{\cT})$ for the associated  increasing point sequence
and $\cP^{\circ}(\vec{\cT})$ for the {\em interior increasing sequence associated to $\vec{\cT}$},
which excludes the first and last points $P_0$ and $P_k$.

\item
The box $\bx[\vec{\cT}]$ spanned by $\vec{\cT}$ is the smallest
box containing $\vec{\cT}$.  If $P_0,\cdots,P_k$ is the increasing point sequence associated
to $\vec{\cT}$ then $\bx[\Vec{\cT}]=\bx(P_0,P_k)$.
\item If $\vec{\cT}$ spans box $\cB$ then for each $x \in X(\cB)$ there is a unique box
$\cT \in \vec{\cT}$ such that $x \in X(\cT)$.  This box is denoted $\vec{\cT}[x]$
\item
Given a strip decomposition $\vec{\cS}$ of $\cB$, a box chain
$\vec{\cT}$ is {\em compatible with $\vec{\cS}$} if $\vec{\cS}$ is composed of strips
$\strip{\cT}{\cB}$  for $\cT \in \vec{\cT}$.  In this case, the box spanned by $\vec{\cT}$
has the form $X(\cB) \times Y'$ where $Y' \subseteq Y(\cB)$.
\item  If $\vec{\cT}$ is a box chain spanning box $\cB$,
then $\vec{\cT}$ naturally defines a partition of $\cB$ into three
sets: $\vec{\cT}^{\cup}$, 
$\vec{\cT}^{NW}$ (for northwest) and $\vec{\cT}^{SE}$ (for southeast),
where:
\begin{itemize}
\item
$\vec{\cT}^{\cup}=\bigcup_{\cT \in \vec{\cT}} \cT$.  Equivalently, $\vec{\cT}^{\cup}=\cB \cap  \bigcap_{P \in \cP(\vec{\cT})} (P^{SW} \cup P^{NE})$.
\item
$\vec{\cT}^{NW}=\cB \cap \bigcup_{P \in \cP(\vec{\cT})}P^{NW}$.
\item
$\vec{\cT}^{SE}=\cB \cap \bigcup_{P \in \cP(\vec{\cT})}P^{SE}$.
\end{itemize}
\end{itemize}
\end{description}

\subsection{Increasing sequences and the functions \lis{} and \loss}

A set $X \subseteq (0,n]$ of indices is said to be {\em $F$-increasing} if the set
$F(X)=\{F(x):x \in X\}$ is an increasing point sequence.
For a box $\cB$ we say that
$X$ is $(F,\cB)$-increasing if it is $F$-increasing and $F(X) \subseteq \cB$. 
If $\vec{\cB}$ is a box chain we say that $X$ is
{\em $(F,\vec{\cB})$-increasing} if it is $F$-increasing  
and $F(X) \subseteq \vec{\cB}^{\cup}$.

\begin{itemize}
\item $\lis(\cB)=\lis_f(\cB)$ is the size of a 
longest increasing (point) sequence (LIS) contained in $\cB$,
which is also the size of the largest $(F,\cB)$-increasing set.
\item
$\loss(\cB)=|\cB \cap \cF|-\lis(\cB)$, i.e. $\loss(\cB)$ is the smallest number of $\cF$-points in $\cB$ that must
be deleted so that the remaining points of $\cB \cap \cF$ form an increasing sequence.
\end{itemize}

\subsection{The LIS approximation problem}

We develop an algorithm $\approxlis$ that takes as input a function $f$ and  box $\cB$  
and outputs an approximation $\approxlis(\cB)$ to $\lis(\cB)$.
The required quality of approximation is specified by input parameters $\taupar$ and $\deltapar$.
The algorithm is recursive and calls itself with different choices of these input parameters.
To prevent confusion, the symbols $\taupar$ and $\deltapar$ denote the initial setting,
while $\tau$ and $\delta$ are used to generically refer to these parameters in the algorithm.

In our analysis we will require a carefully chosen measure of the quality of
the estimate $\approxlis(\cB)$.  
For $\taupar,\deltapar \in (0,1)$ we say that $\approxlis$ 
is a {\em $(\taupar,\deltapar)$-approximation to $\lis$ on box $\cB$} provided that:

\[
|\approxlis(\cB)-\lis(\cB)| \leq \taupar \loss(\cB) + \deltapar \width(\cB).
\] 

A few remarks:

\begin{itemize}
\item A $(0,\deltapar)$-approximation is an 
{\em additive $\deltapar \width(\cB)$-approximation} 
to $\lis(\cB)$.
Since $\loss(\cB) \leq \width(\cB)$ a
$(\taupar,\deltapar)$-approximation is also an additive 
$(\taupar+\deltapar)$-approximation. 
\item Our initial goal is to get a good additive approximation, so the reader may wonder why we introduce
the parameter $\taupar$.  Separating the error into these two parts is important for the analysis of our algorithm.
Our algorithm
is recursive. The value of $\tau$ in the base algorithm is very large (essentially infinite), but we have the freedom
to choose $\delta$ to be very 
small.  Each level of recursion shrinks $\tau$, at the cost
of making $\delta$ larger.  By applying enough
recursive levels, we can make the final $\tau$ less than the desired bound of $\taupar$. 
By starting with a small initial $\delta$, we can keep the final $\delta$ at most $\deltapar$. This ensures 
the final algorithm is a $\taupar + \deltapar$-additive approximation.
\item We refer to the quantity $\tau \loss(\cB)$ as the {\em primary error} and
to $\delta \width(\cB)$ as the {\em secondary error}.
\end{itemize}

\section{The main theorems}
\label{sec:main}

We present two polylogarithmic time approximation algorithms for LIS, which we refer to as  the {\em basic algorithm}
and
the {\em improved algorithm}. The basic algorithm is somewhat simpler (though still fairly involved) while the improved algorithm
enhances the basic algorithm to give significantly better running time.  For the running time of the basic algorithm, 
the exponent of $\log n$ is $\Theta(1/\taupar)$.  For the improved version, the exponent of $\log n$ is a constant
independent of the error parameters $\taupar$ and $\deltapar$. 

\begin{theorem}
\label{thm:alg1}  
 There is a randomized algorithm $\basicmain$ that:
\begin{itemize}
\item  takes as input an integer $n$, an array $f$ of length $n$ and an error parameter $\taupar \in (0,1)$,
\item  runs in time $(\log n)^{O(1/\taupar)}$, and
\item outputs a value that, with probability at least $1-n^{-\Omega(\log n)}$, is a
{\em $(\taupar,\frac{5}{\taupar\log n})$}-approximation to $\lis(\cB)$. 
\end{itemize}
%
\end{theorem}

Our improved algorithm gives:

\begin{theorem}
\label{thm:alg2}   
 There is a randomized algorithm $\improvedmain$ that:
\begin{itemize}
\item  takes as input an integer $n$, an array $f$ of length $n$ and  error parameters $\taupar, \deltapar \in (0,1)$,
\item  runs in time $(1/\deltapar\taupar)^{O(1/\taupar)}(\log n)^c$, and 
\item outputs a value that, with probability at least $3/4$, is a
$(\taupar, \deltapar)$-approximation to $\lis_f$.
\end{itemize}
\end{theorem}

In the second theorem the probability of error is 1/4, as compared to $n^{-\Omega(\log n)}$ in the first.    It just happens that the analysis
of the first algorithm  gives a better error probability.  As we noted in the introduction, this difference
is not significant:
we can always reduce the error probability of the second algorithm to any desired $\errorprob>0$ by the standard trick of doing $O(\log(1/\errorprob)$
independent trials of the algorithm, and outputing the median of the trials.

We deduce \Thm{main} and \Thm{dist} from  \Thm{alg2}.

\Thm{dist} requires a few calculations.

\begin{proof}[of \Thm{main}]
A $(\taupar,\deltapar)$-approximation is also an additive $\taupar+\deltapar$ approximation.  Given a desired additive error
$\delta$ in \Thm{main}, we set $\deltapar=\taupar=\delta/2$ and run $\improvedmain$. \Thm{alg2}
gives us the desired guarantee.
\end{proof}

\begin{proof}[of \Thm{dist}]
For convenience, we will assume that $\improvedmain$ has an error of $n^{-\Omega(\log n)}$.
(Since we will make at most $\poly(n)$ runs of $\improvedmain$ and can union bound, 
we henceforth assume no error.)
Suppose we run $\improvedmain$ with parameters $\taupar,\deltapar$ and $a$ is the estimate.
Then, $|n-a - \loss_f| \leq \taupar \loss_f + \deltapar n$.
We divide by $n$, note that $\eps_f = \loss_f/n$ and denote $b = 1 - a/n$.
Hence, $b \in [\eps_f(1-\taupar) - \deltapar, \eps_f(1+\taupar) + \deltapar]$.
Rearranging, $\eps_f \in [(b-\deltapar)/(1+\taupar),(b+\deltapar)/(1-\taupar)]$.
For convenience, we denote this interval by $[b_1, b_2]$.
%
%

Our aim is to choose $\taupar$ and $\deltapar$ such that $b_2/b_1 \leq (1+\tau)$.
Suppose we set $\taupar = q\cdot\tau$ and $\deltapar \leq q\cdot\taupar\eps_f$,
for some sufficiently small absolute constant $q$. Then,
\begin{eqnarray*}
	\frac{(b+\deltapar)/(1-\taupar)}{(b-\deltapar)/(1+\taupar)} 
	= \frac{(b+\deltapar)(1+\taupar)}{(b-\deltapar)(1-\taupar)} 
	& \leq & \frac{(\eps_f(1+\taupar) + 2q\cdot\taupar\eps_f)(1+\taupar)}{(\eps_f(1-\taupar)-2q\cdot\taupar\eps_f)(1-\taupar)} \\
	& \leq & \frac{(1+q(1+2q)\tau)(1+q\tau)}{(1-q(1-2q)\tau)(1-q\tau)} \leq 1+\tau.
\end{eqnarray*}
%
%
%
%
%

But the value of $\eps_f$ is not known in advance. We fix $\taupar = q\cdot\tau$
and run $\improvedmain$ iteratively 
where the value of $\deltapar$ during the $j$th run is $\taupar/2^j$. 
%
The algorithm will terminate before $\deltapar \leq q\cdot\taupar\eps_f/2$.
The running time is dominated by that of the last iteration,
which is $(1/\tau \eps_f)^{O(1/\tau)}(\log n)^c$.
%
%
%
%
%
%
%
%
%
\end{proof}

\section{Algorithmic and analytical  building blocks}
\label{sec:blocks}
We present some procedures  used in  our algorithms, and some analytic tools we'll need.
First we establish some conventions and review  basic tail bounds that will be useful in analyzing the use of randomness in our algorithms
We define the notion of a good {\em splitter} of a box and present a subroutine for finding
splitters.  We present the important {\em dichotomy lemma} that roughly says: if there are few good splitters in a box $\cB$
then $\loss(\cB)$ must be a non-trivial fraction of $|\cB \cap \cF|$. 
We present a simple subroutine that given a box constructs a grid inside it that is suitably representative of the box. 
\subsection{Conventions for random bits}
\label{subsec:random}

Random sampling is needed  in the following procedures:

\begin{itemize}
\item
The procedures $\findsplitter$, $\buildgrid$ which are presented later in this section. 
We refer to the random bits used
in these procedures as {\em secondary random bits}. 
\item 
The main procedure $\approxlis$ which is presented in \Sec{first algorithm}.  
The random bits used in this procedure are called
the {\em primary random bits}.
\end{itemize}

All procedures depend on the function $f$. We treat the function $f$ as fixed, so the dependence on $f$ is implicit.  
The other arguments to these procedures are boxes, indices, and auxiliary precision and error parameters.
We can enumerate the possible arguments to each of these procedures.
Indices have $n$ possible values, and boxes always have their corners on the universal grid, so there are at most
$n^4$ of them.  The auxiliary parameters will be from a restricted set of size at most $((\log n)/\tau)^{O(1)}$ where
$\tau$ is the primary error parameter.
 
For purposes of analysis, we imagine that  before running the algortihm we pregenerate
all of the random bits needed for each procedure and each possible set of arguments for the procedure, 
Once all of these bits are (conceptually) fixed, the running of the algorithm
is deterministic.   This viewpoint has three advantages for us.

\begin{enumerate}
\item In our overall algorithm, the same procedure may be called multiple times with the same arguments. For each
such call, we use the \emph{same sequence of random bits}, as generated at the outset of the algorithm.
Therefore, every such duplicate
call produces the same output.
\item The randomized procedure $\classify$ (defined in Section) will itself be run on a few randomly chosen indices.
In the analysis, we need to reason about the output of $\classify$
on all inputs, not just the ones evaluated by the algorithm. By generating random bits for every possible call
to a procedure, the output of $\classify$ can be seen as fixed on all possible inputs.
If the random bits were not fixed in advance, then the output of $\classify$ on a new input 
is a random variable.  Both points of view lead to the same mathematical conclusions,
but viewing the randomness as fixed from the  beginning simplifies the analysis.
\item In the analysis, we identify some useful (deterministic) assumptions of the ensemble of all of the random bits.   
The main analysis
shows that if the ensemble of random bits satisfies these assumptions, 
then the output of our algorithm is guaranteed to have the right properties. This involves no probabilistic reasoning.  
Probabilistic arguments are only required to show that these two assumptions 
hold with very high probability over the choice of the ensemble of random bits.  
\end{enumerate}

The pre-generation of random bits is, of course, simply an analytical
device.  Since we want our algorithm to be fast, we cannot afford
to generate all  of bits in advance. So 
we generate random bits only as we need them, and only for those choices of parameters to procedures that actually arise.  Once we evaluate
a procedure with a given set of arguments, we record the input parameters and output. If the same procedure is called again
with the same arguments, we return the same value.    It is evident that this online approach to generating randomness produces the same
distribution over executions and outputs as the inefficient offline approach. Hence, we execute the online algorithm,
but use the offline viewpoint for analysis.

\subsection{Tail bounds for sums of random variables}
\label{subsec:tail}

We recall a version of the Chernoff-Hoeffding bound for sums
of random variables:

\begin{prop}
\label{prop:hoeff1}
\cite{Ho63}
Let $X_1,\ldots,X_n$ be independent random variables with $X_i \in [a_i,b_i]$
and $\mu_i=\mathbb{E}[X_i]$ and $\mu=\sum_i \mu_i$.  Then for any $T>0$:
\begin{eqnarray*}
\pr[\sum_i X_i - \mu \geq T] & \leq & e^{-2T^2/\sum_i (b_i-a_i)^2}\\
\pr[\sum_i X_i - \mu \leq -T] & \leq & e^{-2T^2/\sum_i (b_i-a_i)^2}.
\end{eqnarray*}

\end{prop}

A standard application of is to bound
the probability of error when estimating the size of a subpopulation within
a given population. 
For a finite set $X$,  a {\bf random sample of 
size $m$ from a set $X$} means a sequence
$x_1,\ldots,x_m$ of elements each drawn uniformly and independently  from $X$. 

\begin{prop}
\label{prop:hoeff2}
Let $X$ be an index interval and $\gamma \in [0,1]$ and
$s \in \mathbb{N}^+$.
Let $A \subseteq X$ be fixed.
For a random sample  $x_1,\ldots,x_s$ from $X$,
let $r$ denote the fraction of points 
that belong to $A$.
Then
$\Pr[|r - \frac{|A|}{|X|}| \geq \gamma] \leq 2e^{-2\gamma^2s}$.
\end{prop}

We will also need the following upper tail bound:

\begin{prop}
\label{prop:cher} (Theorem 1, Eq. (1.8) in~\cite{DuPa})
Let $X = \sum_i X_i$, where $X_i$'s are independently distributed in $[0,1]$. 
If $T > 2e\EX[X]$, then $\pr[X > T] \leq 2^{-T}$.
\end{prop}

\subsection{Splitters and the subroutine $\findsplitter$}
\label{subsec:splitters}

A basic operation in our algorithm is to take a box $\cT$ and choose an index $s \in F^{-1}(\cT)$, that is used 
to ``split'' the box $\cT$ into the two boxes $\bx(P_{BL}(\cT),F(s))$ and $\bx(F(s),P_{TR}(\cT))$.
This gives a box chain of size two spanning $\cT$.  
For this reason, the index $s \in F^{-1}(\cT)$ is called a {\em splitter} for $\cT$; we also refer to
the point $\langle s,f(s) \rangle$ as a splitter.   Note that $s>x_L(\cT)$ for all $s \in F^{-1}(\cT)$ and therefore
 $\bx(P_{BL}(\cT),F(s))$ is always a nonempty box.  
It could happen that $s = x_R(\cT)$, in which case $\bx(F(s),P_{TR}(\cT))$ is a trivially empty box.
In this case we say that the splitter
 is {\em degenerate};  any other splitter is said to be  {\em nondegenerate}.

We give a subroutine that, given a pair of boxes $\cT \subseteq \cB$,
looks for a ``useful" splitter $s$ for $\cT$ with $F(s) \in \cT$.
We require a splitter to be \emph{balanced} and \emph{safe} (in a precise sense to be specified).  Roughly,
a balanced splitter is not too close to either the  left or right edge of $\cT$ (and in particular is not degenerate).
A splitter is safe if most $\cF$-points of $ \strip{\cT}{\cB}$ are comparable to $F(s)$. 

The definitions in this section are inspired by the classic ideas of \emph{inversion counting}~\cite{EKK+00,ACCL1}
for estimating the distance to monotonicity.

We begin by formally defining balanced indices.

\begin{definition} \label{def:balanced} For a box $\cT$ and $\rho \in [0,1]$, an index $x$
is said to be {\em $\rho$-balanced in $\cT$}
if $x-x_L(\cT)$ and $x_R(\cT)-x$
are both at least  $\rho \width(\cT)$, and is \emph{$\rho$-unbalanced} otherwise.
\end{definition}
It follows that the number of $\rho$-balanced indices is at least $\lfloor (1-2\rho)\width(\cT) \rfloor$ and the number of $\rho$-unbalanced indices
is at most $2\rho \width(\cT) + 1$.  Excluding the degerate splitter $x_R(\cT)$ we get:
\begin{prop}
\label{prop:unbalanced} For any $\rho>0$,
the number of nondegenerate $\rho$-unbalanced splitters of $\cT$ is at most $2\rho \width(\cT)$.
\end{prop}

\medskip

The definition of safe has several parameters and ``moving parts'' and some preliminary discussion may be helpful.   The definition
involves three parameters: a box $\cR$, a real number $\mu \in (0,1)$ and a positive real number $L$.  
The notion of safety is expressed by saying that $s$ is a {\em $(\mu,L)$-safe splitter for $\cR$}.   The requirements are 
that $s$ pass a collection of tests, one for each substrip $\cS$ of $\cR$ that is
adjacent to $s$ in the following sense: either the maximum index in $X(\cS)$ is $s-1$ or the minimum index is $s+1$.
The requirement corresponding to substrip $\cS$ of $\cR$ is that the number 
of $\cF$-points inside of $\cS$ that are in violation with $F(s)$ should be ``not too large'' compared with the
total number of $\cF$-points inside of $\cS$: specifically it should be at most $L+\mu|\cF\cap \cS|$.

\begin{definition} \label{def:safe} This is a series of definitions used to formalize safeness of splitters.
 \begin{itemize}
\item $viol(s,\cS)$: This is the number of points $P \in \cS \cap \cF$ that are in violation with
$F(s)$.
\item $Z(s,\cS) := viol(s,\cS)-\mu|\cF \cap \cS|$.  
\item $Z_x(s,\cS)$: This is defined for index $x \in X(\cS)$. Suppose $F(x) \in \cS$. If $F(x) \sim F(s)$, $Z_x(s,\cS)=-\mu$
and if $F(x) \not\sim F(s)$, $Z_x(s,\cS)=1-\mu$.  If $F(x) \not \in \cS$, $Z_x(s,\cS)$ is 0.
Note that $Z(s,\cS) = \sum_{x \in X(\cS)} Z_x(s,\cS)$.
\item  $(\mu,L)$-accepting: A strip $\cS$ is {\em $(\mu,L)$-accepting for $s$} if $Z(s,\cS)\leq L$ and is
{\em $(\mu,L)$-rejecting for $s$} otherwise.
\item Adjacent strips: A strip $\cS$ is said to be {\em adjacent to $s$} if either $x_L(\cS)=s$ (i.e. the lowest index  in  $X(\cS)$ is $s+1$)
or $x_R(\cS)=s-1$ (the highest index in $X(\cS)$ is $s-1$).   
\item $(\mu,L)$-unsafe: We say that $s$  is {\em $(\mu,L)$-unsafe  for $\cR$}  if there is
some $\cR$-strip $\cS$ that is  adjacent to $s$  and $(\mu,L)$-rejecting for
$s$.  
\item $(\mu,L)$-safe:We say that $s$ is {\em $(\mu,L)$-safe for $\cR$} if every $\cR$-strip $\cS$ that is adjacent to $s$
is $(\mu,L)$-accepting for $s$. We remark that the $(\mu,L)$-safe and $(\mu.L)$-unsafe indices for $\cR$ partition the set of splitters of $\cR$.
\end{itemize}
\end{definition}

With this preamble, we can state the main definition of \emph{adequate splitters}.

\begin{definition} Let $\cT \subseteq \cB$ be a pair of boxes and $\mu,\rho \in (0,1)$ and $L>0$ .
A splitter $s \in X(\cT)$ is {\em $(\mu,L,\rho)$-adequate for $\cT,\cB$} if:
\begin{itemize}
\item $F(s) \in \cT$
\item $s$ is $(\mu,L)$-safe for $\strip{\cT}{\cB}$
\item $s$ is $\rho$-balanced in $\cT$.
\end{itemize}
\end{definition}

We describe a procedure $\findsplitter$ that takes as input
a pair of boxes $\cT \subseteq \cB$  and parameters $\rho,\mu$ and $L \geq 1$ and
searches for such a splitter. The procedure
returns a pair $(\splitterfound, s)$
where $\splitterfound$ is set to {\tt TRUE} or {\tt FALSE}  
and $s \in X(\cT)$.

We say that an execution
of  $\findsplitter$ on input $(\cT,\cB,\mu,L,\rho)$ is 
{\em reliable} if the following holds:
\begin{itemize}
\item If $\cT$ has at least $\rho \width(\cB)$ splitters that are $(\mu,L,\rho)$-adequate  for $\cT,\cB$ then $\splitterfound$
is set to {\tt TRUE}.
\item If $\splitterfound={\tt TRUE}$ then $s$ is a $(\mu,2L,\rho)$-adequate splitter for $\cT,\cB$.  (Note that here
we relax the parameter $L$ to  $2L$.)
\end{itemize}
An execution where either of these two conditions fails is called {\em unreliable}. 
The procedure $\findsplitter$ is designed so that each execution
is reliable with high probability. The construction  is  straightforward application of
random sampling, though the details are a bit technical. We first state the claims related to $\findsplitter$.
On first reading the reader may wish to skip the details and simply take note of these claims.

\begin{prop}
\label{prop:splitter}
For any input $\cT,\cB,\mu,L,\rho$, an execution of $\findsplitter(\cT,\cB,\mu,L,\rho)$ is reliable with probability at least $1-n^{-\Omega(\log n)}$
and has running time $O((\frac{\width(\cT)}{L})^3(\log n)^4/\rho)$.
\end{prop}

For further reference. we note the following direct corollary.

\begin{corollary}
\label{cor:fails}
Let $\mu,\rho \in (0,1)$ and $L>0$.
Let $\cT$ be a subbox of $\cB$. 
Assume that a reliable run of $\findsplitter(\cT,\cB,\mu,L,\rho)$ fails to find a splitter.
Then the number of nondegenerate $(\mu,L)$-safe splitters for $\strip{\cT}{\cB}$ is at most $3\rho\width(\cT)$. 
\end{corollary}

\begin{proof} (of \Cor{fails} assuming  \Prop{splitter}) 
Since the run is reliable and no splitter was found, there are at most $\rho \width(\cT)$
splitters that are $\rho$-balanced and $(\mu,L)$-safe for $\cT,\cB$. By  \Prop{unbalanced},
the number of nondegenerate $\rho$-unbalanced splitters is at most $2\rho \width(\cT)$.  Summing, the total
number of nondegenerate $(\mu,L)$-safe splitters is at most $3\rho \width(\cT)$.
\end{proof}

The procedure $\findsplitter$ uses the following auxiliary procedures.
 
\begin{itemize}
\item {\bf $approxZ$:} The input is an index $s$, box $\cC$, and integer $m \leq \width(\cC)$
and the output is an approximation of $Z(s,\cC)$. If $m=\width(\cC)$ then this returns the exact
value $Z(s,\cC)=\sum_{x \in X(\cC)}Z_x(s,\cC)$.  Otherwise, this is obtained by taking a random sample $M$ of size $m$ from $X(\cC)$
and outputting $\frac{\width(\cC)}{m}\sum_{x \in M} Z_x(s,\cC)$.
\item {\bf $Test\_Safe$:} Input is an index $s$, box $\cR$ and parameter $L$. Output is either {\em accept} or {\em reject}.
For each strip
$\cS$ of $\cR$ that is adjacent to $s$ and has width a multiple of $\lceil L/3 \rceil$, evaluate $approxZ(s,\cS,m)$
with $m = \min(\width(\cS),10((\log n)\width(\cS)/L)^2)$.  $Test\_Safe(s,\cR,L)$ \emph{accepts} if every evaluation of $approxZ$ returns a value
less than $4L/3$ and otherwise {\em rejects}.
\item  $\findsplitter(\cT,\cB,\mu,L,\rho)$ takes a random sample $R$ of size $10(\log n)^2/\rho$ 
from  the interval $X_{\rho}(\cT)$ consisting of all indices that are $\rho$-balanced in $\cT$.
For each $s \in R$, we first check if both $F(s) \in \cT$ and $Test\_Safe(s,\strip{\cT}{\cB},L)$ accepts.  If some $s \in R$
is accepted, then
$\findsplitter$ returns $\splitterfound={\tt TRUE}$ and $splitter$ is set to $s$.  If all are rejected,
then $\splitterfound={\tt FALSE}$.
\end{itemize}

\begin{proof} (of \Prop{splitter})
The running time of  $\findsplitter(\cT,\cB,\mu,L,\rho)$ is at most $O((\frac{\width(\cT)}{L})^3(\log n)^4/\rho)$.
There are $O((\log n)^2/\rho)$ invocations to $Test\_Safe$. Note that $\width(\cT|\cB) = \width(\cT)$.
Each call to $Test\_Safe$ runs $approxZ$ with $m = O(((\log n) \width(\cS)/L)^2)$ on at most $3\width(\cT)/L$ substrips.
Multiplying all of these numbers together gives the final bound.
%
(Later, we select $\rho = 1/(\log n)^{O(1)}$ and $L = \width(\cT)/(\log n)^{O(1)}$, to bound the running time by $(\log n)^{O(1)}$.)

The random seed of $\findsplitter$ is used to specify the sample $R$  as well as the samples $M$ for each
call to $approxZ$.  We say that a particular value of the
random seed is {\em sound for input $(\cT,\cB,\mu,L,\rho)$} provided that:
\begin{itemize}
\item If the number of  $(\mu,L,\rho)$-adequate splitters for $\cT,\cB$ is at least $\rho \width(\cT)$, 
then at least one such index is selected for $R$.
\item For every call to $approxZ(s,\cS,m)$, the estimate returned is within $L/3$ of $Z(s,\cS)$.
\end{itemize}

We say that the random seed is {\em sound} if it is sound for all possible inputs $(\cT,\cB,L,\rho)$.  We now
show (1) the probability that the random seed is not sound is $n^{-\Omega(\log(n))}$, and (2)  if the random
seed is sound then for any input $(\cT,\cB,L,\rho)$ to $\findsplitter$, 
the execution $\findsplitter(\cT,\cB,L,\rho)$ is reliable.  This will complete the proof of the  proposition.

First consider (1).   First we fix an input $(\cT,\cB,L,\rho)$ to $\findsplitter$ and upper bound the probability that the random seed is not
sound for that input.  Consider the probability that the first condition of soundness is violated.
 by the hypothesis of (1), a randomly selected index from $X(\cT)$ is $(\mu,L,\rho)$-adequate for $\cT,\cB$
with probability at least $\rho$.  The probability that $R$ contains no such index is at most $(1-\rho)^{10(\log n)^2/\rho}=n^{-\Omega(\log n)}$.
Next consider the probability that the second condition of soundess is violated.
 Consider some call to $approxZ(s,\cS,m)$.
The output is $X = X_1 + X_2 + \ldots + X_m$, where each random variable $X_i = (\width(\cS)/m)Z_{x_i}(s,\cS)$ for the  $i$th sample $x_i$.
Note that $\EX[Z_{x_i}] = Z(s,\cS)/\width(\cS)$, so $\EX[X] = Z(s,\cS)$. Also, $X_i \in [-\width(\cS)/m,0]$ (since $\mu \in (0,1)$).
By \Prop{hoeff1}, since $m = 10((\log n)\width(\cS)/L)^2$, the probability that
the estimate has error more than $L/3$ is $\exp(-\Omega(L^2m/\width(\cS)^2)) = n^{-\Omega(\log n)}$. Thus the overall
probability that the seed is not sound for $(\cT,\cB,L,\rho)$ is at most $n^{-\Omega(\log n)}$.

Note that the number of possible settings of the parameters for $approxZ(s,\cS,m)$ is at most $n^3$ (at most $n$ choices each for $s$ and $n$
and the boundary of the $\cB$-strip $\cC$), so we can sum the probability of errors over all possible calls and conclude that
the probability that a random seed is not sound is $n^{-\Omega(\log n)}$.

%
%

Next we show (2).  Assume the seed is sound and fix the input $(\cT,\cB,L,\rho)$ to $\findsplitter$.  For the first condition of reliability,
 suppose that there  are at least $\rho \width(\cT)$
splitters that are $(\mu,L,\rho)$-adequate for $\cT,\cB$.  By the soundness of the seed, at least one such splitter $s$ is chosen to be in $R$. 
 When $Test\_Safe(s,\strip{\cT}{\cB},L)$  is performed, 
for each examined strip $\cS$, $approxZ(s,\cS,m)$ will be at most $4L/3$.  This is because $s$ is $(\mu,L)$-safe for $\strip{\cT}{\cB}$
and thus $Z(s,\cS) \leq L$ and the second condition of soundness guarantees that $approxZ(s,\cS,m) \leq Z(s,\cS)+L/3$.
So $s$ will be accepted, and thus $\splitterfound$ will be set to {\tt TRUE}. 

For the second condition defining reliable, it suffices to show that
no $(\mu,2L)$-unsafe splitter is accepted.  Suppose $s$ is $(\mu,2L)$-unsafe.  Then there is a strip $\cS$ adjacent to $s$ such that $Z(s,\cS) > 2L$.  If $\cS'$
is the strip adjacent to $s$ whose width is the largest multiple of $\lceil L/3 \rceil$ below $\width(\cS)$
then $Z(s,\cS') \geq Z(s,\cS)-L/3 > 5L/3$. By the soundness of the seed, $Test\_Safe(s,\strip{\cT}{\cB},L)$ evaluates
$approxZ(s,\strip{\cT}{\cB},m)$, which returns a value greater than $4L/3$. Hence, $Test\_Safe(s,\strip{\cT}{\cB},L)$ rejects.
\end{proof}

\subsection{The dichotomy lemma}
\label{subsec:dichotomy}

In this subsection, we prove a key technical lemma. The lemma expresses the 
dichotomy discussed in the overview of the algorithm in the introduction.
If a given box has few good splitters, then any increasing sequence in the box must
miss a significant fraction of $\cF$-points in the box. This lemma can be seen as a generalization
(and a different viewpoint) of Lemma 2.3 in~\cite{ACCL1}. That result was a key part
of the $2$-approximation for the distance to monotonicity, and related the distance to
(roughly speaking) the number of unsafe points.

\begin{figure*}[tb]
  \centering
  \subfloat[The setting]{\includegraphics[width=0.35\textwidth]{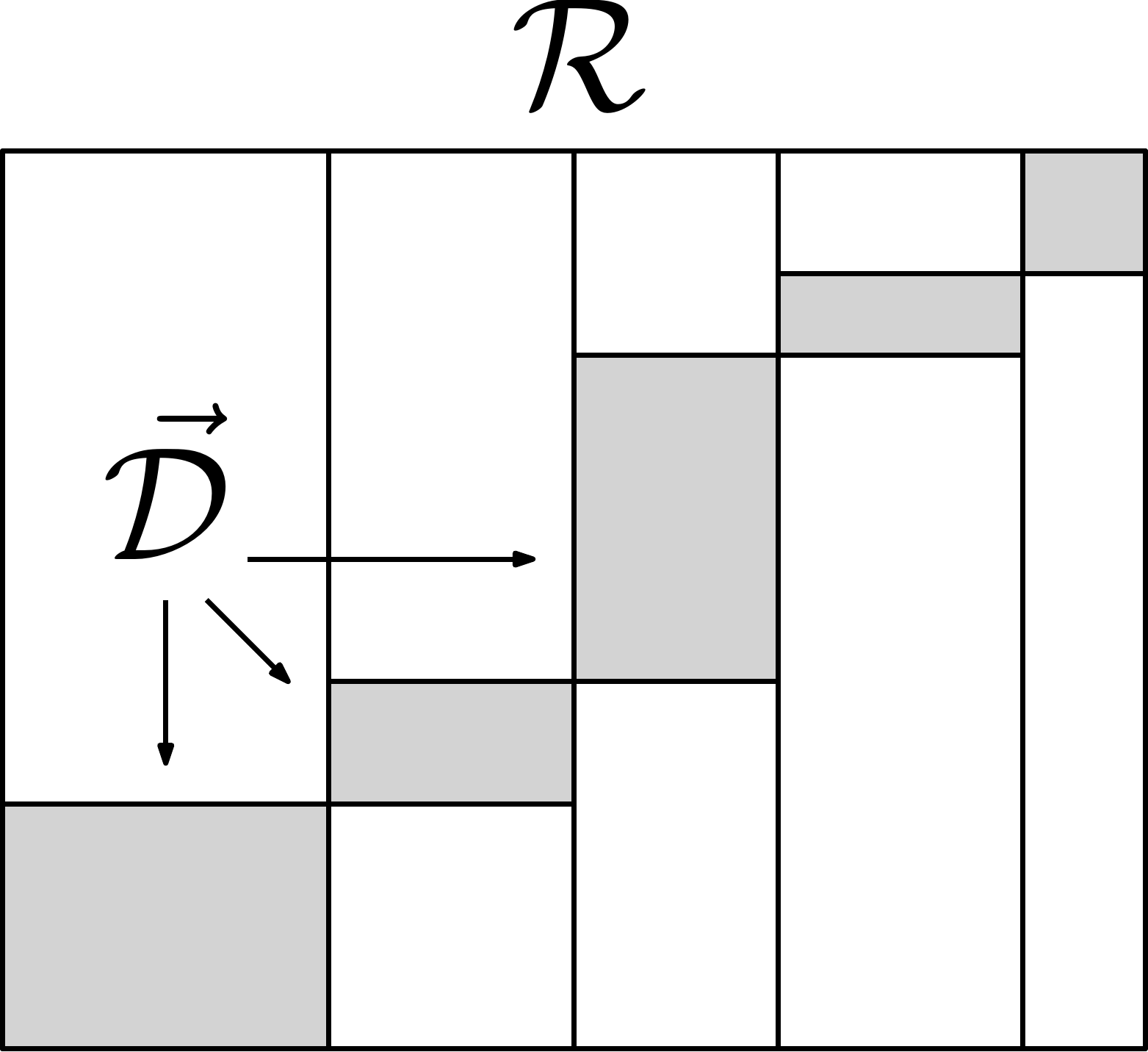} \label{fig:dichot1}} 
  $\qquad \qquad$
  \subfloat[The strategy]{\includegraphics[width=0.35\textwidth]{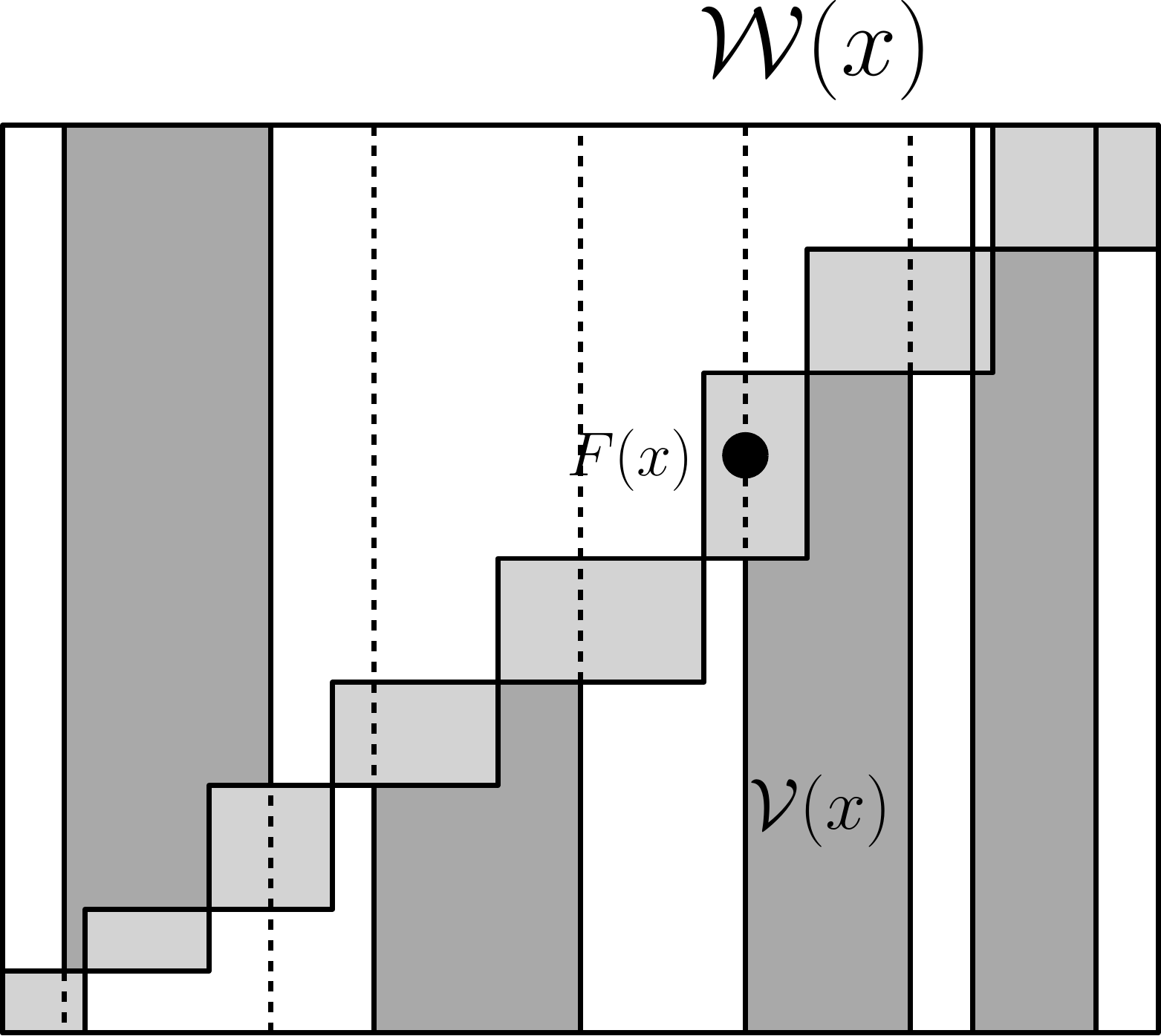} \label{fig:dichot2}}
    \caption{(a) We have a boxchain $\vec{\cD}$ compatible with strip decomposition $\vec{\cS}$, spanning box $\cR$.
    (b) We bound $\chi^{out}$ through the dark gray regions. These disjoint regions are constructed through a set
    of strips. $\cV(x)$ consists of violations in the strip that are not contained in the main box chain.}
\end{figure*}

The set up for the lemma is:
\begin{description}
\item[D1] $\cR$ is a box.
\item[D2]  $\vec{\cS}$ is a strip decomposition of  $\cR$, i.e., a partition of $\cR$ into $\cR$-strips. 
\item[D3] For $x \in X(\cR)$, $s(x)$ is the width of the strip containing $x$.
\item[D4] $\mu \in (0,1)$.
\item[D5]  An index $x \in X(\cR)$ is called {\em safe} or {\em unsafe},
depending on whether it is a $(\mu,s(x))$-safe splitter for $\cR$.
\item[D6] $\vec{\cD}$ is an arbitrary box-chain that is 
compatible with $\vec{\cS}$. (This means that $\vec{\cS}$ consists of the strips $\strip{\cD}{\cR}$ for $\cD \in \vec{\cD}$.
Refer to \Fig{dichot1}.)
\item[D7] An index $x$ such that $F(x) \in \vec{\cD}^{\cup}$ is called a {\em  $\vec{\cD}$-index}. 
\item[D8] $S$ is the set of $\vec{\cD}$-indices $x$ that are safe, and $U$ 
is the set of  $\vec{\cD}$-indices that are unsafe. 
\item[D9] $\chi=|\cR \cap \cF|$ is the number of $\cF$-points in $\cR$, $\chi^{in}$ is the number of such points
that lie inside $\cD^{\cup}$ and $\chi^{out}=\chi-\chi^{in}$ is the number of such points that lie outside of
$\cD^{\cup}$.
\end{description}

\begin{lemma}
\label{lem:dichotomy}
Under hypotheses D1-D9,
we have  $\chi^{out} \geq \frac{\mu}{1-\mu}|U|$  and
$\chi^{in} \leq (1-\mu)\chi + \mu(|S|)$.
\end{lemma}

The point of the lemma is that if there are few safe indices then the fraction of $\cF$-points of $\cR$ that 
lie in $\vec{\cD}$ can't be much larger than $1-\mu$.
The idea of the proof is that each unsafe index in $\vec{\cD}^{\cup} \cap \cF$ can be associated
to a non-trivial set of points in $\cF$ that are outside of $\vec{\cD}$.

\begin{proof} For each $x \in U$, there exists a strip $\cW(x)$ with $x$ on either 
the left or right end such that $\cW(x)$ is $(\mu,s(x))$-rejecting for $x$.  
In \Fig{dichot2}, we show such a point $F(x)$ with $\cW(x)$ to the right.
Let $\cV(x)$ be the set of violations with $F(x)$, contained in $\cW(x)$ but not in $\vec{\cD}^\cup$.
In other words, $\cV(x) = \{F(y) | F(y) \not\sim F(x), F(y) \in \cW(x) \setminus \vec{\cD}^\cup\}$.
In \Fig{dichot2}, the corresponding region in $\cW(x)$ is marked in dark gray. Observe
that it is contained in $\vec{\cD}^{SE}$. We now lower bound $|\cV(x)|$. 

\begin{claim} 
\label{clm:cV}
$|\cV(x)| \geq \frac{\mu}{1-\mu} |\vec{\cD}^{\cup} \cap \cW(x)\cap \cF|$.
\end{claim}

\begin{proof}
The number of violations with $x$ in $\cW(x)$ is at least $\mu|\cF \cap \cW(x)|+s(x)$.
Any violation with $x$ contained in $\vec{\cD}^\cup$ must be contained in $\vec{\cD}[x]$ (the unique box of $\vec{\cD}$ whose index set includes $x$).
This is because $\vec{\cD}$ is a box chain (a look at \Fig{dichot2} should make this clear). 
Hence, the number of such violations is at most $\width(\vec{\cD}[x]) \leq s(x)$.
Combining, $|\cV(x)| \geq \mu|\cF \cap \cW(x)|$. Since $\cV(x)$ is disjoint with $\vec{\cD}^\cup$,
$|\vec{\cD}^\cup \cap \cW(x)\cap \cF| \leq (1-\mu)|\cF \cap \cW(x)|$ and $|\cV(x)| \geq \frac{\mu}{1-\mu} |\vec{\cD}^{\cup} \cap \cW(x) \cap \cF|$.
\end{proof}
%
%
%

Let $L$ (resp. $R$) be the indices $x \in U$ such that $x$ lies to the left (resp. right) of  $\cW(x)$. 
For $x \in R$, $\cV(x) \subseteq \vec{\cD}^{NW}$ (refer to \Fig{dichot2}).
Similarly for $x \in L$, $\cV(x) \subseteq \vec{\cD}^{SE}$.  

We construct  $L' \subseteq L$ and $R' \subseteq R$
such that:

\begin{itemize}
\item The family   $\{\cW(x):x \in L'\}$ of strips 
is pairwise disjoint and the 
family $\{\cW(x):x \in R'\}$ 
is pairwise disjoint.
\item $\bigcup_{z \in L'} \cW(z)$ contains $\{F(z) | z \in L\}$
and $\bigcup_{z \in R'} \cW(z)$ contains $\{F(z) | z \in R\}$. 
\end{itemize}

The sets $L'$ and $R'$ are constructed separately by simple greedy algorithms.  To construct $L'$ (resp. $R'$), start with $L'$ (resp. $R'$) set to
the $\emptyset$ and repeatedly select
the least index $x \in L$ (resp. largest index $x \in R$) for which $F(x)$ is not already covered by $\cW(z)$ for a previously selected $z$.

Let us define $\cW'$ to be the union $\bigcup_{x \in L' \cup R'} \cW(x)$.
We lower bound $\chi^{out}$ by  $|\bigcup_{x \in L' \cup R'} \cV(x)|$ which is equal to $|\bigcup_{x \in L'}|\cV(x)|+|\bigcup_{x \in R'}|$, since  $\cV(x) \subseteq \vec{\cD}^{NW}$ for $x \in R$ and   $\cV(x) \subseteq \vec{\cD}^{SE}$ for $x \in L$. 
Furthermore, since
$\cW(x)$ are disjoint for $x \in R'$ and for $x \in L'$, this is equal to $\sum_{x \in L' \cup R'}|\cV(x)|$.
Using   \Clm{cV} 
and the fact that $U \subseteq \vec{\cD}^\cup \cap \cW'$, we obtain:
\begin{eqnarray*}
\chi^{out} \geq \Big|\bigcup_{x \in L' \cup R'} \cV(x) \Big|& = &
\sum_{x \in L' \cup R'} |\cV(x)| \\
 &  \geq & \sum_{x \in L' \cup R'}\frac{\mu}{1-\mu}|\vec{\cD}^{\cup} \cap \cW(x) \cap \cF|\\
& \geq & \frac{\mu}{1-\mu}|\vec{\cD}^{\cup} \cap \cW'(x)| \geq \frac{\mu}{1-\mu}|U|.
\end{eqnarray*} 
%
For the second conclusion of the lemma, using  $\chi^{in}= |U|+|S|$ we get:

\begin{eqnarray*}
\chi^{in} = \chi - \chi^{out} & \geq & \chi - \frac{\mu}{1-\mu}|U| \\
& = & \chi - \frac{\mu}{1-\mu} (\chi^{in} - |S|).
\end{eqnarray*}
Multiplying both sides by $1-\mu$ and adding $\mu \chi^{in}$ to both sides yields the desired inequality.
\qed\\
\end{proof}

\subsection{Value nets and the subroutine \buildnet}
\label{subsec:nets}

Given a box, we want to select a suitably  representative set of values from the box.
Let us say that a value interval $J$ is {\em $\netprec$-popular for box $\cB$} if there
are at least $\netprec \width(\cB)$ indices $x \in X(\cB)$ such that $f(x) \in J$. 
If $\cB$ is a box and $\netprec \in (0,1)$, 
a {\em $\netprec$-value net} for $\cB$ is a subset $V$ of $Y(\cB)$ such that:
\begin{itemize}
\item $y_T(\cB) \in V$.
\item For all subintervals $J$ of $Y(\cB)$ that are $\netprec$-popular for $\cB$, $V \cap J \neq \emptyset$.
\end{itemize}

Refer to \Fig{valuenet}. The value net contains a value in $J$ whenever the corresponding dark gray region contains at least
$\netprec \width(\cB)$ $\cF$-points.

\begin{figure*}[tb]
  \centering
  \subfloat[Value nets]{\includegraphics[width=0.25\textwidth]{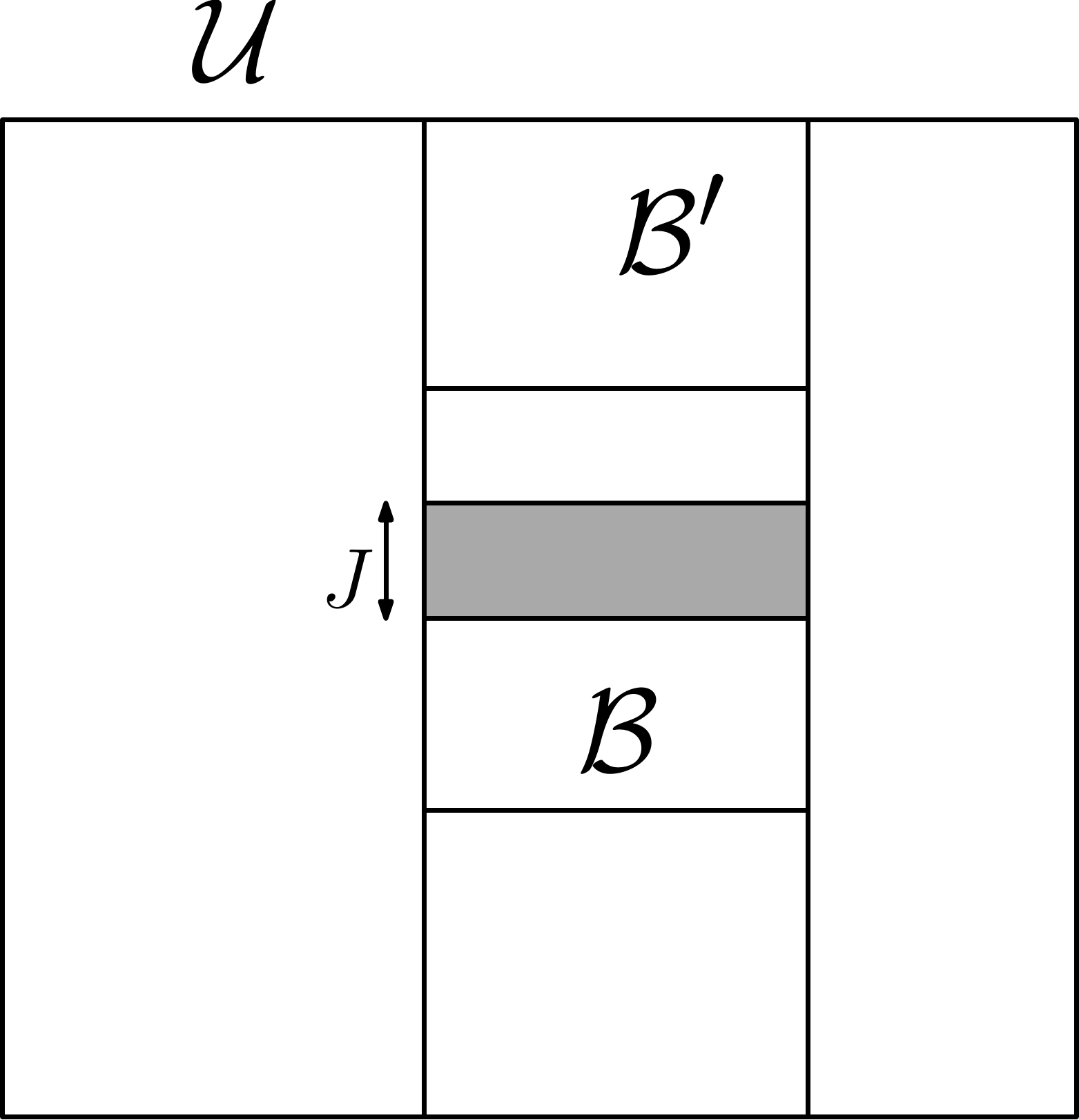} \label{fig:valuenet}} 
  $\qquad \qquad$ 
  \subfloat[The digraph $D(\Gamma)$]{\includegraphics[width=0.45\textwidth]{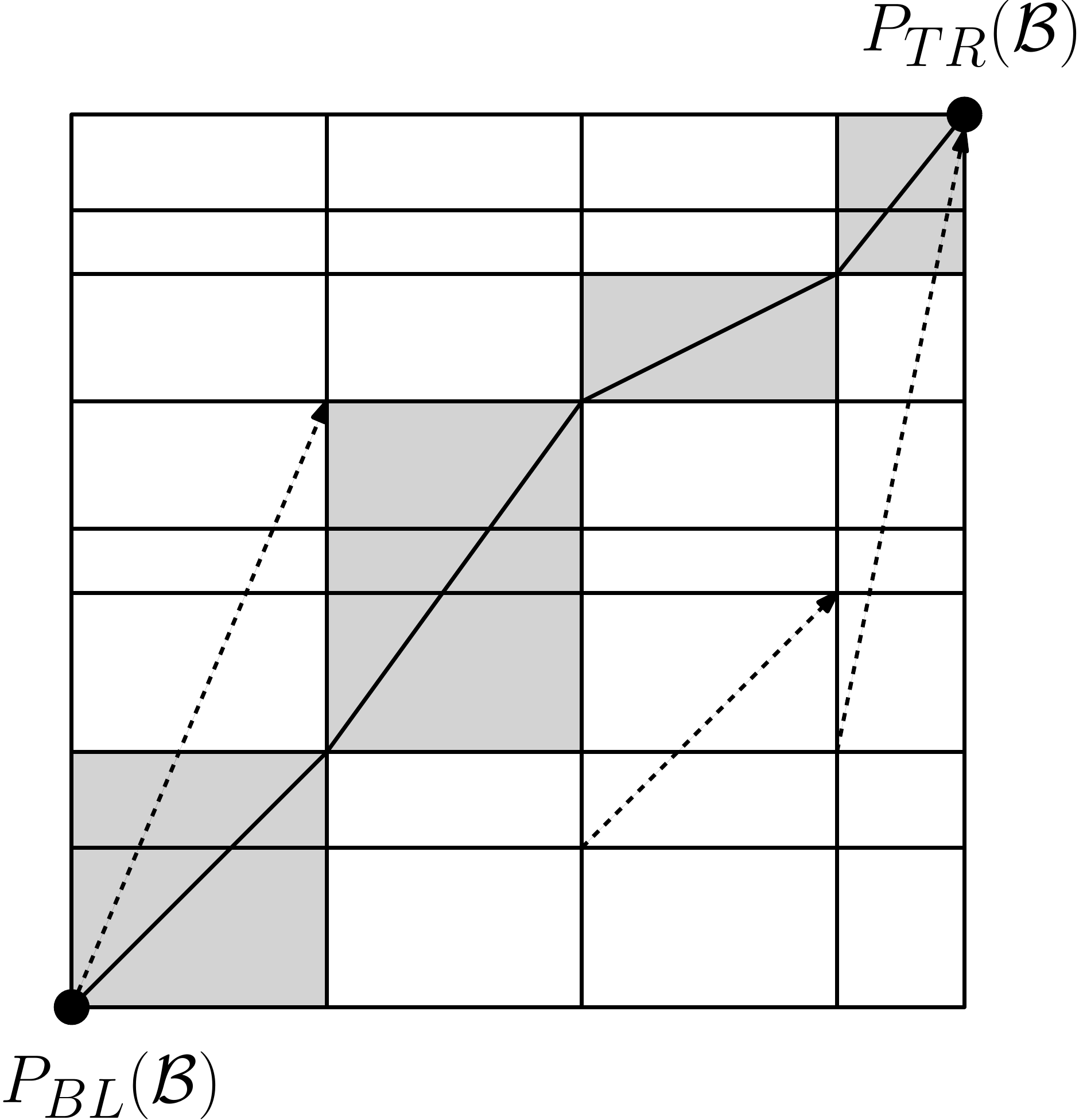} \label{fig:digraph}}
    \caption{(a) If the number of points in the gray region is large enough, then a value
    net must contain a value in $J$.
    (b) All grid points (which are internal to the box) are vertices. The dashed arrows show the 3 different types of arcs in $D(\Gamma)$. The solid line traces a $D(\Gamma)$ path,
    and the corresponding box chain is colored gray.}
\end{figure*}

If we had access to  the set of values $Y(\cB) \cap f(X(\cB))$ in nondecreasing order, then we could
construct a $\netprec$-value net for $\cB$ by taking those values whose position in the order is a multiple of  
$\lceil \netprec \width(\cB) \rceil$.  However, we can only access the values indirectly
by evaluating $f$ at indices in $X(\cB)$ so constructing this order would require evaluating $f$ at every
index in $X(\cB)$, which is too time consuing.

To  construct a  $\netprec$-value net quickly we use random sampling.


%
\begin{prop}
\label{prop:buildnet}
There is a randomized procedure $\buildnet$ that takes
as input a triple $(\cB,\netprec,\neterror)$ where $\cB$ is a box
and $\netprec,\neterror \in (0,1)$, runs in time $(1/\netprec)^{O(1)}\log(1/\neterror)$ and outputs 
a subset $V$ of $Y(\cB)$ of size at most $4\lceil 1/\gridprec \rceil$
such that with probability at least $1-\neterror$, $V$  is a $\netprec$-value net for $\cB$. 
\end{prop}

\begin{proof}
Given a box $\cB$, let $\cB'$ be the strip $\strip{\cB}{\cU}$. (This is the set
of all points whose index belongs to $X(\cB)$.)
It suffices to  construct a $\netprec$-value net $V'$ for $\cB'$.  Once we do this
we take $V=V' \cap Y(\cB) \cup \{y_T(\cB)\}$. This                                                                                                           is a $\netprec$-value net for $\cB$
since every $\netprec$-popular value  interval $J$  for $\cB$ is also
$\netprec$-popular for $\cB'$.

So let us construct a $\netprec$-value net for $\cB'$.
Let $s=4\lceil 1/\gridprec \rceil$.
Let $M=\lceil s\ln(s/2\neterror)\rceil$.  
Let $R$ be a sequence of $sM-1$ uniform random samples from $X(\cB')$
and let $y_1 \leq \cdots \leq y_{sM-1}$ be their $y$-values in sorted order.
Note that all of these are in $Y(\cB')$.
Let $y_0=0$ and $y_{sM}=\maxval$.  Define $V=\{y_{iM}:i \in [1,s-1]\}$.

We analyze the probability that $V$ is not a $\netprec$-value net for $\cB'$.

Let $B=\lceil \gridprec \width(\cB')/2 \rceil$ (which is also equal to $\lceil \gridprec \width(\cB)/2 \rceil$).
Let $x_1,\ldots,x_{\width(\cB)}$ be the indices of $X(\cB')$ ordered so that $f(x_1) \leq \cdots \leq f(x_{\width(\cB')})$.
Write $\width(\cB')$ as $qB+r$ where $r < B$, and note that $q \leq \width(\cB')/B \leq 2/\gridprec \leq s/2$.  Partition $x_1,\ldots,x_{qB}$ into
$q$ ``bins'', where each bin is a sequence of $B$ consecutive indices.
Let  $A_i$ for $i \in [1,q]$ be the number of samples from $R$ that fall in bin $i$.   

We now show that (1) with probability at least $1-\neterror$, $A_i \geq M$ for all $i \in [1,q]$, and
(2) 
if $A_i \geq M$ for each $i$, then $V$ is a $\netprec$-value net for $\cB'$. 

We prove (2) first.   Note that $V$ contains
at least $1$ value from each bin.
Let $J$ be a $\netprec$-popular subinterval of $Y(\cB')$.  Then $|F^{-1}(J) \cap X(\cB')| \geq \netprec \width(\cB') > 2B-2$.
Since $F^{-1}(J) \cap X(\cB')$ is a consecutive subsequence of $x_1,\ldots,x_{\width(\cB')}$, it must contain at least one bin,
and hence intersects $V$.

It remains to prove (1).  Fix $i \in [1,q]$ and  consider the event that  $A_i \leq M-1$. 
$A_i$ is the sum of $sM-1$ Bernoulli trials each having probability at least $\frac{B}{\width(\cB')} \geq \gridprec/2$, and so has expectation at least
$\frac{\gridprec}{2}(sM-1)\geq 2M-1$. By
\Prop{hoeff1}, $\prob[A_i \leq  M-1] \leq e^{-M/s}$. By the union bound, the probability that some $A_i\leq M-1$
is at most $se^{-M/s}/2$ and by the choice of $M$ this is at most $\neterror$.  This completes the proof of the value net construction for $\cB'$. 
\end{proof}

\subsection{Grids, the subroutine \buildgrid and grid digraphs}
\label{subsec:grids}

A grid $\Gamma$ is a {\em $\cB$-grid} if $X(\Gamma) \subset X(\cB)-\{x_L(\cB),x_R(\cB)\}$ and
$Y(\cB)$ is a value net for $\cB$. If $x_1<\cdots<x_k$ are the indices of $X(\Gamma)$, then
they define a $\cB$-strip decomposition  of $X(\cB)$ whose associated index partition is $(x_L(\cB),x_1],(x_1,x_2],\cdots,(x_k,x_R(\cB)]$.
We call this the strip decomposition of $\cB$ induced by $\Gamma$.

\begin{definition} \label{def:fine}
For a box $\cB$ and $\gridprec>0$, a grid $\Gamma$ is an {\em $\gridprec$-fine} $\cB$-grid if:

\begin{itemize}
\item $X(\Gamma)$ contains an index from every subinterval $I$ of  $X(\cB)$ having size exceeding $\gridprec \width(\cB)$.
\item $Y(\Gamma)$ is a $\frac{\gridprec}{|X(\Gamma)|}$-value net.
\end{itemize}
\end{definition}

We define a procedure $\buildgrid{}$ that  takes as input a triple $(\cB,\gridprec,\griderror)$ (as does
\buildnet)
and outputs a $\cB$-grid $\Gamma$ that is {\em $\gridprec$-fine} with probability at least $1-\griderror$.  

The procedure works as follows.
If $\width(\cB) \leq 1/\gridprec$ then $X(\Gamma)=X(\cB)$ and $Y(\Gamma)=Y(\cB)$. Otherwise, set $r= \lceil 1/\gridprec \rceil$ and for $i \in [0,r]$, define
$x_i=x_L(\cB) + \lfloor i \gridprec  \width(\cB) \rfloor$.  We take $X(\Gamma)=\{x_1,\ldots,x_{r-1}\}$.
It follows that $x_0=x_L(\cB)$, $x_r=x_R(\cB)$ and  $x_i-x_{i-1} \leq \lceil \gridprec \width(\cB) \rceil$ for each $i \in [1,r]$. 
The set $Y(\Gamma)$ is constructed by applying $\buildnet(\cB,\gridprec^2/2,\griderror)$.  
We say that  $\buildgrid(\cB,\gridprec,\griderror)$ {\em is reliable} if the grid $\Gamma$ that it outputs is
$\gridprec$-fine.  The definition of the grid, together with  \Prop{buildnet} implies:

\begin{prop}
\label{prop:buildgrid}  $\buildgrid(\cB,\gridprec,\griderror)$ produces a grid $\Gamma$ with 
$|X(\Gamma)| \leq 1+\lceil 1/\gridprec \rceil \leq 3/\gridprec$  and 
$|Y(\Gamma) \leq \lceil 8/\gridprec^2 \rceil \leq 16/\gridprec^2$  that is reliable with probability at least $1-\griderror$.
\end{prop}

{\bf Grid digraph $D(\Gamma)$:} This is associated with the
$\cB$-grid $\Gamma$. The vertex set is 
$\Gamma\cup \{P_{BL}(\cB),P_{TR}(\cB)\}$. The arc sets
consists of pairs $(P_{BL}(\cB),Q)$ where $Q$ lies in the leftmost columnn of $\Gamma$,
$(Q,P_{TR}(\cB))$ where $Q$ belongs to the rightmost column of $\Gamma$,  and $(P,Q)$ where
$P \prec Q$ and $P$ and $Q$ are in adjacent columns of $\Gamma$. $D(\Gamma)$ is acyclic and has unique
source $P_{BL}(\cB)$ and unique sink $P_{TR}(\cB)$.  A {\em $D(\Gamma)$-path} is a source-to-sink path in $D(\Gamma)$.
Every arc $(P,Q)$ of $D(\Gamma)$ corresponds to a box $\bx(P,Q)$ and a $D(\Gamma)$-path
corresponds to a $\cB$-box chain.  A box chain arising in this way is a $D(\Gamma)$-chain.
Each box correponding to an arc is called a \emph{grid box}.
Refer to \Fig{digraph}.

\medskip

The following lemma says that if $\Gamma$ is a $\gridprec$-fine $\cB$-grid then for any increasing sequence in $\cB$
there is a $D(\Gamma)$-chain that contains all but a small number of points from the sequence.  

\begin{lemma}(Grid approximation lemma)
\label{lem:grid-approx}
Suppose $\Gamma$ is a $\gridprec$-fine $\cB$-grid.
Let $\cL$ be an increasing sequence of $\cF$-points in $\cB$.
Then there is
a $D(\Gamma)$-chain $\vec{\cD}=(\cD_1,\dots,\cD_r)$
such that the subset $\cL-\vec{\cD}^{\cup}$ 
has size at most $\gridprec \width(\cB)$.
\end{lemma}

\begin{proof}
Let $X(\Gamma)=\{x_1<\ldots<x_{r}\}$. 
To specify a $D(\Gamma)$-chain we need to choose a nondecreasing comparable sequence of points $P_1,\ldots,P_{r}$,
such that $P_i$ is in column $x_i$ of $\Gamma$. 
%
Let $\cL_i$ be the portion of $\cL$ with $x$ coordinate at most $x_i$ and let $Q_i$ be the largest
point in $\cL_i$.  Take $P_i$ to be the least point of $\Gamma$ in column $i$ such that $Q_i \preceq P_i$.

The points of $\cL$ that lie outside of the corresponding box chain are those that belong to $P_i^{NW}$ or $P_i^{SE}$
for some point $P_i$.  Since every point of $\cL_i$ is less than $P_i$, there are no points of $\cL$ in $P_i^{NW}$.
The points of $\cL$ in $P_i^{SE}$ have $x$-coordinate greater than $x_i$ and $y$-coordinate
strictly between $y(Q_i)$ and $y(P_i)$. By choice of $P_i$, there is no value in $Y(\Gamma)$
strictly between $y(Q_i)$ and $y(P_i)$. Since $Y(\Gamma)$ is a $\gridprec/r$-value net for $\cB$,
there are at most $\gridprec\width(\cB)/r$ points of $\cL$ in $P_i^{SE}$.
%

Since each point $P_i$ is in violation with at most $\gridprec\width(\cB)/r$ points of $\cL$, the total
number of points of $\cL$ that violate some $P_i$ is at most $\gridprec \width(\cB)$.
\end{proof}
%

\section{The basic LIS approximation algorithm}
\label{sec:first algorithm}

In this section  we describe the algorithm \basicmain, which achieves the properties
stated in \Thm{alg1}.    As asserted in the theorem, \basicmain{} takes as input a natural number $n$, an array $f$
of size $n$ and an error parameter $\taupar \in (0,1)$ and outputs an estimate to $\lis_f$.  Recall that we assume the
values of the array are in the range $[1,\maxval]$, where $\maxval$ is a known integer.

The program \basicmain{} sets certain global parameters, initializes a parameter $t_{\max}$ and then calls a subroutine $\approxlis$, which is the main
part of the algorithm.  The subroutine $\approxlis$ takes as input a box $\cB$ and a nonnegative integer $t$ and outputs an estimate of $\lis_f(\cB)$.
We denote
an invocation  of $\approxlis$ on box $\cB$ with parameter $t$ by  $\approxlis_t(\cB)$. 
The global parameters are all used within $\approxlis$ and the procedures it calls. 
The values for the global parameters are chosen to make the error analysis work.
 The array size $n$ and the array $f$
are also treated as global parameters in $\approxlis$.

\vspace{-10pt}
\begin{center}
	\fbox{\begin{minipage}{\columnwidth}  
$\basicmain(n,f,\taupar)$\\
Output: Approximation to $\lis_f$.
\begin{enumerate}
\item Fix global parameters according to table (unchanged throughout algorithm).
\begin{center}
  \begin{tabular}{ | l || c | c |}
    \hline
    {\bf Name} & {\bf Symbol} & {\bf Value} \\ \hline
    Initial precision & $t_{\max}$ & $\lceil 4/\taupar \rceil$ \\ \hline
    Sample size parameter & $\ssz$ & $10(\log n)^4$ \\ \hline
    Grid precision parameter & $\gridprec$ & $\gamma$ \\ \hline
    Width threshold & $\omega$ & $1/\rho$ \\ \hline 
    Tainting parameter & $\taint$ & $\frac{1}{10 \log n}$ \\ \hline
    Primary splitter parameter & $\mu_r$ & $2/(r+3)$ \\ \hline
    Secondary splitter parameter & $\gamma$ & $\frac{1}{(C_1\log n)^3}$ \\ \hline
    Splitter balance parameter & $\rho$ & $\frac{1}{C_1\log n}$ \\ \hline
  \end{tabular}
\end{center}

\item Let box  $\cU$ be the  box $[1,n] \times [1,\maxval]$.
\item  Return $\approxlis_{t_{\max}}(\cU)$.
\end{enumerate}
\end{minipage}}
\end{center}

The algorithm $\approxlis$ is recursive, and in recursive calls will be run on r values of $t \leq t_{\max}$ and subboxes $\cB$ of $\cU$.
We will prove the following property of $\approxlis$:

\begin{theorem}
\label{thm:basic approxlis}  
Suppose $\approxlis$ is run with the global parameters set as in $\basicmain$.
On input a box $\cB \subseteq \cU$ and an integer $t$, $\approxlis_t(\cB)$:
\begin{itemize}
\item  runs in time $(\log n)^{O(t)}$, and
\item outputs a value that, with probability at least $1-n^{-\Omega(\log(n)}$, is a
{\em $(\tau_t,\delta_t )$}-approximation to $\lis(\cB)$, where:
\begin{eqnarray*}
\tau_t&=& \frac{4}{t} \\
\delta_t& = & \frac{t}{\log(n)}.
\end{eqnarray*} 
\end{itemize}
\end{theorem}

\Thm{alg1} follows immediately from \Thm{basic approxlis}:

\begin{proof}[of \Thm{alg1}]
$\basicmain(n,f,\taupar)$ returns the output of $\approxlis_{t_{\max}}(\cU)$.  By \Thm{basic approxlis}, this gives a
$(4/t_{\max},t_{\max}/\log(n))$-approximation to $\lis_f$ which runs in time $\log(n)^{O(t_{\max})}$.  Using the fact that $t_{\max} = \lceil 4/\taupar \rceil$ which is
between $4/\taupar$ and  $5/\taupar$ gives the desired running time and approximation error for $\basicmain$.
\end{proof}

\subsection{Description and pseudocode for \approxlis}
\label{subsec:first overview}

The procedure $\approxlis$, uses four subprocedures $\classify$, $\critbox$, $\terminalbox$, and $\gridchain$,
each taking as input a box $\cB$ and a nonnegative quality parameter $t$,
and possibly other input.  These procedures use 
the previously defined procedures $\findsplitter$ 
and $\buildgrid$. For later convenience, we put the quality parameter as a subscript to the procedure.

The main approximation algorithm $\approxlis_t(\cB)$ returns
an estimate of $\lis(\cB)$. If the
input box $\cB$ has width 1, then it outputs $|\cF \cap \cB|$. Otherwise,
$\approxlis_t$ uses a subroutine 
$\classify_t$. This takes as input $\cB$ 
and an index $x \in X(\cB)$ and outputs a classification of $x$ as {\em good} or {\em bad}.
$\approxlis_t(\cB)$ outputs an estimate of the number of good point by running $\classify_t(x,\cB)$
on a small random sample of indices.
%

Recall from Section \ref{subsec:random} 
that 
our viewpoint of fixing the random bits at the outset specifies the behavior of $\classify_t$ on every index.
We define $Good_t(\cB)$ to be the set of indices $x$ for which $\classify_t(x,\cB)$ would return {\em good}.
Thus $\approxlis_t(\cB)$ returns an estimate of $|Good_t(\cB)|$.  
The algorithm $\classify_t$ is designed so that $Good_t(\cB)$ is the index set of an increasing sequence,
and with probability close to 1, is close in size to $\lis(\cB)$.  Hence, $\approxlis_t(\cB)$
should be a good estimate of $\lis(\cB)$. 

The procedure $\classify_t$ is recursive.  If $F(x) \not\in \cB$, then $x$ is declared bad.  
Otherwise, if $\cB$ has width 1, we declare $x$ to be good,
and if $\width(\cB)>1$  and $t=0$ then $x$ is declared bad.
The main case ($F(x) \in \cB$, $\width(\cB) >1$ and $t \geq 1$) is accomplished by
calling
$\critbox_t(x,\cB)$, which returns a subbox $\cC$ of $\cB$
such that $x \in X(\cC)$. The procedure then recursively calls
$\classify_{t-1}(x,\cC)$.  The classification returned by this recursive call is the output
of $\classify_t(x,\cB)$.  

The procedure $\critbox_t(x,\cB)$  finds a subbox $\cC$, called the {\em critical subbox of $\cB$ for $x$}.
The procedure operates in two stages.
The first stage  is performed by $\terminalbox$,
which shrinks $\cB$ to a subbox $\cT$, called the {\em terminal box} with $x \in X(\cT)$.
Intuitively, $\terminalbox$  attempts to simulate the interactive protocol discussed in the introduction.
$\terminalbox(\cB,x)$ initializes $\cT$ to $\cB$.  It  uses the subroutine $\findsplitter{}$ 
to look for an index $s$ such that $F(s) \in \cT$
and $s$ is a good splitter for the $\cB$-strip $\strip{\cT}{\cB}$. 
If $\findsplitter$ succeeds
in finding $s$, then the splitter defines a box chain of size
2 spanning $\cT$, and $\cT$ is replaced by the box in the chain whose index set contains $x$.  This process is repeated 
until either $\width(\cT) \leq \omega$ ($\cT$ is {\em narrow}) or $\findsplitter$ fails to find a good splitter.
This ends the first stage.

In the second stage,  a box chain $\vec{\cC}(\cT)$ spanning $\cT$, called the {\em critical chain for $\cT$}, is constructed.
Intuitively, this part implements approximation boosting sketched in the introduction.
We
use $\buildgrid$ to build a suitably
fine grid for $\cT$ of size $(\log n)^{O(1)}$.  
We then recursively
evaluate $\approxlis_{t-1}(\cC)$ 
for every grid box $\cC$.  Think of these values as giving
a length function on the edges of the grid digraph.  The procedure
performs a longest path
computation to compute the exact longest  
path in the grid digraph 
from the lower left corner to the upper right
corner of $\cT$. (This computation takes $(\log n)^{O(1)}$ time.)

Having found the critical chain $\vec{\cC}(\cT)$, the output of $\critbox$ is
the box $\vec{\cC}(\cT)[x]$ of $\vec{\cC}(\cT)$  whose index set contains $x$.

The pseudocode for the algorithm is presented below.  
The arguments $\cB,\cT$ represent boxes, $t$ is a nonnegative integer, and $x$ is an index.
The array $f$ and the domain size $n$ are
treated as implicit global parameters.

\vspace{-10pt}
\begin{center}
	\fbox{\begin{minipage}{\columnwidth}  
$\approxlis_t(\cB)$\\
Output: Approximation to $\lis(\cB)$.
\begin{enumerate}
\item If $\width(\cB)=1$, output $|\cB \cap \cF|$. 
\item Otherwise ($\width(\cB)>1$): Select $\ssz$ uniform random indices 
from $X(\cB)$. Run $\classify_t(x,\cB)$ on
each sample point. Let $g$ be the number of points classified
as $good$ and return $g\width(\cB)/\ssz$.
\end{enumerate}
\end{minipage}}
\end{center}

\vspace{-10pt}
\begin{center}
	\fbox{\begin{minipage}{\columnwidth} 
$\classify_t(x,\cB)$\\
Output: {\em good} or {\em bad}
\begin{enumerate}
\item If $F(x) \not\in \cB$, return {\em bad}.
\item Otherwise ($F(x) \in \cB$)
\begin{enumerate}
\item Base case: If $\width(\cB)=1$, return {\em good}. If $t=0$ and $\width(\cB)>1$, return {\em bad}.
\item  Main case ($\width(\cB)>1$ and $t \geq 1$): 
\begin{enumerate}
\item $\cC \longleftarrow \critbox_t(x,\cB)$.
\item Run $\classify_{t-1}(x,\cC)$ and return its output. 
\end{enumerate}
\end{enumerate}
\end{enumerate}
\end{minipage}
}
\end{center}

\vspace{-10pt}
\begin{center}
	\fbox{\begin{minipage}{\columnwidth}  
$\critbox_t(x,\cB)$\\
Output: Subbox $\cC$ of $\cB$ such that $x \in X(\cC)$
\begin{enumerate}
\item $\cT \longleftarrow \terminalbox_t(x,\cB)$.
\item Call $\gridchain_t(\cT)$ and let $\vec{\cC}(\cT)$ be the chain of boxes returned.  
\item Return $\vec{\cC}(\cT)[x]$ (the box $\cC \in \vec{\cC}(\cT)$ with $x \in X(\cC)$).
\end{enumerate}
\end{minipage}
}
\end{center}

\vspace{-10pt}
\begin{center}
	\fbox{\begin{minipage}{\columnwidth}
$\terminalbox_t(x,\cB)$ \\
Output: subbox $\cT$ of $\cB$ such that $x \in X(\cT)$.
\begin{enumerate} 
\item Initialize $\cT$ to $\cB$ and boolean variable $\splitterfound$ to {\tt TRUE}.
\item Repeat until $\width(\cT) \leq \omega$ ($\cT$ is narrow) or $\splitterfound$ is {\tt FALSE}: 
\begin{enumerate}
\item Run $\findsplitter(\cT,\cB,\mu_t,\gamma \width(\cT),\rho)$:
returns boolean $\splitterfound$ and  index
$splitter$.
\item If $\splitterfound={\tt TRUE}$ then
\begin{enumerate}
\item If $x \leq splitter$ then replace $\cT$ by the
box $\bx(P_{BL}(\cT),F(splitter))$.
\item If $x > splitter$ then replace $\cT$ by the
box $\bx(F(splitter), P_{TR}(\cT))$.
\end{enumerate}
\end{enumerate}
\item Return $\cT$.
\end{enumerate}
\end{minipage}}
\end{center}

\vspace{-10pt}
\begin{center}
	\fbox{\begin{minipage}{\columnwidth} 
$\gridchain_t(\cT,\gridprec)$\\
Output: box chain $\vec{\cC}(\cT)$ spanning $\cT$.
\begin{enumerate}
\item Call $\buildgrid(\cT,\gridprec,n^{-2\log n})$
which returns a grid $\Gamma$.
\item  Construct the associated digraph $D(\Gamma)$
\item For each grid-box $\cD$ of $D(\Gamma)$.
recursively evaluate $\approxlis_{t-1}(\cD)$.
\item Compute the longest path in $D(\Gamma)$
from $P_{BL}(\cT)$ to $P_{TR}(\cT)$ according to the
length function $\approxlis_{t-1}(\cD)$.
\item Return the $\Gamma$-chain $\vec{\cC}(\cT)$ associated to the longest path.
\end{enumerate}
\end{minipage}}
\end{center}

\section{Properties of \approxlis}

In this section and the next we prove \Thm{basic approxlis} by showing that
a call to $\approxlis_t(\cB)$ runs in time
$(\log n)^{O(t)}$ and that with probability $1-n^{O((\log n)}$.   the output
is within $\tau_t \loss(\cB)+ \delta_t \width(\cB)$ of $\lis(\cB)$.

We first present the easy running time analysis.  We next turn to the much more difficult task of  
bounding the error of the estimate returned by the algorithm. We start with 
some important structural observations about the behavior of the algorithm. We then
identify two assumptions about the random bits used by the algorithm, which encapsulate 
all that we require from the random bits in the error analysis.  We show that these
assumptions hold with probability very close to 1.

In \Sec{correctness} 
we  formulate and prove \Thm{basic correctness}, which says
that whenever the two assumptions hold then $\approxlis$ returns a suitably accurate estimate, and which immediately
implies the desired error bounds in Theorem \ref{thm:basic approxlis}.
  
\subsection{Running time analysis}
\label{subsec:running time}

Let $A_t=A_t(n)$ be the running time of $\approxlis_t$  and
$C_t=C_t(n)$ be the running time of $\classify_t$ on boxes of width at most $n$.
In what follows we use $P_i=P_i(n)$ to denote functions of the form $a_i ((\log n))^{b_i}$, where $a_i,b_i$ are constants that
are independent of $n$ and $t$.

\begin{claim} \label{clm:runtime} For all $t \geq 1$,
\begin{eqnarray*}
A_t & \leq & P_1 C_t + P_2. \\
C_t & \leq & C_{t-1}+ P_3 A_{t-1}+P_4.
\end{eqnarray*}
\end{claim}

\begin{proof} The first recurrence with $P_1=\ssz=10(\log n)^4$ is immediate from the definition of 
$\approxlis_t$.  The function $P_2$ is an upper bound on the cost of operations within $\approxlis_t$ excluding the calls to $\classify_t$

For the second recurrence, the final recursive call to $\classify_{t-1}$  gives the $C_{t-1}$ term.
The rest of the cost comes from  $\critbox_t$ which invokes $\terminalbox_t$, which involves several iterations
whose cost is dominated by the cost of $\findsplitter$.  Each iteration reduces the width of $\cT$ by at least a $(1-\rho)$ 
factor, so the number of iterations is at most $(\log n)/\rho$.  The cost of $\findsplitter$ is $(\log n)^{O(1)}$,
so the cost of $\terminalbox_t$ is included in the term $P_4$.  $\critbox_t$ then calls $\gridchain_t$.  This involves building a grid of
size $(\log n)^{O(1)}$ and making one call to $\approxlis_t$ for each grid box, which accounts for the term $P_3 A_{t-1}$.  All of the rest
of the cost of $\gridchain_t$ is in doing a longest path computation on the grid digraph, which is absorbed into the $P_4$ term.  
\end{proof}

\begin{corollary} \label{cor:runtime} For all $t \geq 1$, $A_t$ and $C_t$ are in $\log n^{O(t)}$.
\end{corollary}

\begin{proof} Note that $A_0, C_0 = (\log n)^{O(1)}$. By the recurrences of \Clm{runtime},
$C_t \leq P_5(P_1P_3+1)^t$ and $A_t \leq P_6(P_1P_3+1)^t$  which are both $(\log n)^{O(t)}$.
\end{proof}

\subsection{The t-splitter tree and terminal chain}
\label{sec:terminal}

We analyze the structure of the output of $\terminalbox_t$.
This procedure takes three parameters: the level $t$, the box $\cB$ and an index $x$.
It also uses randomness within the calls to $\findsplitter$. 
 Recall from Section \ref{subsec:random} that we classify the random bits used in $\approxlis$ as primary random bits
and all other random bits as secondary. Since $\terminalbox_t$ never calls $\approxlis$, all random bits used are secondary.

In the following discussion, the box $\cB$, level $t$, and secondary
random bits as fixed. Under this view, $\terminalbox$ maps each index $x \in X(\cB)$
to a box $\cT(x)$.   
We now define the {\em $t$-splitter tree}, which summarizes all important information
about the execution and output of $\terminalbox_t(x,\cB)$.


For each subbox $\cT$ of $\cB$, consider the output of $\findsplitter_t(\cT,\cB,\mu_t,\gamma \width(\cT),\rho)$.
If $\splitterfound$ is {\tt FALSE} we say that $\cT$ is {\em splitterless}. Otherwise, we say that $\cT$ is {\em split}.
For each split box $\cT$, $\findsplitter$ returns a splitter $s=splitter(\cT)$ which is used to define two
subboxes of $\cT$: the {\em left child} $\bx(P_{BL}(\cT),F(s))$ and the {\em right child} $\bx(F(s),P_{TR}(\cT))$.
Note that, when viewed as boxes in the plane, the left child lies below and to  the left of the right child.  The
left- and right-child relations together
define a directed acyclic graph on the subboxes of $\cB$ in which each splitterless box has out-degree 0
and each split box has out-degree 2.  Note that if there is a path from box $\cT$ to box $\cT'$ then
$\cT' \subset \cT$.

Now consider the subdigraph $R(\cB)$ induced on the set of boxes reachable from $\cB$.  

\begin{itemize}
\item $R(\cB)$ is a binary tree rooted at $\cB$.  
We refer to $R(\cB)$ as the {\em $t$-splitter tree for $\cB$} and the leaves of the tree as {\em terminal boxes}.
\item The sequence of boxes encountered along any root-to-leaf path are nested.
\item Two boxes $\cT$ and $\cT'$  in the tree such that neither is an ancestor of the other have disjoint index sets, i.e. $X(\cT) \cap X(\cT') = \emptyset$.  
\item The terminal boxes together form a box chain that spans $\cB$, called the {\em terminal chain}, which we denote by $\vec{\cT}=\vec{\cT}_t(\cB)$.
\item
Recall (from Section \ref{subsec:basic defs}) that the box chain $\vec{\cT}$ has an associated increasing sequence of points
and $\cP^{\circ}(\vec{\cT})$ denotes this sequence, excluding the first and last point.
Every point $P \in \cP^{\circ}(\vec{\cT)}$ is the splitter of a unique non-terminal box of $R(\cB)$, denoted $\cH(P)$.
\item The
point sequence $\cP^{\circ}(\vec{\cT})$
associated to $\vec{\cT}$ is the same as the sequence obtained by doing 
a depth first traversal of the tree $R(\cB)$, always visiting the left child of a node before the right child,
and recording
the sequence of splitters $splitter(\cH)$ in post-order (listing the splitter of box $\cH$ immediately after having listed all splitters in its left subtree).
\item Each terminal box $\cT$ contains a grid $\Gamma(\cT)$, as formed in the procedure $\gridchain$. We can concatenate
grid chains spanning each $\cT \in \vec{\cT}$ to get a box chain that spans $\cB$. This is called
a \emph{spanning terminal-compatible grid chain in $\cB$}. Note that there are many possible such spanning terminal-compatible grid chains.
\end{itemize}

The following lemma is evident from the above observations and the definition of $\terminalbox_t$:

\begin{lemma}
\label{lem:terminal}
For every $x \in F^{-1}(\cB)$, the set of boxes in the $t$-splitter tree of $\cB$ whose index set contains $x$
is a root-to-leaf path in the tree.   This path is equal to the sequence of boxes produced during the execution
of $\terminalbox_t(x,\cB)$.  In particular, the leaf that is reached is the terminal box that is returned
by $\terminalbox_t(x,\cB)$, and is equal to
$\vec{\cT}[x]$ (the unique box $\cT \in \vec{\cT}$ such that $x \in X(\cT)$).
\end{lemma}

\subsection{Two assumptions about  the random bits}
\label{subsec:assumptions}

As described in Section \ref{subsec:random}, the random bits used in the algorithm are classified as 
either {\em secondary random bits} (those used in $\findsplitter$ and $\buildgrid$), and  {\em primary
randomness} used within $\approxlis$.
  Note that the procedures $\findsplitter$ and $\buildgrid$
do not involve calls to themselves or other procedures, while $\approxlis_t$  makes calls to $\classify_t$, and $\classify_t$ makes
calls to $\approxlis_{t-1}$.    The primary random bits used in all calls
to $\approxlis_t$ for a fixed $t$ are called the {\em level $t$ random bits}.

We now identify two assumptions about the random bits used in the algorithm
and show that these assumptions hold with probability $1-n^{\Omega(-(\log n))}$. 
These assumptions encapsulate the only properties of  the random bits needed for the error analysis.
In the main analysis performed in the next section, we assume that
all random bits are fixed so that these conditions are satisfied.  
The algorithm can then be viewed as deterministic.
The first assumption involves the secondary random bits.

\smallskip
\noindent
{\bf Assumption 1.}  For every possible 
choice of arguments, the procedures $\findsplitter$ and $\buildgrid$ are reliable according to the definitions in Sections \ref{subsec:splitters} and \ref{subsec:grids}.

\smallskip

Propositions \ref{prop:splitter} and \ref{prop:buildgrid} imply that the probability that a call to $\findsplitter$ or $\buildgrid$ is unreliable
is $n^{-\Omega((\log n)}$.  As indicated earlier, there are at most $n^{O(1)}$ different possible arguments to
either procedure, so we can apply a union bound.

\begin{prop}
\label{prop:assumption 1}
The probability that Assumption 1 fails is at most $n^{-\Omega(\log(n)}$.
\end{prop}

We henceforth view the secondary randomness as fixed in a way that satisfies Assumption 1.
We note an important consequence of Assumption 1 that we'll need later.
Recall that for a box $\cB$, with terminal chain $\vec{\cT}$,
each point $P \in \cP^{\circ}(\cT)$ was found as the splitter of a unique non-terminal box
$\cH(P)$.  Under Assumption 1,  and the definition of reliable, each of those splitters is 
$(\mu_t,2\gamma w(\cH(P)))$-safe for $\cH(P)$.

\begin{prop}
\label{prop:ass 1 cons}
Under assumption 1, for any box $\cB$ with terminal chain $\vec{\cT}$, each of the
points $P \in \cP^{\circ}(\cT)$ is $(\mu_t,2\gamma w(\cH(P)))$-safe for th box $\cH(P)$.
\end{prop}
%
%
%
%

Next we turn to the second assumption.
Assumption 2 will state the conditions we need for the primary randomness.  
To formulate this assumption,
we now  introduce  a somewhat technical definition of \emph{tainted boxes}. 
We do not need this definition to analyze the basic  $\approxlis$ algorithm, but
the improved version will need this notion.   We introduce this notion here because
this will allow us to reuse the proof for the improved algorithm.

For any $x \in \cB$ and integer $t$, the output of $\classify_t(x,\cB)$ is either \emph{good}
or \emph{bad}.
The set of indices classified \emph{good} is denoted by $Good_t(\cB)$.
The procedure $\approxlis_t$ tries to approximate $|Good_t(\cB)|$ by random sampling.
The randomness used for this random sampling is primary randomness from level $t$.
 is therefore  independent of $Good_t(\cB)$.

\begin{definition} \label{def:tainted}  Let $\taint$ be the taint parameter specified 
within \basicmain{}.
A box-level pair $(\cB,t)$ is said to be \emph{tainted} if $t \geq 1$ and $\width(\cB)>1$ and at least one of the following holds:
\begin{itemize}
	\item $|\approxlis_t(\cB)  - |Good_t(\cB)|| > \taint \width(\cB)$.
	\item There exists a  spanning terminal-compatible grid chain $\vec{\cC}$ for $\cB$, such that 
	the total width of the boxes $\{\cC \in \vec{\cC} | \textrm{$(\cC,t-1)$ is tainted}\}$ is at least $\taint \width(\cB)$.
\end{itemize}
\end{definition}

\smallskip
\noindent
{\bf Assumption 2.}  There are no tainted box-level pairs.

\begin{prop}
\label{prop:assumption 2}
The probability that Assumption 2 fails   is at most  $n^{-\Omega(\log(n))}$.  
\end{prop}

\begin{proof}
By Proposition \ref{prop:hoeff2} for each box-level pair $(\cB,t)$,  the probability that it satisfies the first condition of tainting, $|\approxlis_t(\cB)  - |Good_t(\cB)|| > \taint \width(\cB) (= \width(\cB)/(10\log n))$,
holds with probability at most $n^{-\Omega((\log n))}$. Taking a union bound over all box-level  pairs $(\cB,t)$ ensures that with probability
at least $1-n^{-\Omega(\log(n)}$, there are no box-level pairs that satisfy the first condition for being tainted.  If no box-level pair satisfies the first condition 
for being tainted, then a trivial induction on $t$ implies
that  no such pair satisfies the second condition either.

There is a subtle point here.   The set $Good_t(\cB)$ depends on the random bits, as does the set of indices
sampled by $\approxlis_t(\cB)$. When we apply Proposition \ref{prop:hoeff2}, the set $A$ is $Good_t(\cB)$, which
is itself determined randomly.  It is crucial that the
set of indices selected for sampling is uniformly distributed after conditioning on the 
set $Good_t(\cB)$. This is indeed the case.  The set
$Good_t(\cB)$ is determined by the output of $\classify_t(x,\cB)$ on all $x \in X(\cB)$.  The random bits needed to determine
these are the secondary random bits, together with the primary random bits of level at most $t-1$.    
The choice of the sample in $\approxlis_t(\cB)$ depends on the primary bits at level $t$, and is therefore
a uniformly distributed sample of the $Good_t(\cB)$.
\end{proof}

%
\section{Analysis of correctness}
\label{sec:correctness}

We introduce some notation.

\begin{itemize}
\item $\alis_t(\cB)$ is the output of $\approxlis_t(\cB)$.
\item $\aloss_t(\cB)=|\cB \cap \cF| - \alis_t(\cB)$.
\item $\liserror_t(\cB)=\lis(\cB)-\alis_t(\cB)$.
\end{itemize}

\begin{theorem}
\label{thm:basic correctness}
Suppose the random bits satisfy Assumptions 1 and 2. For all $t \geq 1$ and boxes $\cB \subseteq \cU(f)$:
%
\begin{eqnarray}
\label{goal 1}
|\liserror_t(\cB)| & \leq & \tau_t\loss(\cB)+\delta_t \width(\cB),
\end{eqnarray}
where $\tau_t=\frac{4}{t}$ and $\delta_t=\frac{t}{\log n}$.
\end{theorem}

\begin{proof}[of \Thm{basic approxlis}]  The claimed run time for
$\basicmain$ follows from \Cor{runtime}.  By \Prop{assumption 1} and \Prop{assumption 2},
Assumptions 1 and 2 hold with probability $1-n^{O(\log(n)}$ and by \Thm{basic correctness} this is enough
to guarantee the desired error bound.
\end{proof}

So it remains to prove \Thm{basic correctness}.
Recall that we refer to the term $\tau_t \loss(\cB)$ as {\em primary error} and  to 
the term $\delta_t \width(\cB)$ as {\em secondary error}.
The secondary error term comes from several sources, including sampling error.
It is a bit of a nuisance to track, but is not the main issue in the analysis;
we will have the
freedom to make $\delta_t$ as small as we like. In contrast, 
the coefficient  $\tau_t$ and primary splitter parameter $\mu_t$ are 
tightly constrained by the structure of our recursive algorithm.

In \Sec{primary},
we focus on the primary error terms. We structure the analysis to show how $\tau_t$ and  $\mu_t$ were determined. We identify 
a series of \emph{secondary error terms} $\errorterm_i$ (for $i$ betwen 1 and 5).
In Section \ref{subsec:secondary analysis} we show that the sum of these error terms is bounded above by $\delta_t \width(\cB)$. 
After isolating these secondary terms we will be left
with a recurrence that constrains $\tau_t$ by an expression involving $\mu_t$ and $\tau_{t-1}$. We 
choose $\mu_t$ to minimize this expression, which yields a recurrence inequality for $\tau_t$ in terms of
$\tau_{t-1}$. By inspection, $\tau_t=4/t$ satisfies this recurrence.

%
\subsection{Setting up the proof} \label{sec:stripnot}

We now summarize some notation about subsets of the box $\cB$.
\begin{itemize}
\item $\cL$ denotes a fixed LIS of $\cB$.
\item   The terminal chain $\vec{\cT}$ spans the box $\cB$, and the associated sequence $\vec{\cS}(\cT)$ of strips of the form
$\strip{\cT}{\cB}$ is a strip decomposition of $\cB$.
\item For  each terminal box $\cT$, there is an associated grid $\Gamma(\cT)$ which is constructed by the subroutine
$\buildgrid$, and a grid chain $\vec{\cC}(\cT)$ 
in $\Gamma(\cT)$  which is constructed by a call to $\gridchain$.  (We remind the reader that in the analysis we assume
that we have generated separate random bits for each subroutine and each choice of input parameters so that
the output of the subroutine is specified whether or not we actually execute it.)
\end{itemize}
  
Much of our analysis focuses on the behavior of $\cL$ as well as our algorithm within each strip.
This motivates the following notation. Refer to \Fig{lis-strip}.
\begin{figure*}[tb]
  \centering
  \subfloat[Dividing $\cL$ into strips]{\includegraphics[width=0.35\textwidth]{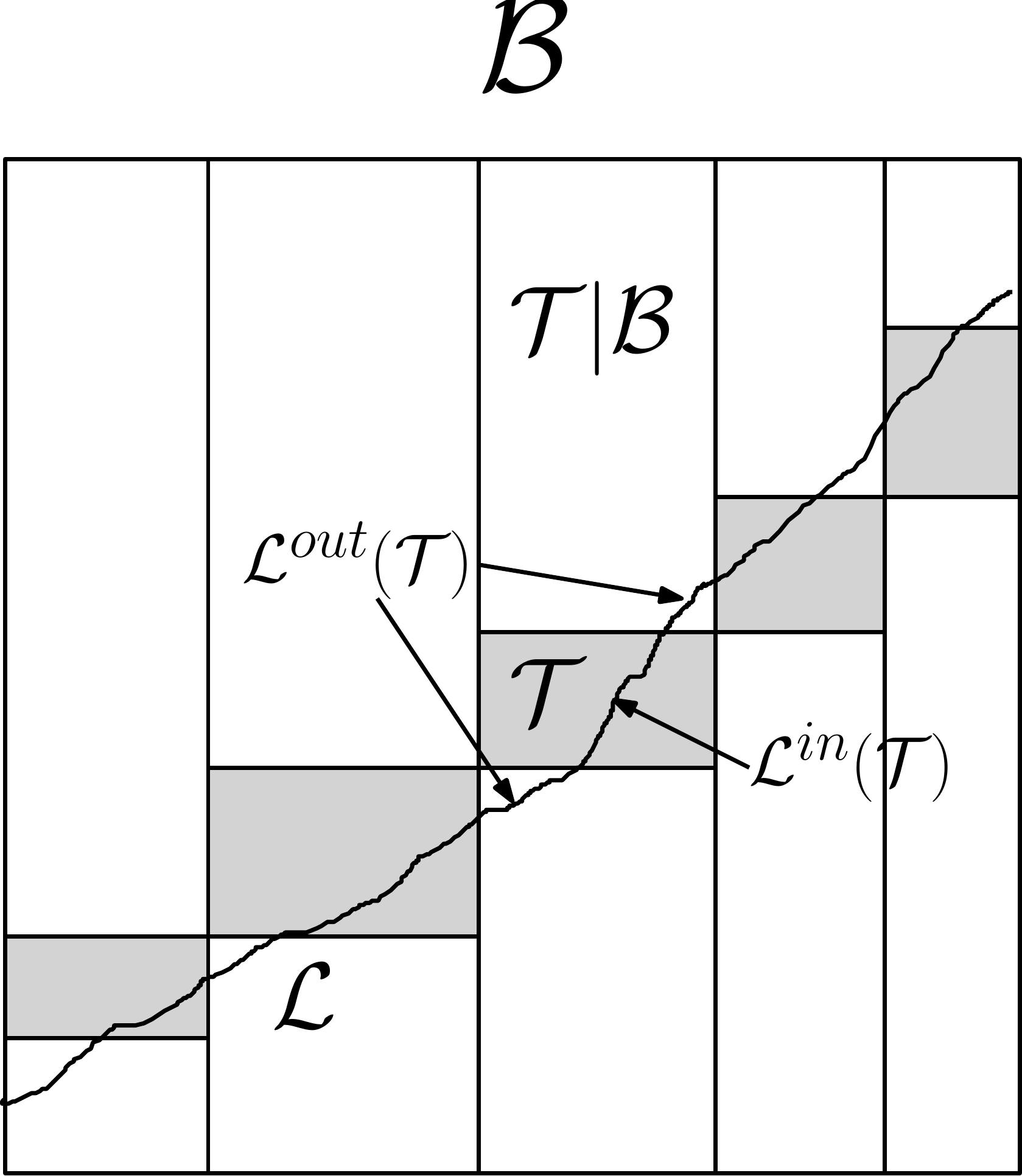} \label{fig:lis-strip}} 
  $\qquad \qquad$
  \subfloat[The grid chain $\vec{\cD}(\cT)$]{\includegraphics[width=0.35\textwidth]{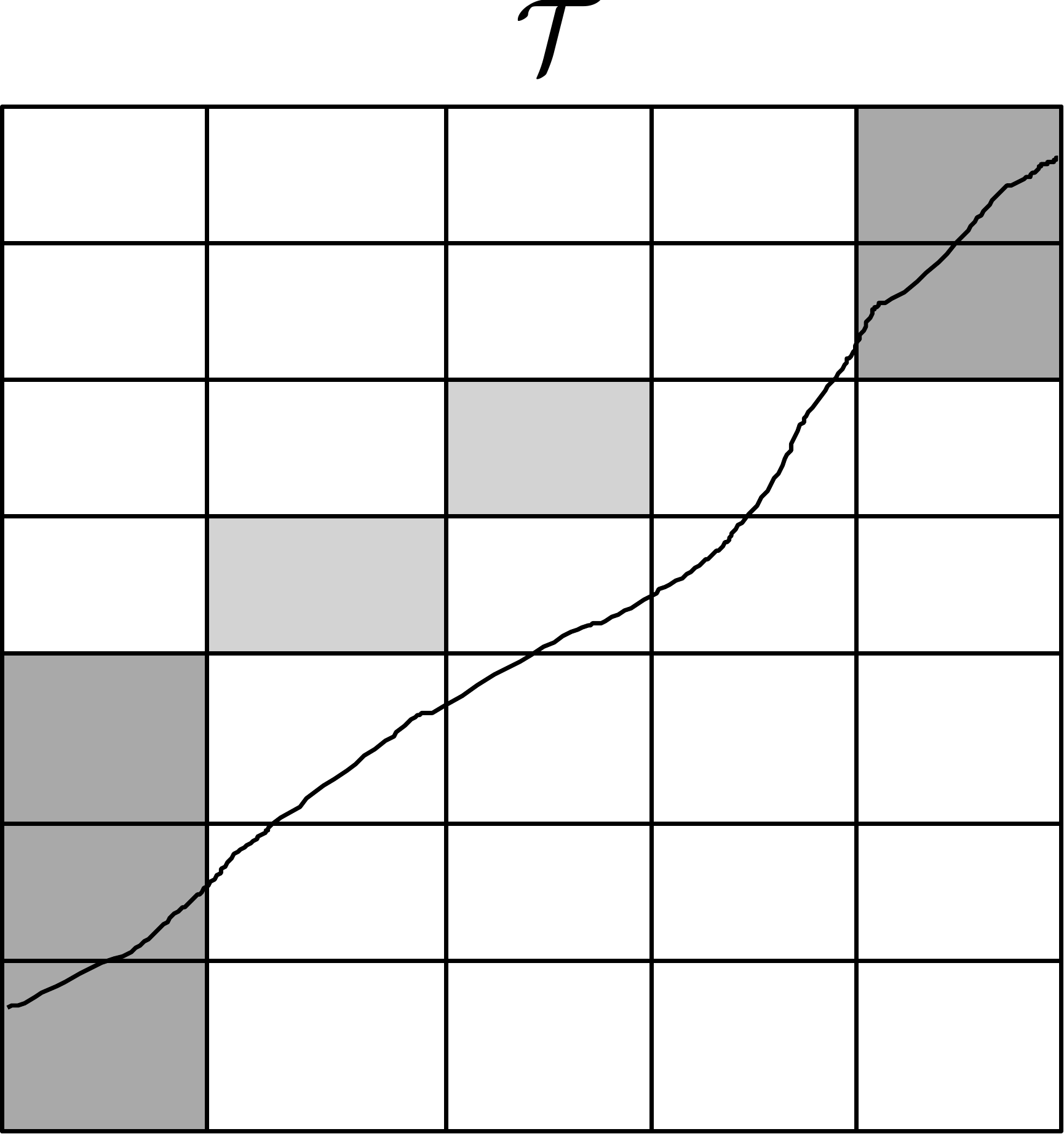} \label{fig:critical}}
    \caption{(a) The terminal chain $\vec{\cT}$ is given by shaded boxes. The LIS $\cL$ is depicted as a monotonically increasing freeform curve.
    (b) The grid chain $\vec{\cD}(\cT)$ is denoted by the shaded boxes. 
The set $\vec{\cE}$ of boxes (not indicated explicitly) are the boxes intersecting the LIS that are not tainted.}
\end{figure*}
%
%
%
For each terminal box $\cT$:
\label{subsec:more defs}
\begin{itemize}
	\item $\beta(\cT) = |(\strip{\cT}{\cB}) \cap \cF|$.
	\item $\cL(\cT)=\cL \cap (\strip{\cT}{\cB})$.
	\item $\cL^{in}(\cT) = \cL \cap \cT$.
	\item $\cL^{out}(\cT) = \cL \cap ((\strip{\cT}{\cB}) - \cT)$.
\end{itemize}

 We give some notation regarding critical boxes.
\begin{itemize}
	\item The concatenation of the $\vec{\cC}(\cT)$ for all $\cT \in \vec{\cT}$ is a box chain $\vec{\cC}$ called the {\em critical chain}. Members of this chain are {\em critical boxes}.
	\item For index $x  \in X(\cB)$, $\vec{\cC}[x]$ denotes the unique critical box $\cC$ such that $x \in X(\cC)$.
	Observe that for each $x \in F^{-1}(\cB)$, the function $\critbox_t(x)$ returns $\cC[x]$.
\end{itemize}

For each terminal box $\cT$, let $\vec{\cD}(\cT)$ denote the grid chain of $\Gamma(\cT)$ containing the maximum number
of points from $\cL$. Refer to \Fig{critical}. 
We will actually use the subsequence $\vec{\cE}(\cT)$, which consists of 
boxes $\cE \in \vec{\cD}(\cT)$ such that $\cE \cap \cL \neq \emptyset$ and $(\cE,t-1)$ is not tainted.
The following quantities are used heavily in our proof.
%
%
%
\begin{itemize}
\item $\outside(\vec{\cD}(\cT)) = \beta(\cT) - \sum_{\cD \in \vec{\cD}(\cT)} |\cD \cap \cF|$.
(Number of $\cF$-points in $(\cT|\cB)$ outside $\vec{\cD}(\cT)$.)
\item $\outside(\vec{\cE}(\cT)) = \beta(\cT) - \sum_{\cE \in \vec{\cE}(\cT)} |\cE \cap \cF|$.
(Number of $\cF$-points in $(\cT|\cB)$ outside $\vec{\cE}(\cT)$.)
\item $\loss(\vec{\cE}(\cT)) = \sum_{\cE \in \vec{\cE}(\cT)} \loss(\cE)$.  (Recall that  $\loss(\cE)=|\cE \cap \cF| - \lis(\cE)]$, so $\loss(\vec{\cE})$ is the number
of $\cF$-points from $\bigcup_{\cE \in \vec{\cE}(\cT)} \cE$  that are missed by the union of
the  LIS  for $\cE \in \vec{\cE}(\cT)$.)
\item $\aloss_t(\vec{\cE}(\cT)) = \sum_{\cE \in \vec{\cE}(\cT)} \aloss(\cE)$. (Recall that  
$\aloss(\cE)=[|\cE \cap \cF| - \alis_t(\cE)]$ so $\aloss_t(\vec{\cE}(\cT))$ is the
estimate of $\loss(\vec{\cE}(\cT))$ using $\alis_t$ in place of \lis.)
\end{itemize}
We begin with a few simple propositions.  The first, which follows immediately from the definition
of $\classify$ specifies what happens in the base case $\width(\cB)=1$.
%
\begin{prop}
\label{prop:narrow}
If $\cB$ is a box of width 1, then   $\alis_t(\cB)=\lis(\cB)$ and 
$\liserror_t(\cB)=0$. For the unique $x \in X(\cB)$, $x \in Good_t$ if and only if $F(x) \in \cB$.
\end{prop}
Henceforth, we assume that $\cB$ has width at least 2. We remind that $Good_t(\cB)$ denotes the set of indices
ini $X(\cB)$ classified as  {\em good} by $\classify_t(\cB)$.
The quantity $\alis_t(\cB)$ returned by the algorithm is an estimate of $|Good_t(\cB)|$. 
\begin{prop}
\label{prop:good}
For any box $\cB$ of width at least 2, and $t \geq 1$:
\begin{enumerate}
\item $Good_t(\cB)$ is equal to the union over critical boxes $\cC$ of $Good_{t-1}(\cC)$.
\item $Good_t(\cB)$ indexes an increasing sequence in $\cB$ and thus $|Good_t(\cB)| \leq \lis(\cB)$.
\end{enumerate}
\end{prop}
\begin{proof}
The first part follows from the main case of $\classify_t$.
For the second part, if $x,y \in Good_t(\cB)$
then either $F(x)$ and $F(y)$ belong to different boxes of $\vec{\cC}$
or $F(x),F(y)$ lie in the same critical box $\cC$.
In the former case, they are comparable since the boxes lie in a box chain.
In the latter case, they are both classified good by $\classify_{t-1}(\cdot,\cC)$
and must be comparable by induction on $t$.
\end{proof}

\subsection{Components of secondary error } 
\label{subsec:secondary components}

In this section we identify five components of secondary error, denoted $\errorterm_1,\ldots,\errorterm_5$.
We gather the definitions together here for easy reference, but the motivation for each term comes from the analysis
presented in Section \ref{sec:primary}.  

$$\errorterm_1 = \alis_t(\cB) - |Good_t(\cB)|.$$

Recall that the subroutine $\terminalbox_t$ together with the secondary random bits determine the terminal  box chain $\vec{\cT}$.
The remaining four secondary error terms are associated to each individual terminal box $\cT$  in $\vec{\cT}$; the secondary
error will be obtained by summing these over all $\cT$. We define:

\begin{eqnarray*}
\errorterm_2(\cT) & = & \sum_{\cC \in \vec{\cC}(\cT)} \alis_{t-1}(\cC) - \sum_{\cC \in \vec{\cC}(\cT)}|Good_{t-1}(\cC)|.\\
\errorterm_3(\cT) & = & |\cL^{in}(\cT)| - \sum_{\cE \in \vec{\cE}(\cT)} |\cE \cap \cL|.\\
\errorterm_4(\cT) &= &|\cL^{out}(\cT)| - \mu_t \cdot \outside(\vec{\cE}(\cT)).\\
\errorterm_5(\cT)
& = &  \mu_t \cdot \aloss_{t-1}(\vec{\cE}(\cT)) - (1-\mu_t)\cdot \outside(\vec{\cE}(\cT)).
\end{eqnarray*}

\subsection{Transforming our goal}
\label{sec:transform}

Our goal is to upper bound $|\nu_t(\cB)|=|\lis(\cB)-\alis_t(\cB)|$.
We start from the definition of $\liserror_t(\cB)$, substitute $|Good_t(\cB)|+ \errorterm_1$ for $\alis_t(\cB)$
and apply the inequality (from \Prop{good}) $\lis(\cB) \geq Good_t(\cB)$:
\begin{eqnarray*}
|\liserror_t(\cB)|&=&|\lis(\cB)-\alis_t(\cB)|\\
& = & |\lis(\cB) - |Good_t(\cB)| -\errorterm_1|\\
&\leq & \lis(\cB) - |Good_t(\cB)| + |\errorterm_1|
\end{eqnarray*}
%
%
Define:
\[
\Delta_t=\lis(\cB)-|Good_t(\cB)|-\tau_t\loss(\cB),
\]
Our goal (\ref{goal 1}) will follow from:

\[
\Delta_t \leq \delta_t\width(\cB)-|\errorterm_1|.
\]

Noting that $\lis(\cB)=|\cL|$ and $\loss(\cB)=|\cB \cap \cF|-|\cL|$. we have:

\[
\Delta_t= (1+\tau_t)|\cL|-|Good_t(\cB)|-\tau_t|\cB \cap \cF|
\]

%
Each term on the right counts a subset of $\cB \cap \cF$. Partitioning the box $\cB$ into the strips $(\cT|\cB)$
for each terminal box $\cT$, we define  $\Delta_t(\cT)$ as follows to be the contribution of the points in
$\cT|\cB$ to the right hand side (using the notation in Section \ref{sec:stripnot} 
and noting that $\Good_t(\cB) \cap (\cT|\cB) \subseteq \cT$):

\begin{align} \label{eq:Delta}
\Delta_t(\cT) & = (1+\tau_t) |\cL(\cT)| -  |Good_t(\cB) \cap \cT | - \tau_t \beta(\cT) 
\end{align}

Our goal now is to show:
\begin{eqnarray}
\label{eq:goal12.1}
\sum_{\cT \in \vec{\cT}} \Delta_t(T) & \leq & \delta_t \width(\cB)-|\errorterm_1|.
\end{eqnarray}
This is broken down into two steps. We use the notation $x^+ = \max(x,0)$.

\begin{claim} \label{clm:primary} (Primary error)
Let $t \geq 1$ and let $\cB$ be a box such that $(\cB,t)$ is untainted.  Let $\vec{\cT}$ be the
terminal chain associated with a call of $\approxlis_t(\cB)$.  We have:
\begin{eqnarray}
\label{sum Delta}
\sum_{\cT \in \vec{\cT}} \Delta_t(\cT) & \leq & (\delta_t - \delta_1) \width(\cB)+\sum_{\cT \in \vec{\cT}}  (\errorterm_2(\cT) + 5 \errorterm_3(\cT)^+ + 5\errorterm_4(\cT)^+ + 2 \errorterm_5(\cT)^+).
\end{eqnarray}
\end{claim}

\begin{claim} \label{clm:secondary} (Secondary error)
%
With the same hypotheses as in Claim \ref{clm:primary} together with Assumption 1, we have:
\begin{eqnarray}
\label{goal 1.3}
|\errorterm_1|+\sum_{\cT \in \vec{\cT}} (\errorterm_2(\cT) + 5 \errorterm_3(\cT)^+ + 5\errorterm_4(\cT)^+ + 2 \errorterm_5(\cT)^+) & \leq & \delta_1\width(\cB).
\end{eqnarray}
\end{claim}

Summing these bounds yields \Eqn{goal12.1}, which proves \Thm{basic correctness}. 

\subsection{Bounding the primary error} \label{sec:primary}
Here we prove Claim \ref{clm:primary}.
The proof is by induction on $t$.  The base case is $t=1$.  We prove  the base case and the induction step together,
indicating where they differ.
%
%

The following claim  bounds $\Delta_t(\cT)$ in a more convenient form.

\begin{claim} \label{clm:delta}
For each terminal box $\cT$,
\begin{eqnarray*}
\Delta_t(\cT) & \leq & (1+\tau_t)|\cL^{out}(\cT)| - \tau_t \cdot \outside(\vec{\cE}(\cT))
- (1+\tau_t) \cdot \loss(\vec{\cE}(\cT)) + \aloss_{t-1}(\vec{\cE}(\cT)) \\
%
%
&&+\errorterm_2(\cT)+(1+\tau_t)\errorterm_3(\cT).
\end{eqnarray*}
\end{claim}

\begin{proof} 
By Proposition \ref{prop:good}, $|Good_t(\cB)\cap \cT|=\sum_{\cC \in \vec{\cC}(\cT)} |Good_{t-1}(\cC)|$. Thus,
(\ref{eq:Delta}) can be rewritten as:
\begin{eqnarray}
\nonumber
\Delta_t(\cT) & = & (1+\tau_t) |\cL(\cT)| - \sum_{\cC \in \vec{\cC}(\cT)} |Good_{t-1}(\cC)| - \tau_t \beta(\cT) 
\end{eqnarray}

We replace $|Good_{t-1}(\cC)|$ by $\alis_{t-1}(\cC)$, using $\errorterm_2(\cT)$. 

\begin{eqnarray*}
\label{eq:Delta2}
\Delta_t(\cT)
& = & (1+\tau_t)|\cL(\cT)| - \sum_{\cC \in \vec{\cC}(\cT)} \alis_{t-1}(\cC) - \tau_t \beta(\cT) + \errorterm_2(\cT)
\end{eqnarray*}

By construction, $\vec{\cC}(\cT)$  is a grid-chain with respect to the grid
$\Gamma(\cT)$ 
that maximizes the sum $\sum_{\cC \in \vec{\cC}(\cT)} \alis_{t-1}(\cC)$.  This sum is at least
$\sum_{\cE \in \vec{\cE}(\cT)} \alis_{t-1}(\cE)$, since $\vec{\cE}$ is a subsequence of a grid-chain.
\begin{eqnarray*}
\nonumber
\Delta_t(\cT) & \leq &(1+\tau_t)|\cL(\cT)| - \sum_{\cE \in \vec{\cE}(\cT)} \alis_{t-1}(\cE) - \tau_t \beta(\cT) + \errorterm_2(\cT)
\end{eqnarray*}
Now, $|\cL(\cT)|=|\cL^{in}(\cT)|+|\cL^{out}(\cT)|$.  Using the definition of $\errorterm_3(\cT)$,  
$|\cL^{in}|=\errorterm_3(\cT) + \sum_{\cE \in \vec{\cE}(\cT)} |\cE \cap \cL|$;  the grid approximation
lemma (\Lem{grid-approx}) will be used to show that $\errorterm_3(\cT)$ is small.  
Furthermore, for each $\cE \in \vec{\cE}(\cT)$, $|\cE \cap \cL| \leq \lis(\cE)$. Combining,
\begin{eqnarray*}
\Delta_t(\cT) & \leq &(1+\tau_t)|\cL^{out}(\cT)| + (1+\tau_t)\sum_{\cE \in \vec{\cE}(\cT)} \lis(\cE)
- \sum_{\cE \in \vec{\cE}(\cT)} \alis_{t-1}(\cE) - \tau_t \beta(\cT) \\
& & + \errorterm_2(\cT)+(1+\tau_t)\errorterm_3(\cT)
%
\end{eqnarray*}
Substituting $\lis(\cE) = |\cE \cap \cF|-\loss(\cE)$, and $\alis_t(\cE)=|\cE \cap \cF|-\aloss(\cE)$ and
performing some simple algebraic manipulations completes the proof:
\begin{eqnarray*}
\Delta_t(\cT) & \leq &(1+\tau_t)|\cL^{out}(\cT)| + (1+\tau_t)\Big[\sum_{\cE \in \vec{\cE}(\cT)} |\cE \cap \cF|
- \loss(\vec{\cE}(\cT))\Big]
- \Big[\sum_{\cE \in \vec{\cE}(\cT)} |\cE \cap \cF|-\aloss_{t-1}(\vec{\cE}(\cT))\Big] \\
& & - \tau_t \beta(\cT) + \errorterm_2(\cT)+(1+\tau_t)\errorterm_3(\cT) \\
& = & (1+\tau_t)|\cL^{out}(\cT)| - \tau_t \Big[\beta(\cT) - \sum_{\cE \in \vec{\cE}(\cT)} |\cE \cap \cF|\Big]
- (1+\tau_t) \cdot \loss(\vec{\cE}(\cT)) + \aloss_{t-1}(\vec{\cE}(\cT)) \\
&&+\errorterm_2(\cT)+(1+\tau_t)\errorterm_3(\cT) \\
& = & (1+\tau_t)|\cL^{out}(\cT)| - \tau_t \cdot \outside(\vec{\cE}(\cT))
- (1+\tau_t) \cdot \loss(\vec{\cE}(\cT)) + \aloss_{t-1}(\vec{\cE}(\cT)) \\
&&+\errorterm_2(\cT)+(1+\tau_t)\errorterm_3(\cT).
\end{eqnarray*}
%
\end{proof}

\begin{claim}
$\loss(\vec{\cE}(\cT)) \geq K_t[\aloss_{t-1}(\vec{\cE}(\cT))-\delta_{t-1} \width(\cT)]$, where
$K_1=0$ and $K_t = 1/(1+\tau_{t-1})$ for $t \geq 2$.
\end{claim}

\begin{proof} For $t=1$, simply note that the loss is always non-negative.

Suppose $t \geq 2$.
By construction of $\vec{\cE}$, for each $\cE \in \vec{\cE}$, $(\cE,t-1)$ is not tainted.
Using the induction hypothesis, we can relate $\loss(\cE)$ to $\aloss_{t-1}(\cE)$.
The proof is completed by summing the following bound over all $\cE \in \vec{\cE}(\cT)$.

\begin{eqnarray*}
\aloss_{t-1}(\cE) \leq (1+\tau_{t-1}) \loss(\cE) + \delta_{t-1} \width(\cE) \\
\Longleftrightarrow \loss(\cE) \geq \frac{\aloss_{t-1}(\cE)}{1+\tau_{t-1}} - \frac{\delta_{t-1} \width(\cE)}{1+ \tau_{t-1}}
\end{eqnarray*}
\end{proof}
Combining this claim with \Clm{delta} gives:


\begin{eqnarray*}
\Delta_t(\cT) & \leq &(1+\tau_t)|\cL^{out}(\cT)| -\tau_t \cdot \outside(\vec{\cE}(\cT)) + [1-K_t(1+\tau_t)]\cdot\aloss_{t-1}(\vec{\cE}(\cT))\\
&& +\errorterm_2(\cT)+(1+\tau_t)\errorterm_3(\cT)+K_t\delta_{t-1} \width(\cT). 
\end{eqnarray*}
We use $\errorterm_4(\cT)$ and $\errorterm_5(\cT)$ to eliminate $|\cL^{out}|$ and $\aloss_{t-1}(\vec{\cE}(\cT))$.
\begin{eqnarray}
\Delta_t(\cT) & \leq &(1+\tau_t)[\mu_t\cdot\outside(\vec{\cE}(\cT)) + \errorterm_4(\cT)] -\tau_t \cdot \outside(\vec{\cE}(\cT)) + \nonumber \\
&&[1-K_t(1+\tau_t)]\mu^{-1}_t[(1-\mu_t)\cdot\outside(\vec{\cE}(\cT)) + \errorterm_5(\cT)]
 +\errorterm_2(\cT)+(1+\tau_t)\errorterm_3(\cT)+K_t\delta_{t-1} \width(\cT) \nonumber \\
& = & [(1+\tau_t)\mu_t - \tau_t + (1-K_t(1+\tau_t))(\mu^{-1}_t - 1)] \outside(\vec{\cE}(\cT)) \nonumber \\
& & +\errorterm_2(\cT)+(1+\tau_t)\errorterm_3(\cT) + (1+\tau_t)\errorterm_4(\cT) + [1-K_t(1+\tau_t)]\mu^{-1}_t \errorterm_5(\cT) + K_t\delta_{t-1} \width(\cT) \label{eq:errors}
\end{eqnarray}
Our focus is on the coefficient of $\outside(\vec{\cE}(\cT))$ and the parameters will be chosen to ensure
this is negative.
\begin{eqnarray*}
& & (1+\tau_t)\mu_t - \tau_t + [1-K_t(1+\tau_t)](\mu^{-1}_t - 1) \\
& = & (1+\tau_t)\mu_t - (1+\tau_t) + K_t(1+\tau_t) + [1-K_t(1+\tau_t)]\mu^{-1}_t \\
& = & \mu^{-1}_t[ (1+\tau_t)\mu^2_t - (1+\tau_t)(1-K_t) \mu_t - K_t(1+\tau_t) + 1] \\
& = & \mu^{-1}_t\{ (1+\tau_t)[\mu^2_t - (1-K_t)\mu_t - K_t] + 1\}
\end{eqnarray*}
Setting $\mu_t = (1-K_t)/2$ (which is positive since $K_t < 1$) minimizes the inner quadratic. The condition that this expression
is non-positive is equivalent to the following.
\begin{align*}
(1+\tau_t)[(1-K_t)^2/4 - (1-K_t)^2/2 - K_t] \leq -1 
\Longleftrightarrow (1+\tau_t) \geq \frac{4}{(1+K_t)^2}
\end{align*}
%
%
%
%
%
%
For $t=1$, we have $K_t=0$ and the requirement is $\tau_t \geq 3$.  For $t \geq 2$,
$K_t=1/(1+\tau_{t-1})$ and the requirement becomes:
\[
\tau_t \geq \frac{4\tau_{t-1} + 3\tau^2_{t-1}}{(2+\tau_{t-1})^2}
\]
Guessing a solution of the form $\tau_t = C/t$, one verifies that $\tau_t=4/t$ works.  The
righthand side becomes

\begin{eqnarray*}
\frac{\frac{16}{t-1}+\frac{48}{(t-1)^2}}{(2+\frac{4}{t-1})^2} = \frac{16(t-1)+48}{(2(t-1)+4)^2} = \frac{4t+8}{(t+1)^2} \leq \frac{4t+8+\frac{4}{t}}{(t+1)^2} = \frac{4}{t}. 
\end{eqnarray*}
%

Thus the coefficient of $\outside(\vec{\cE}(T))$ in \Eqn{errors} is nonpositive.  To further simplify the bound
note that $K_t = (t-1)/(t+3) \leq 1$	and  $(1-K_t(1+\tau_t))/\mu_t  =2 (1-K_t(1+\tau_t))/(1-K_t) \leq 2$.
Thus from \Eqn{errors} we deduce: 
\begin{eqnarray*}
\Delta_t(\cT)
 &\leq & \errorterm_2(\cT)+(1+\tau_t)\errorterm_3(\cT)+(1+\tau_t)\errorterm_4(\cT)+ \errorterm_5(\cT)(1-K_t(1+\tau_t))/\mu_t+ K_t\delta_{t-1} \width(\cT) \\
 &\leq & \errorterm_2(\cT)+ 5\errorterm_3(\cT)+ 5 \errorterm_4(\cT)+ 2\errorterm_5(\cT)+ (\delta_{t}-\delta_1) \width(\cT).
\end{eqnarray*}
Summing over $\cT$ completes the proof of the claim.
\smallskip

{\bf Remark.} The reader should note that in the proof of \Clm{primary}, 
the only thing used about the algorithm $\terminalbox$ is \Lem{terminal}. The
exact details of \emph{how} the terminal chain was chosen are not used. Similarly, the particular choice of $\delta_1$ was
not used; we only use the fact that $\delta_t = t\delta_1$. 
The assumptions on the random seed are just Assumption 1 and that fact that $(\cB,t)$ is not tainted.
The specifics of the definition of tainted do no enter this proof.
All the calculations were basically algebraic manipulations (most details about the algorithm
have been pushed into the secondary error terms).
We mention this because our improved algorithm 
will be obtained by modifying $\terminalbox$ and changing $\delta_1$ (but not $\tau_t$ or $\mu_t$).
In analyzing the quality of approximation in the second
algorithm, we can reuse the analysis of this section.
%
%
%
%
\subsection{Bounding the secondary error terms}
\label{subsec:secondary analysis}

In this section we prove \Clm{secondary}.
We examine each of the error terms separately and show that each one contributes at most $\frac{1}{5\log(n)}=\frac{\delta_1}{5}$ to the secondary error, and so the total secondary error is at most $\delta_1$.

\begin{claim}
\label{clm:et1+et2}
If  $(\cB,t)$ is not tainted, $|\errorterm_1| \leq \taint\width(\cB)$ and $\sum_\cT \errorterm_2(\cT) \leq 2\taint \width(\cB)$.
\end{claim}

\begin{proof}
The first inequality follows immediately from the definition of (non)tainted boxes: $|\errorterm_1| = |\alis_t(\cB)-|Good_t(\cB)|| \leq \taint \width(\cB)$.

For the second inequality,
recall that $\vec{\cC}$ is the critical chain for $\cB$ and $\vec{\cC}(\cT)$ is the portion of the critical chain in terminal box $\cT$.
Let $\vec{\cC}^*\subset \vec{\cC}$ be the set of critical boxes $\cC \in \vec{\cC}$ 
such that $(\cC,t-1)$ is tainted.
Since $(\cB,t)$ is not tainted, $\sum_{\cC \in \vec{\cC}^*} \width(\cC) \leq \taint \width(\cB)$.
We therefore have:

\begin{eqnarray*}
\sum_\cT \errorterm_2(\cT) & = & \sum_{\cC \in \vec{\cC}} \alis(\cC)-|Good_{t-1}(\cC)|\\
& \leq & \sum_{\cC \in \vec{\cC}} |\alis(\cC)-|Good_{t-1}(\cC)||\\
& \leq	& \sum_{\cC \in \vec{\cC}^*}\width(\cC) + \sum_{\cC \in \vec{\cC}-\vec{\cC}^*} \taint\width(\cC)\\
& \leq &  \taint\width(\cB) + \taint\width(\cB) \ =  2\taint \width(\cB).
\end{eqnarray*}
\end{proof}
   
Since $\taint=1/10\log(n)$, both 
$|\errorterm_1|$ and  $\sum_{\cT \in \vec{\cT}} \errorterm_2(\cT)$ are  bounded above by $\width(\cB)/(5 \log n)$.


We move on to the error term $\errorterm_3(\cT)= |\cL^{in}(\cT)| - \sum_{\cE \in \vec{\cE}(\cT)} |\cE \cap \cL|.$
%
\begin{claim}
\label{clm:et3} 
\begin{eqnarray*}
\sum_{\cT \in \vec{\cT}} \errorterm_3(\cT)
& \leq & (\gridprec + \taint) \width(\cB).
\end{eqnarray*}
\end{claim}
\begin{proof} Focus on a single terminal box $\cT$.
By the grid approximation lemma (\Lem{grid-approx}), we can choose the chain $\vec{\cD}(\cT)$ so that it misses at most
$\gridprec\width(\cT)$ points of $\cL^{in}(\cT)$. (Note that $\vec{\cE}$ is the subsequence consisting of non-tainted boxes.) Let the total width of the tainted boxes in $\vec{\cD}(\cT)$ be denoted as $b(\cT)$.
Hence, $\errorterm_3(\cT) \leq \gridprec\width(\cT) + b(\cT)$. Summing over all $\cT$, $\sum_{\cT \in \vec{\cT}} \errorterm_3(\cT) \leq \gridprec\width(\cB) + \sum_{\cT \in \vec{\cT}} b(\cT)$.
Since $(\cB,t)$ is not tainted, the latter sum is at most $\taint\width(\cB)$. 
\end{proof}

Substituting  $\gridprec = 1/(C_1\log n)^3$ and $\taint = 1/(10\log n)$,
gives $\sum_{\cT \in \vec{\cT}} \errorterm_3(\cT) \leq \width(\cB)/(5\log n)$.


Next we consider the error term $\errorterm_4(\cT) = |\cL_{out}(\cT)| - \mu_t\cdot\outside(\vec{\cE}(\cT))$.  

\begin{claim}
\label{clm:et4}
\begin{eqnarray*}
\sum_{\cT \in \vec{\cT}} \errorterm_4 (\cT)& \leq & \Big(2\gridprec+\frac{4\gamma \log n}{\rho}\Big) \width(\cB) \leq \width(\cB)/(5\log n).
\end{eqnarray*}
\end{claim}

The main part of the claim is \Lem{et4}, which is formulated more generally to allow for its reuse later.  
Recall that in \Sec{terminal}, we constructed a tree associated with calls to $\terminalbox_t(\cdot,\cB)$
whose leaves were the terminal boxes and whose intermediate nodes were boxes that occur along some execution of $\terminalbox_t(\cdot,\cB)$.
Each intermediate node was also associated to a unique splitter selected to split the associated box.   
These splitters make up the interior increasing point sequence $\cP^{\circ}(\vec{\cT})$ associated with the terminal
chain $\vec{\cT}$.  For each $P \in \cP^{\circ}(\vec{\cT})$ we defined
$\cH(P)$ to be the box in the splitter tree that was split by $P$, and noted that by Assumption 1,
$P$ is $(\mu_t,2\gamma \width(\cH(P)))$-safe for $\cH(P)$.  

In the improved algorithm to be presented in \Sec{improved}, the parameter $\gamma$ will not be fixed, and the splitter selected for different boxes may
be selected based on a different $\gamma$.  To anticipate this we introduce the notation $\gamma(P)$ to be the value of $\gamma$ that was used
when $P$ was selected as splitter for $\cH(P)$ and so Assumption 1 gives that $P$ is a $(\mu_t,2\gamma(P) \width(\cH(P)))$-splitter for $\cH(P)$.
Nothing in the previous analysis relied on $\gamma$ being fixed (indeed, this is the first time in the proof of \Thm{basic correctness}  that $\gamma$ has been 
used.)

\begin{figure*}[tb]
  \centering
    \includegraphics[width=0.3\textwidth]{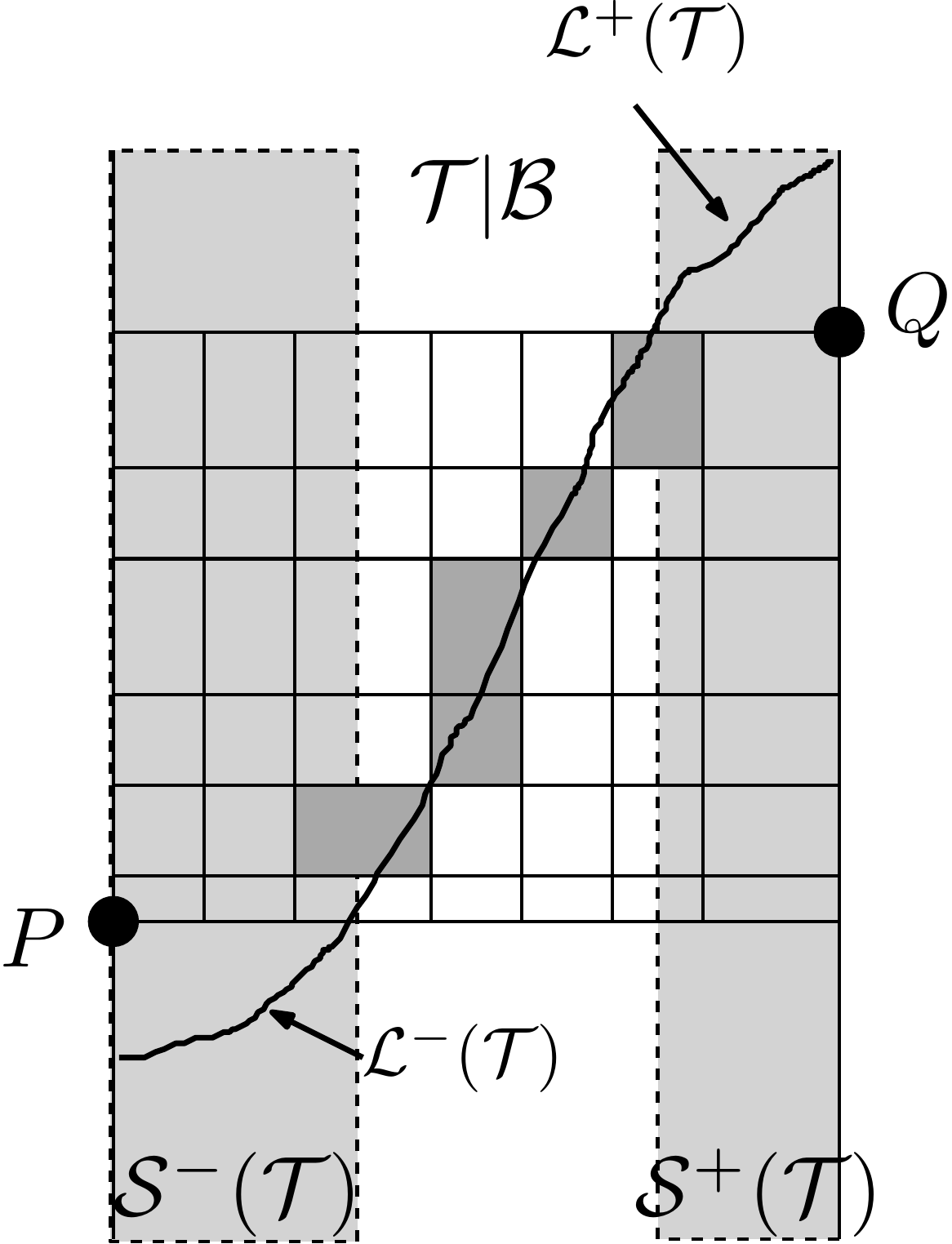}
    \caption{The LIS $\cL$ is depicted as a monotonically increasing freeform curve. The boxes in $\vec{\cE}$ are colored dark gray, 
    and the strips $\cS^-(\cT)$, $\cS^+(\cT)$ are in light gray.
}\label{fig:et4}
\end{figure*}

\begin{lemma}
\label{lem:et4}
\begin{eqnarray*}
\sum_{\cT \in \vec{\cT}} \errorterm_4(\cT) & \leq & 2\gridprec \width(\cB) + 4\sum_{P \in \cP^{\circ}(\vec{\cT})}\gamma(P)\width(\cH(P)).
\end{eqnarray*}
\end{lemma}
\begin{proof}
Looking at the definition of $\errorterm_4(\cT)$, 
our first goal is to obtain a lower bound on 
$\outside(\vec{\cE}(\cT)) = \beta(\cT)-\sum_{\cE \in \vec{\cE}(\cT)}|\cE \cap \cF|$, the number of $\cF$-points in the strip $\strip{\cT}{\cB}$ that lie outside of $\vec{\cE}$.

Note that $\cL^{out}(\cT)$ can be split
into the set $\cL^{-}(\cT)$ of points that lie below $\cT$ and $\cL^{+}(\cT)$ of points that lie above $\cT$.
Refer to \Fig{et4}.
Let $\cS^-(\cT)$ be the smallest strip starting from the left edge of $\strip{\cT}{\cB}$ that contains all of $\cL^{-}(\cT)$ (so $\cS^-(\cT)$ is
either empty, or has its rightmost edge defined by the largest $x$ for which $F(x) \in \cL$ and lies below $\cT$).
Similarly, let $\cS^+(\cT)$ be the smallest strip that ends at the right edge of $\strip{\cT}{\cB}$ and contains all of $\cL^+(\cT)$.
Since $\cL$ is increasing, $\cS^-(\cT)$ and $\cS^+(\cT)$ are disjoint. 

The overlap of $\cS^-(\cT) \cup \cS^+(\cT)$ with $\vec{\cE}(\cT)^{\cup}$ is small.
We claim that only the leftmost box $\cE_L$ of
$\vec{\cE}$ can overlap $\cS^-(\cT)$ (and a similar argument shows that only the rightmost box $\cE_R$ can overlap $\cS^+(\cT)$). 
The right edge of $\cS^-(\cT)$ occurs
at the rightmost point of $\cL^-$. The right edge of $\cE_L$ must be to the right of that because
$\cE_L$ contains at least one point in $\cL^{in}$.
Hence,
\begin{eqnarray*}
\outside(\vec{\cE}(\cT)) & \geq & |(\cS^-(\cT)\cup \cS^+(\cT)) \cap \cF| - |(\cE_L \cup \cE_R) \cap \cF|\\
& \geq & |(\cS^-(\cT)\cup \cS^+(\cT)) \cap \cF| - 2\gridprec \width(\cT),
\end{eqnarray*}
where the second inequality 
uses the fact that each of the grid strips of $\strip{\cT}{\cB}$ have width at most $\gridprec \width(\cT)$.
We now have:

\begin{eqnarray*}
\errorterm_4(\cT) 
& = & |\cL^{out}(\cT)| - \mu_t\cdot\outside(\vec{\cE}(\cT))\\
 & \leq & |\cL^-(\cT)|+|\cL^+(\cT)| - \mu_t(|(\cS^-(\cT)\cup \cS^+(\cT)) \cap \cF| - 2\gridprec \width(\cT))\\
& \leq &(|\cL^-(\cT)|-\mu_t|\cS^-(\cT) \cap \cF|)+(|\cL^+(\cT)| - \mu_t| \cS^+(\cT) \cap \cF|) + 2\mu_t\gridprec \width(\cT).
\end{eqnarray*}

We now bound the first two terms in the final expression.
Let $P$ be the bottom left point of $\cT$ and $Q$ be the top right point of $\cT$.
All of the points in $\cL^-(\cT)$ are in violation with $P$. (Refer to \Fig{et4}.)	  Since
$P$  is a $(\mu_t,2\gamma(P) \width(\cH(P)))$-splitter for $\cH(P)$, there are at most
$\mu|\cS^-(\cT) \cap \cF| + 2\gamma(P) \width(\cH(P))$ points of $\cF$ that are in violation with $P$
and so $|\cL^-(\cT)| - \mu|\cS^-(\cT) \cap \cF| \leq 2\gamma(P) \width(\cH(P))$.
Similarly  $|\cL^+(\cT)| - \mu|\cS^+(\cT) \cap \cF| \leq 2\gamma(Q) \width(\cH(Q))$. 
Thus:

\begin{eqnarray*}
\errorterm_4(\cT) 
& \leq &2\gamma(P)\width(\cH(P))+2\gamma(Q)\width(\cH(Q)) + 2\mu_t\gridprec \width(\cT).
\end{eqnarray*}

Finally, we sum the inequality in the lemma over all terminal boxes $\cT$. 
Observing that $\mu_t \leq 1$ and  that each point $P$ of the point sequence
$\cP^{\circ}(\vec{\cT})$ gets included twice (once as a bottom left point and once as an upper right point), we
get:
\begin{eqnarray*}
\sum_{\cT \in \vec{\cT}^w} \errorterm_4(\cT) & \leq & 2\gridprec \width(\cB) + 4\sum_{P \in \cP^{\circ}(\vec{\cT})}\gamma(P) \width(\cH(P)).
\end{eqnarray*}
\end{proof}

\begin{proof}[of \Clm{et4}]
We use the previous lemma with $\gamma(P)=\gamma$ for all $P$.
For a given level of the splitter tree the sum of $\width(\cH)$ is equal to $\width(\cB)$ so the total sum is
bounded by $\width(\cB)d$ where $d$ is the depth of the splitter tree.  Since each splitter is $\rho$-balanced, the width of a box
at depth $d$ in the splitter tree is at most $n(1-\rho)^d \leq ne^{-\rho d}$. This must be at least 1, so $d \leq (\log n)/\rho$. 
Thus we have:
\begin{eqnarray}
\label{eterm 4 bound}
\sum_{\cT \in \vec{\cT}} \errorterm_4(\cT) & \leq & \Big(2 \gridprec + \frac{4\gamma (\log n)}{\rho}\Big)\width(\cB).
\end{eqnarray}
Applying $\rho = C_1/\log n$ and $\alpha = \gamma = (1/C_1\log n)^3$, we get the final inequality.
\end{proof}
Finally, we consider $\errorterm_5(\cT)= \mu_t\cdot \aloss_{t-1}(\vec{\cE}(\cT)) - (1-\mu_t)\cdot \outside(\vec{\cE}(\cT))$. Let us review the basic terms. We have a terminal box $\cT$ inside $\cB$ andgrid $\Gamma(\cT)$. $\vec{\cD}(\cT)$ is the grid-chain with the largest number of points in $\cL$.
Tainted boxes and those containing no point of $\cL$ are removed from $\vec{\cD}(\cT)$ to get
the chain $\vec{\cE}(\cT)$. 

\begin{proposition} \label{prop:et5} $\errorterm_5(\cT) \leq \mu_t\cdot\aloss_{t-1}(\vec{\cD}(\cT))
- (1-\mu_t)\cdot \outside(\vec{\cD}(\cT))$.
\end{proposition}

\begin{proof} Since $\vec{\cE}(\cT)$ is a subsequence of $\vec{\cD}(\cT)$,
$\aloss_{t-1}(\vec{\cE}(\cT)) \leq \aloss_{t-1}(\vec{\cD}(\cT))$ and 
$\outside(\vec{\cE}(\cT)) \geq \outside(\vec{\cD}(\cT))$.
\end{proof}

%
\begin{claim}
\label{clm:et5}
For every terminal box $\cT$, if $\cT$ is narrow then $\errorterm_5(\cT) \leq 0$ and if $\errorterm_5(\cT)$ is wide then
$\errorterm_5(\cT) 
 \leq  4\rho\width(\cT) \leq \frac{1}{5 \log(n)} \width(\cT)$.

\end{claim}
\begin{proof} For convenience, we abbreviate $\vec{\cD}(\cT)$ by  $\vec{\cD}$. 
We apply \Lem{dichotomy} to the box $\strip{\cT}{\cB}$ with the strip decomposition given by the grid $\Gamma(\cT)$.
We get $(1-\mu_t)\cdot \outside(\vec{\cD}) \geq \mu_t\cdot|U|$, where $U$ is the set of $\vec{\cD}$-points that 
are $\mu_t$-unsafe in $\Gamma(\cT)$. 
Hence, $\errorterm_5(\cT) \leq \mu_t(\aloss_{t-1}(\vec{\cD}) - |U|)$.

%
%
%

Suppose $\cT$ is narrow (that is, $\width(\cT) \leq 1/\rho$).
By the definition of $\buildgrid$, the grid $\Gamma(\cT)$ has
$X(\Gamma(\cT))=X(\cT)$, and so each box $\cD \in \vec{\cD}$ has width 1. 
For any such box $\alis_{t-1}(\cD)=|\cD \cap \cF|$, so $\aloss_{t-1}(\cD) = 0$.
Hence, $\errorterm_5(\cT) \leq 0$.
%

Suppose $\cT$ is wide (that is, $\width(\cT) > 1/\rho$).
We can crudely bound $\aloss_{t-1}(\vec{\cD}) \leq |\vec{\cD}^\cup \cap \cF|$,
so $\errorterm_5 \leq \mu_t |S|$, where $S$ is the set of $\cD$-points
that are $\mu_t$-\emph{safe} in $\Gamma(\cT)$.
%
%
%

Every box in $\vec{\cD}$ has
width at most $\gridprec \width(\cT) \leq \gamma \width(\cT)$ (since $\vec{\cD}$ is a $\Gamma(\cT)$-chain).
So any $\mu_t$-safe point in $\Gamma(\cT)$ is $(\mu_t,\gamma \width(\cT))$-safe.  $\cT$ is a wide terminal box and
$\findsplitter_t(\cT,\cB,\mu_t,\gamma \width(\cT), \rho)$ failed to find a splitter. By Assumption 1 (reliability of $\findsplitter)$
and \Cor{fails}, the total number of nondegenerate $(\mu_t,\gamma\width(\cT))$-safe indices for $\strip{\cT}{\cB}$
is at most $3\rho\width(\cT)$.  Therefore, $|S| \leq 3\rho\width(\cT) + 1 \leq 4\rho\width(\cT)$ (where the last
inequality follows from the fact that $\cT$ is wide), completing the proof.
%
\end{proof}

We have thus shown that each of the five error terms is bounded by $\width(\cB)/5\log(n)$ and so the
total secondary error is $\width(\cB)/\log(n)=\delta_1 \width(\cB)$.  
This completes the proof of \Clm{secondary} the approximation error bound for the first algorithm.

\section{The improved algorithm}
\label{sec:improved}

The previous algorithm had running time $(\log n)^{O(1/\taupar)}$ where $\tau$ is the primary approximation
parameter.  
In this section we show how to modify the algorithm so that the running time is
$(1/\deltapar\taupar)^{O(1/\taupar)} (\log n)^{c}$ for some absolute constant $c$, where $\deltapar$ is the secondary error parameter, thereby proving \Thm{alg2}.   
The improved running time is obtained by making some subtle (and somewhat mysterious) changes to the
basic algorithm.  Formally the changes are  minor, and we will start by completely describing the revisions made to the
basic algorithm without
motivating the reasons for the changes.  Then we provide an informal discussion of how we arrived at these
changes.  (The reader may wish to read this discussion before reading the formal description.) 
Finally we prove the properties of the improved algorithm.

\subsection{Formal description of the algorithm}
\label{sec:improve-alg}

The basic algorithm consists of 6 main procedures,$\basicmain$, $\approxlis$, $\classify$, $\critbox$, $\terminalbox$ and $\gridchain$,
and uses two additional procedures $\findsplitter$ and $\buildgrid$.
The new algorithm has the same structure, with the only changes being that  the internal structure of
$\terminalbox$ is modified, and some  auxiliary parameters are changed, and some new parameters are added.
The change in parameter definitions is done by replacing $\basicmain$ by $\improvedmain$.

In the procedure below $C_2$ is a sufficiently large constant.
The key new parameter is $\errorcont$, called the {\em error controller}.  Most of the other parameters are defined in terms of $\errorcont$.  Increasing this parameter reduces the error, but also increases the
running time.  The running time of the algorithm will turn out to be $(1/\errorcont)^{O(\tau)}(\log n)^c$.
This parameter did not appear explicitly in 
the basic algorithm, but implicitly it was set to 
$\Omega(\log n)$.  In this algorithm we take it to be $\Omega(1/(\taupar\deltapar))$, which is what allows us
to make the exponent of $\log n)$ a constant independent of $\taupar$ and $\deltapar$.

\vspace{-10pt}
\begin{center}
	\fbox{\begin{minipage}{\columnwidth}  
$\improvedmain(n,f,\taupar,\deltapar)$\\
Output: Approximation to $\lis_f$.
\begin{enumerate}
\item Fix global parameters (unchanged throughout algorithm).
%
\begin{center}
  \begin{tabular}{ | l || c | c |}
    \hline
    {\bf Name} & {\bf Symbol} & {\bf Value} \\ \hline
		Maximum level & $t_{\max}$ & $\lceil 4/\taupar \rceil$ \\ \hline
    Error controller & $\errorcont$ & $\max(C_2, t_{\max}/\deltapar)$ \\ \hline
    Sample size parameter & $\ssz$ & $\sszval$ \\ \hline
    Grid precision parameter & $\gridprec$ & $\gridprecval$ \\ \hline
    Width threshold & $\omega$ & $1/\gridprec$ \\ \hline
    Tainting parameter & $\taint$ & $\taintval$ \\ \hline
    Primary splitter parameter & $\mu_r$ & $\frac{2}{r+3}$ \\ \hline
    Secondary splitter parameter & $\gamma_j$ & $16^j\alpha/(\log n)^4$ \\ \hline
    Splitter balance parameter & $\rho_j$ & $(\gamma_j)^{1/4}=2^j\alpha^{1/4}/\log n$ \\ \hline

  \end{tabular}
\end{center}

\item Let box  $\cU$ be the  box $[1,n] \times [1,\maxval]$.
\item  Return $\approxlis_{t_{\max}}(\cU)$.
\end{enumerate}
\end{minipage}}
\end{center}

Comparing $\basicmain$ and $\improvedmain$ one sees that $t_{\max}$ and the function $\mu_r$
are unchanged.    The parameters $\ssz$, $\gridprec$, $\omega$ and $\taint$ are
now polynomials (or inverse polynomials) in $\errorcont$ rather than $\log(n)$.  
The parameters $\gamma$ and $\rho$ are now replaced by functions
$\gamma_j$ and $\rho_j$.  This reflects the following key change in the algorithm:
In the basic algorithm, the procedure $\terminalbox$ terminates the first time that $\findsplitter$ fails to
find a splitter, or when the width of the box dropped below the threshold $\omega$.  
In the improved algorithm, this is not the case.   We view the improved
$\terminalbox$ as proceeding in phases where a phase
ends whenever $\findsplitter$ fails.  The index $j$
keeps track of the {\em phase number} and is incremented each time $\findsplitter$ fails to find a splitter
and the values of $\gamma$ and $\rho$ are increased to $\gamma_j$ and $\rho_j$ 
and $\findsplitter$ is repeated. Increasing $\gamma$ means that a splitter
is allowed to have more violations while increasing $\rho$ means that the we require any selected splitter to be 
closer to the center of the interval. Thus as the phases proceed,  one of the splitter 
conditions is relaxed and the other is made more stringent.
In the basic algorithm, there was a fixed parameter $\omega$  used in $\terminalbox$ to provide
a threshhold to determine whether the current box was ``narrow'' or not and to provide one of the criteria
for the procedure to halt.  In the improved $\terminalbox$ there is a variable parameter $\theta$ and the only stopping
criterion for $\terminalbox$ is that the current box has width below $\theta$. The 
parameter $\omega$ is used for the initial value for $\theta$. 
The value of $\theta$ may increase
at the beginning of each new phase.  If $\findsplitter$ fails on box $\cT$ in phase $j$, then $\theta$ is reset to
the maximum of its present value and $ \gamma_j \width(\cT)/\gridprec$.

The following auxiliary parameters are not used explicitly by the algorithm, but appear in the analysis:

\begin{center}
  \begin{tabular}{ | l || c | c |}
    \hline
    {\bf Name} & {\bf Symbol} & {\bf Value} \\ \hline
    Base additive error & $\delta_1$ & $\deltaval$ \\ \hline
    $t$-level additive error & $\delta_t$ & $t\delta_1$ \\ \hline
    Taint probability & $\taintprob$ & $\frac{1}{C_2}\gridprec^5$    \\ \hline
  \end{tabular}
\end{center}

Below is the improved version of $\terminalbox$.

\vspace{-10pt}
\begin{center}
	\fbox{\begin{minipage}{\columnwidth}
$\terminalbox_t(x,\cB)$  
\algcomment{$x \in X(\cB)$}\\
Output: subbox $\cT$.\\
\begin{enumerate} 
\item Initialize: $\cT \longleftarrow \cB$; $j \longleftarrow 0$;
$\theta \longleftarrow \omega$; $\gamma \longleftarrow \gamma_0$; $\rho \longleftarrow \rho_0$.
\item Repeat until $\width(\cT) \leq \theta$:
\begin{enumerate}
\item Run $\findsplitter(\cT,\cB,\mu_t,\gamma \width(\cT),\rho)$ which
returns the boolean variable $\splitterfound$ and the index
$splitter$.
\item If $\splitterfound={\tt TRUE}$ then
\begin{enumerate}
\item If $x \leq splitter$ then replace $\cT$ by the
box $\bx(P_{BL}(\cT),F(splitter))$.
\item If $x > splitter$ then replace $\cT$ by the
box $\bx(F(splitter), P_{TR}(\cT))$.
\end{enumerate}
\item else (so $\splitterfound = {\tt FALSE}$ and new phase starts)
\begin{enumerate}
\item $\theta \longleftarrow \max(\theta, \gamma \width(\cT)/\gridprec)$.
\item $j \longleftarrow j+1$; $\gamma \longleftarrow \gamma_j$; $\rho \longleftarrow \rho_j$. 
\end{enumerate}
\end{enumerate}
\item Return $\cT$.
\end{enumerate}
\end{minipage}}
\end{center}

We will prove the following property of $\approxlis$ with the improved version of $\terminalbox$:

\begin{theorem}
\label{thm:improved approxlis}  
Suppose $\approxlis$ is run with the global parameters set as in $\improvedmain$ using
the improved version of $\terminalbox$.
On input a box $\cB \subseteq \cU$ and an integer $t$, $\approxlis_t(\cB)$:
\begin{itemize}
\item  runs in time $(1/\errorcont)^{O(t)}(\log n)^c$, and
\item outputs a value that, with probability at least $3/4$, is a
{\em $(\tau_t,\delta_t )$}-approximation to $\lis(\cB)$, where:
\begin{eqnarray*}
\tau_t&=& \frac{4}{t} \\
\delta_t& = & \frac{t}{\errorcont}.
\end{eqnarray*} 
\end{itemize}
\end{theorem}

\Thm{alg2} follows immediately from \Thm{improved approxlis}:

\begin{proof}[of \Thm{alg2}]
$\improvedmain(n,f,\taupar,\deltapar)$ 
returns the output of $\approxlis_{t_{\max}}(\cU)$.  By \Thm{improved approxlis}, this gives a
$(4/t_{\max},t_{\max}/\errorcont)$-approximation to $\lis_f$ which runs in time $(1/\errorcont)^{O(t_{\max})}\log(n)^c$.  Using the fact that $t_{\max} = \lceil 4/\taupar \rceil$ and $\errorcont=t_{\max}/\deltapar$,
gives the desired running time and approximation error for $\improvedmain$.
\end{proof}

In the next subsection we provide some intuition for the improved algorithm.  The remainder of the section
is devoted to the proof of \Thm{improved approxlis}.  The proof is structured to parallel
that of \Thm{basic approxlis}.   We start with the running time analysis.  We revisit the $t$-critical tree,
highlighting the differences that arise in the improved algorithm.  Also, we revisit the randomness assumptions, modifying them as needed for the new analysis.  The primary and secondary error terms are defined
as in the basic algorithm, and the primary error analysis is unchanged.  The secondary error analysis reuses
much of the secondary error analysis from the basic algorithm,  but there are some significant
differences.

\subsection{Intuition for the improved algorithm}

Before analyzing the improved algorithm, we discuss two lines of thought that led to the improvement.
Our first line of thought gave the desired speed-up, but the algorithm was technically cumbersome. 
In trying to simplify the algorithm we found an alternative way to think about the improvement,
leading to the simplified version presented here.

Consider the main contributions  to the running time of the first algorithm.
The algorithm is recursive with recursion depth $t_{\max}=\Theta(1/\taupar)$.
Let $Time_t$ denote the running time of $\approxlis_t$ on instances of size at most $n$.  
Roughly,
$Time_t=R \times Time_{t-1}+S$ where $R$ is the number of recursive calls to $\approxlis_{t-1}$,
and $S$ which is the cost of all other computation performed by $\approxlis_t$, and this gives an overall running time of
$O(R^t S)$.  In the basic algorithm, $R$ and $S$ are both $(\log n)^{\Theta(1)}$, giving a time bound
of $(\log(n))^{O(1/\taupar)}$.  To improve the algorithm
we seek to reduce $R$.

What determines $R$?  In the main $\approxlis$ procedure we start by selecting a sample index set of size $\ssz=\log n^{\theta(1)}$.
For each sample index $x$ we eventually have to run $\gridchain_t(\cT)$ where $\cT$ is the terminal box
selected for $x$.  Within $\gridchain_t(\cT)$ recursive calls are made to $\approxlis_{t-1}$.  The number of
such calls is $1/\gridprec^{\Theta(1)}$.  Thus the total number of recursive calls
to $\approxlis_{t-1}$ is $\ssz (1/\gridprec)^{\Theta(1)}$.  Since $\ssz$ and $1/\gridprec$ are both
$(\log n)^{\Theta(1)}$, we would have to reduce each of them in order to make $R$ independent of $n$.

Note that the costs of $\terminalbox_t$ (which includes calls to $\findsplitter$), and the cost
of $\buildgrid$, which are both $((\log n))^{\Theta(1)}$, are part of the ``nonrecursive'' cost that contribute to
$S$ but not $R$. We can ignore them in this discussion.

Reducing $\ssz$ is fairly straightforward.
The choice of  $\ssz$ was originally made to guarantee that all pairs $(\cB,t)$ are untainted.
In other words, $\approxlis_t$ accurately estimates the number of $Good_t$ points
because $\ssz$ is large enough. In the proof of \Prop{assumption 2}, we take a union bound over all possible pairs $(\cB,t)$ (basically
all possible $\approxlis$ calls).  Since the number
of potential calls is $n^{\Theta(1)}$, we needed $\ssz$ to be $(\log n)^{\Omega(1)}$.
The first idea for reducing $\ssz$ is to relax the union bound so that it only pertains to the $(\log n)^{O(1)}$
calls that are actually made by the algorithm. This would allow us
to reduce $\ssz$ to $(\log\log n)^{O(1)}$. This
still has a dependence on $n$ (albeit a much better one). 

In both the original analysis, and the improved analysis sketched above, an execution is considered unsuccessful
unless {\em every} recursive call to $\approxlis_u$ that is actually performed gives a ``sufficiently good''
estimate for $|Good_u|$.  This is quite a stringent requirement, and in the improved algorithm 
we relax this criterion for success and allow a small
fraction of calls to $\approxlis_u(\cC)$ to give a poor estimate to $|Good_u(\cC)|$.
More precisely when we make a recursive call to $\approxlis(\cT)$ we will allow the probability 
that $\approxlis_t$ errs significantly on its
estimate of the number of $Good_t$ points to be $\taintprob$ (as defined above)
which is much larger than $n^{-\Omega(1)}$.
The challenge in doing is this is that errors might accumulate over different recursive levels
so that the estimate produced at the top level is inaccurate with high probability.  It turns out
this does not happen, and
the  recursive notion of tainted was defined as an analytical tool to prove this. 
This allows $\ssz$ to be some polynomial in $1/\tau$.

The more challenging problem is to deal with $\gridprec$, which is the grid precision parameter.
The grid constructed inside a terminal box has size ($\poly(1/\gridprec)$).  In order to perform the longest path
computation on the digraph $D(\Gamma|\cT)$ we 
make a recursive call to $\approxlis$  in each grid box.  
This leads to a contribution of $(1/\gridprec)^{O(t)}$
in the run time (since $t$ is the depth of the recursion), and since $\gridprec$ is $\Theta(1/\log(n))$
this is significant.   If $\gridprec$ could be taken to be a constant, this would give us the kind of run time
we are aiming for.

Unfortunately, the analysis of the first algorithm does not allow for such an improvement. 
The analysis of $\errorterm_5(\cT)$ forces the grid to be fine, specified as $\gridprec \leq \gamma$ (refer to the proof of \Clm{et5}). 
The parameter $\gamma$ decides the additive loss of a splitter, and the grid precision must be finer than this.
This is crucial for the application of \Lem{dichotomy}.

Our  analysis of the basic algorithm also requires that $\gamma$ be small.
When we analyze the
error term $\sum_{\cT} \errorterm_4(\cT)$, we obtain an upper bound of  $\gamma (\log n) \width(\cB)/\rho$ (refer to \Clm{et4}). 
This error term accounts for the points that are eliminated (perhaps incorrectly)
because they were in violation with some selected splitter.
To ensure that this term is small we need  $\gamma=\log(n)^{-\Theta(1)}$.  Since $\gridprec \leq \gamma$,
we get the same upper bound on $\gridprec$.

The key insight to overcoming this barrier is that the longest path computation that we perform on the grid is very similar to an exact LIS computation. 
It
can be formulated
as a weighted generalization of LIS on an index set of size $1/\gridprec$.  So rather than do an exact longest path
computation we could instead approximate the longest path computation by applying the ideas of our LIS
algorithm recursively.   This can be indeed be done.  However, 
in this form the algorithm becomes rather complicated and confusing since we have recursion arising in two ways: the basic
algorithm was already recursive, and now another recursion is piled on top.

Fortunately, by  unwinding this new layer of recursion, one can construct
an algorithm that avoids it altogether.  Indeed this new layer of recursion is implicitly
implemented simply by introducing the phases in $\terminalbox$ and adjusting some of the parameters dynamically.  To see how this arises, let us consider a recursive algorithm of
the type proposed above.  When we arrive at a terminal box (because we fail to find a splitter), rather
than do an exact longest path computation inside
of $\gridchain$ we recursively call an approximation algorithm.  In this recursive call, the index set size will be $1/\gridprec$
which is $(\log n)^{O(1)}$, and when we look for splitters, the $\gamma$ parameter will be $1/(\log(1/\gridprec))^{O(1)}$
which will be $(\log\log n )^{O(1)}$, which is significantly smaller than that of the basic algorithm.  
This suggests
that (rather than impose this additional recursive structure) when the algorithm fails to find a good splitter
for a particular
value of $\gamma$ we increase $\gamma$ and try again.  
In order to control all the error terms,
we have to increase $\rho$ and the width threshold $\theta$ (the minimum width that allows us
to stop and declare a terminal box). 
We think of $\terminalbox$ as acting in phases indexed by $j=0,1,\ldots$.  
Phase $j$ ends either because the width of the current box drops below $\theta$ (in which
case $\terminalbox$ terminates) or because $\findsplitter$ fails, in which case phase $j$ ends and phase $j+1$ begins
with $\gamma,\rho$ and $\theta$ all adjusted upwards.

The parameters $\gamma$ and $\rho$ are initially $1/(\log n)^{\Theta(1)}$ as in the original algorithm.
With the recursive view described above,  there is a rapid decrease in the parameters because
the size of the longest path subproblem being solved is polylogarithmic in the size of the initial problem.
This suggests that after $j$ phases the parameters should be the reciprocal of the $j$th iterated logarithm.
This works, but our analysis allows considerable flexibility in choosing the parameters. To simplify
the technical details,
we make a particularly simple choice:
we  double $\rho$ after each phase and multiply $\gamma$ by 16, always maintaining $\gamma=\rho^4$.  

The stopping condition for $\terminalbox$ is now changed.  Before we stopped the first time
$\findsplitter$ failed.  Now we continue until the width of the terminal box drops below the threshhold.

The above approach allows us to set $\gridprec = (1/t)^{O(1)}$.
Note that the term $O(\gridprec \width(\cT))$ still appears in $\errorterm_3(\cT)$ and $\errorterm_4(\cT)$,
leading to a secondary error of $t^{-\Omega(1)}$.

\subsection{Running time analysis of the improved algorithm}
We now turn to the  proof of \Thm{improved approxlis}. We begin with the run time
analysis, and mimic the anaylsis for the basic algorithm with some minor changes.

Let $A_t=A_t(n)$ be the running time of $\approxlis_t$  and
$C_t=C_t(n)$ be the running time of $\classify_t$ on boxes of width at most $n$.

As before  we use $P_i=P_i(n)$ to denote functions of the form $a_i ((\log n))^{b_i}$, where $a_i,b_i$ are constants that
are independent of $n$ and $t$.  We also use $Q_i=Q_i(\errorcont)$ to denote functions of $\errorcont$ of the form $c_i(\errorcont)^{d_i}$ where $c_i,d_i$ are constants.

\begin{claim} \label{clm:imp-time}
For all $t \geq 1$,
%
\begin{eqnarray*}
A_t & \leq & Q_1 C_t + Q_2. \\
C_t & \leq & C_{t-1}+ Q_4 A_{t-1}+P_1Q_3.
\end{eqnarray*}
\end{claim}

\begin{proof}
The first recurrence with $Q_1=\ssz$ is immediate from the definition of 
$\approxlis_t$.   The function $Q_2$ is an upper bound on the cost of operations excluding calls to $\classify$.

For the second recurrence, the final recursive call to $\classify_{t-1}$  gives the $C_{t-1}$ term.
The rest of the cost comes from  $\critbox_t$ which invokes $\terminalbox_t$, which involves several iterations
where the cost of each iterations is dominated by the cost of $\findsplitter$.  Each iteration reduces the size of the
box $\cT$ by at least a $(1-\rho_0)$ factor so the number of iterations is at most $\log n/\rho_0$.  The cost of $\findsplitter$ is $(\log n/\errorcont)^{O(1)}$
so the cost of $\terminalbox_t$ is included in the term $P_1Q_3$. 
$\critbox_t$ then calls $\gridchain_t$.  This involves building a grid of
size $(1/\errorcont)^{O(1)}$ and making one call to $\approxlis_t$ for each grid box, which accounts for the term $Q_4 A_{t-1}$.  
$\gridchain_t$ finds a longest path in  the grid digraph, which can be absorbed into the $P_1Q_3$ term.  
\end{proof}

\begin{corollary} \label{cor:imp-time}
For given input parameters $\taupar, \deltapar$, both $A_{t_{\max}}$ and $C_{t_{\max}}$ are in $(1/(\taupar\deltapar))^{O(1/\taupar)}(\log n)^{O(1)}$.
\end{corollary}

\begin{proof} Using the recurrence for $A_{t-1}$ to eliminate $A_{t-1}$ from
the recurrence for $C_t$, which gives a linear recurrence for $C_t$ in terms of $C_{t-1}$
whose solution has the form $C_t \leq P_5Q_5(Q_1Q_4+1)^t$.  This leads also to  $A_t \leq P_6Q_6(Q_1Q_4+1)^t$,  which are both $(\log n)^{O(1)}(\errorcont)^{O(t)}$.
Since $t_{\max} = O(1/\taupar)$ and $\errorcont = O(1/(\taupar\deltapar))$, the bound follows.
\end{proof}

\subsection{Revisiting the t-splitter tree and the terminal chain}
\label{sec:revised}

In preparation for the improved error analysis, we revisit the terminal tree, and replace Assumption 2 by a relaxed version.
Recall the $t$-splitter tree described in \Sec{terminal}.  We will construct a similar tree that also
reflects the phase structure of the improved algorithm.

A box-phase pair is a pair $(\cT,j)$.  The execution of $\terminalbox_t(x,\cB)$ generates a sequence $(\cT_0,j_0)=(\cB,0),(\cT_1,j_1),\ldots$
of box-phase pairs corresponding to each call of $\findsplitter$.
This is called the {\em $t$-trace} of $x$.  We define the {\em $t$-execution tree} for $\cB$
whose paths from the root correspond to the $t$-traces of all indices.

Classify each box-phase pair $(\cT,j)$ as {\em split} or {\em splitterless} depending on whether \linebreak
$\findsplitter_t(\cT,\cB,\mu_t,\gamma_j \width(\cB),\rho_j)$ 
succeeds or fails  to find a splitter.  For a split pair $(\cT,j)$, let $s_j(\cT)$ be the splitter found, 
and define the left-child and right-child of $(\cT,j)$ to  be  $(\bx(x_{BL}(\cT),F(s)),j)$ and $(\bx(F(s),x_{TR}(\cT)),j)$.  
If $(\cT,j)$ is splitterless, define its child to be $(\cT,j+1)$.   This defines a DAG on the box-phase pairs.

The subgraph of nodes reachable from $(\cB,0)$ is a rooted tree in which every node has 1 or 2 children (depending on whether
it is splitterless or split).  It is convenient to add an additional (splitterless) root $(\cB,-1)$ with unique child $(\cB,0)$.
We assign a number $\omega(\cT,j)$ to each splitterless pair $(\cT,j)$. We set $\omega(\cT,j)=\frac{\gamma_j}{\gridprec}\width(\cT)$ for $j \geq 0$,
and $\omega(\cB,-1)=\omega$.
Define $\theta(\cT,j)$ for every pair $(\cT,j)$ to be the maximum of $\omega(\cR,i)$
over all splitterless ancestors $(\cR,i)$ of $(\cT,j)$.  Observe that along any path from the root 
$\theta(\cT,j)$ is nondecreasing  
while $|\cT|$ is nonincreasing.  Truncate every path at the first node where
$|\cT| \leq \theta(\cT,j)$.   This is the $t$-execution tree for $\cB$.

For each $x \in X(\cB)$, define
$\Pi_E(x,\cB)$ to be the set of nodes in the $t$-execution tree whose box includes $x$ in its index set.
The following lemma is immediate from the (improved) definition of  $\terminalbox$.

\begin{lemma}
\label{lem:execution}
For every $x \in F^{-1}(\cB)$, $\Pi_E(x,\cB)$ 
is a root-to-leaf path in the tree and is equal to the $t$-trace
of $x$.  In particular, the box corresponding to the leaf that is reached is the terminal box that is returned
by $\terminalbox_t(x,\cB)$, and is equal to
$\vec{\cT}[x]$ (the unique box $\cT \in \vec{\cT}$ such that $x \in X(\cT)$).
\end{lemma}

For a box $\cH$ that appears in the $t$-execution tree, the set of nodes labeled by $\cH$ is a path and the sequence of
phase numbers increase by 1.  All of the nodes in the path except the last have exactly one child (which is its successor
in the path.)  The last occurence of $\cH$ is either a leaf or a split node.  We define $\phase(\cH)$ to be the highest phase
that $\cH$ reaches.   Box $\cH$ is a {\em multiphase box} if the associated path has two or more nodes
and is a {\em uniphase box} otherwise.  Note that if $\cH$ is a multiphase box then $(\cH,\phase(\cH)-1)$ is
unsplittable and lies on the execution path that reaches $(\cH,\phase(\cH))$. We state a simple proposition.

\begin{proposition} \label{prop:phase-box} For some positive $j$, let $(\cH,j)$ and $(\cH',j)$ be unsplittable nodes in the $t$-execution tree.
Then $X(\cH)$ and $X(\cH')$ are disjoint.
\end{proposition}

\begin{proof} Since the edges in the $t$-execution tree denote containment, if $X(\cH) \cap X(\cH') \neq \emptyset$, then
(wlog) $\cH$ is contained in $\cH'$. There is a root-to-leaf path passing through $(\cH',j)$ and then $(\cH,j)$.
Because $(\cH',j)$ is unsplittable, the subsequent node in this path must be $(\cH',j+1)$. But phases
cannot decrease as we go down this path, so $(\cH,j)$ cannot appear.
\end{proof}

We can construct a new tree by contracting all of the occurences of each box into a single node.  This gives a binary
tree of boxes and this is the $t$-splitter tree.  Each node in this tree is a box $\cH$.
The reader can readily verify that the $t$-splitter tree satisfies the  bulleted properties of $R(\cB)$ in
\Sec{terminal} and \Lem{terminal}.  As we already noted in formulating \Lem{et4},
the parameter $\gamma$  now depends on $P$.

Define $\Pi(x,\cB)$ to be the path in the $t$-terminal tree corresponding to $x$.  The non-terminal nodes of $\Pi(x,\cB)$
are precisely the split nodes in $\Pi_E(x,\cB)$ (the $t$-trace of $x$ defined earlier).  
The phase numbers of nodes along the path are non-decreasing. 
It is not hard to see that $\gamma_j$ (where $j$ indexes the phase) can never become too large. This is because that threshold $\theta$
increases as $j$ increases, and eventually this will exceed the width of the box $\cT$. 

For $j \geq 0$, let $d_j(x)$ be the number of non-terminal nodes in $\Pi(x,\cB)$ at phase $j$. 
Each of these nodes represents a successful splitter.
This number cannot be too large, since the width of the box decreases with each split.
We make these arguments formal.
%

\begin{prop}
\label{prop:phase} Let $j$ be a nonnegative integer and $x \in X(\cB)$ be such that $d_j(x) \neq 0$. 
\begin{enumerate}
\item $\gamma_{j} \leq 16\gridprec$ and $\rho_j \leq 2\gridprec^{1/4}$.
\item $d_j(x) \leq (1/\rho_j)^2$.
\end{enumerate}
\end{prop}

\begin{proof} Fix $j$ such that $d_j(x) \neq 0$. 
For simplicity of notation write $d$ for $d_j(x)$.
Let $\cH_1,\ldots,\cH_{d}$ ($d \geq 1$) be the sequence of non-terminal boxes in $\Pi(x,\cB)$ at
phase $j$.  By the relationship between the $t$-execution tree and the $t$-splitter tree,
$(\cH_1,j),\ldots,(\cH_d,j)$ is  a path in the $t$-execution tree and $(\cH_1,j-1)$ is the
parent of $(\cH_1,j)$.   Let $\theta_j$ be the width threshhold at the beginning of phase $j$.
We must have $\theta_j < \width(\cH_d)$ since otherwise one of the boxes $\cH_1,\ldots,\cH_d$ would
be a terminal box.

For the first part of the proposition, it is enough to prove the first inequality since
$\rho_j=\gamma_j^{1/4}$. 
For $j=0$, $\gamma_0 \leq \gridprec$ by definition.
For $j \geq 1$, by the definition of $\theta_j$ in $\terminalbox$,
$\theta_j \geq \width(\cH_1)\gamma_{j-1}/\gridprec$ which 
is at most $\width(\cH_1)$. Hence $\gamma_{j-1} \leq \gridprec$ and $\gamma_j =16 \gamma_{j-1} \leq 16 \gridprec$.

Now consider the second part of the proposition.
During phase $j$ each selected splitter is $\rho_j$-balanced and so $\width(\cH_d) \leq \width(\cH_1)(1-\rho_j)^{d-1} \leq \width(\cH_1)e^{-\rho_j(d-1)}$.

For $j=0$, we have $1 \leq \width(\cH_d) \leq n e^{-\rho_0(d-1)}$, so $d \leq 1 + (\ln n)/\rho_0$. 
This is at most $1/\rho_0^2$ by the definition of $\rho_0$. 

For $j \geq 1$, $\width(\cH_d)$
must be at least $\theta_j \geq \width(\cH_1) \gamma_{j-1}/\gridprec$.  Combining this with the
previous upper bound on $\width(\cH_d)$, we have
$e^{-\rho_j(d-1)} \geq \gamma_{j-1}/\gridprec = \rho_{j-1}^4/\gridprec = (\rho_j/2)^4/\gridprec$.  Solving this
final inequality for $d$ gives:
 
\begin{eqnarray*}
d \leq \frac{\rho_j+\ln (\gridprec(2/\rho_j)^4)}{\rho_j} \leq \frac{\rho_j + 4\ln(2/\rho_j)}{\rho_j} 
\leq \left(\frac{1}{\rho_j}\right)^2.
\end{eqnarray*}

To justify the final inequality, observe that $\rho_j \leq 2\gridprec^{1/4} =2/\errorcont \leq 2/C_2 \leq 1/16$.
Furthermore, for any $x \in (0,1/16)$, 
$x+4\ln(2/x) \leq 1/x$.
\end{proof}

Call a terminal box $\cT$ {\em narrow} or {\em wide} depending on whether $\width(\cT) \leq \omega$
or $\width(\cT)>\omega$.  

\begin{prop}
\label{prop:wide}
For every wide terminal box $\cT$ there is a multiphase box $\cH$ on the $t$-splitter path to $\cT$ such that
$\width(\cT) \leq \gamma_{\phase(\cH)-1} \width(\cH)/\gridprec$.
\end{prop}

\begin{proof}
Since $\cT$ is a terminal box, the execution of $\terminalbox$ for any $x \in X(\cT)$ follows the exact same
path in the $t$-execution tree and ends with $\cT$, and the final value
$\theta^*$ of the width threshhold is at least $\width(\cT)$.    
Since $\cT$ is wide, $\theta^* \geq \width(\cT)> \omega$, which implies that 
$\theta$ increased at least once during the execution of $\terminalbox$.   The parameter $\theta$ can only change
when a new phase begins. This happens when a box is found to be unsplittable, and $\theta$ only increases.   
Consider the  change of $\theta$ to $\theta^*$.   This corresponds to 
an  unsplittable pair $(\cH,i)$ on the $t$-execution path to $(\cT,\phase(\cT))$
such that $\theta^*=\omega(\cH,i)=\width(\cH) \gamma_i/\gridprec$.  We therefore have 
$\width(\cH) \gamma_i/\gridprec \geq \width(\cT)$.  Since $\gamma_i$ increases with $i$
and $\phase(\cH)-1$ is the largest phase for which $\cH$ is unsplittable, $\width(\cH) \gamma_{\phase(\cH)-1}/\gridprec \geq \width(\cH) \gamma_i/\gridprec \geq \width(\cT)$.
\end{proof}

\subsection{Revisiting tainted boxes} \label{sec:tainted}

Previously, the tainting parameter $\taint$ was $1/(10\log n)$ and the sample size $\ssz$ was $10(\log n)^2$.
The chance of (large) error in the estimate of $\approxlis_t(\cB)$ for $|Good_t(\cB)|$ was small
enough to use a union bound argument. We easily concluded that with high probability, no box was tainted.
%

Now, the parameters are $\ssz = \sszval$, $\taint = \taintval$, $\taintprob = \gridprec^5/C_2$,
and $\gridprec = \gridprecval$.
Since $\errorcont$ is only bounded below by a fixed constant $C_2$, there is a non-negligible
chance that a single call to $\approxlis_t(\cB)$ errs in its estimate.
A simple union bound does not work and our accounting needs to be more careful.  
We will prove by induction that the probability that box-level pair $(\cB',t')$ is tainted is small.

We remind the reader that $\cB$ has a chain of terminal boxes that spans $\cB$. Furthermore,
we have a grid $\Gamma(\cT)$ in each terminal box $\cT$.  A {\em spanning terminal-compatible grid chain} for $\cB$
is a chain of boxes that is obtained by selecting a spanning grid chain for each terminal box $\cT$
and concatenating them together.

\begin{lemma}
\label{lem:tainted}
Assume that the secondary random bits are fixed in a way such that Assumption 1 holds.
For all $(\cB,t)$, the probability with respect to the primary random bits that $(\cB,t)$ is tainted is at most  $\taintprob$,
where $\taintprob=\gridprec^t/C_2$ is the taint probability defined in \Sec{improve-alg}
\end{lemma}

\begin{proof}
We prove the result by induction on  $t$.
 For the base case, we note that
if $t=0$, then $(\cB,t)$ is a leaf of the instance tree and not tainted by definition.
Suppose $t\geq 1$.
By applying the definition of tainted, the inductive claim follows immediately from
the following two statements:

\begin{enumerate}
\item 
 $\pr[|\approxlis_t(\cB) - |Good_t(\cB)|| > \taint \width(\cB)]$  is at most $\taintprob/2$.
\item   
The probability that $\cB$ has a spanning terminal-compatible grid chain $\vec{\cC}$  such that the
total width of the boxes $\{\cC \in \vec{\cC}| \textrm{$(\cC,t-1)$ is tainted}\}$ is at least $\taint \width(\cB)$, is at most
$\taintprob/2$.
%
\end{enumerate}

By \Prop{hoeff2}. we obtain that
$\pr[|\approxlis_t(\cB)-|Good_t(\cB)|| > \taint \width(\cB)]$ is at most $2e^{-2\taint^2 \ssz} \leq 2e^{-\errorcont} \leq 1/\errorcont^{23} \leq \psi/2$, to prove the first statement.


The proof of the second statement involves more work (and  induction).
Consider a terminal box $\cT$.  By \Prop{buildgrid}, the 
grid $\Gamma=\Gamma(\cT)$ has at most $|X(\Gamma)| \leq 3/\gridprec$ and $|Y(\Gamma)| \leq 16/\gridprec^2$. The points
$|X(\Gamma)|$ divide the box  $\cT$ into at most $3/\gridprec$ grid strips. For wide $\cT$, each of these has width at most $\gridprec \width(\cT)$.
A grid strip $\cS$ of $\cT$ is said to be \emph{blue}
if there is a grid box $\cC$ in the strip, such that $(\cC,t-1)$ is tainted. 
Note that if $\cT$ is narrow, no grid strip of $\Gamma(\cT)$ is blue (since boxes of width $1$ are by definition
not tainted).  

The event whose probability we wish to upper bound in contained in the following event: the total width of blue grid strips
in $\cB$ is at least $\taint\width(\cB)$. This, in turn, is contained in the following event $E$:
the sum over all wide terminal boxes $\cT$ of $\cB$ of the width of blue grid strips of $\cT$ is at least $\taint\width(\cB)$.
Index the collection of all grid strips of terminal boxes of $\cB$
arbitrarily as $\cS_1, \cS_2, \cdots$. Define random variable $Y_i$ as follows.
If $\cS_i$ is blue, $Y_i = \width(\cS_i)/(\gridprec \width(\cB))$. Otherwise, $Y_i = 0$.

\begin{claim} \label{clm:yi} The random variables $Y_i$ are independently distributed in $[0,1]$.
The event $E$ is contained in the event $[\sum_i Y_i \geq \errorcont^3]$.
Furthermore, $\EX[\sum_i Y_i] \leq 1$. 
\end{claim}

\begin{proof} Suppose $\cS_i$ is contained in (wide) terminal box $\cT$.
The variable $Y_i$ is either $0$ or $ \width(\cS_i)/(\gridprec \width(\cB))$ $\leq \gridprec \width(\cT)/(\gridprec \width(\cB))
\leq 1$. All the grid strips of interest are disjoint,  so the random variables $Y_i$
are independent (because the primary random bits they depend on are disjoint).

Note that $\sum_i Y_i = \sum_{i: \cS_i \textrm{blue}} \width(\cS_i)/(\gridprec \width(\cB))$.
When $E$ occurs, $\sum_i Y_i \geq \taint \width(\cB)/(\gridprec \width(\cB)) \geq \errorcont^3$.

The lower bound on $\EX[\sum_i Y_i]$ is where the  induction appears.
The number of grid boxes lying in $\cS_i$ is at most $|Y(\Gamma(\cT))|^2 \leq 256/\gridprec^4$.  By the induction hypothesis,
for each such grid box $\cC$, the probability that $(\cC,t-1)$ is tainted is at most $\taintprob$.
By a union bound, the probability that $\cS_i$ is blue is at most $256 \taintprob/\gridprec^4$.
By linearity of expectation, $\EX[\sum_i Y_i] \leq (256 \taintprob/\gridprec^4)\sum_i \width(\cS_i)/(\gridprec \width(\cB))$
$\leq 256\taintprob/\gridprec^5 \leq 1$.
\end{proof}
We can now apply \Prop{cher} and deduce
$\Pr[\sum_i Y_i \geq {\errorcont}^3] < 2^{-\errorcont^3}$ $< 1/\errorcont^{23} \leq \taintprob/2$.
This completes the proof of \Lem{tainted}.
\end{proof}

\subsection{Error analysis of the improved algorithm}
\label{subsec:revised error}

The formal statement concerning the error of the improved algorithm is:

\begin{theorem}
\label{thm:improved correctness}
Let $\cB$ be a box in $\cU(f)$.
Assume that the random bits used satisfy Assumption 1.
For any $1 \leq t \leq t_{\max}$, if $(\cB,t)$ is not tainted:

\begin{eqnarray*}
\liserror_t(\cB) & \leq & \frac{4}{t}\loss(\cB)+\frac{\width(\cB)}{\errorcont}.
\end{eqnarray*}
\end{theorem}

\begin{proof}[of \Thm{improved approxlis}]  The claimed run time for
\Thm{improved approxlis} follows from \Cor{imp-time}.  By \Prop{assumption 1},
Assumptions 1 fails is $n^{O(\log(n)}$.  By \Lem{tainted}, the probability
that $(\cB,t)$ is tainted is at most $\taintprob$. So the probability that neither of these happens
is (easily) at least $3/4$.  So we apply
\Thm{improved correctness} to get the claimed bounds.
\end{proof}

So we are left to prove  \Thm{improved correctness}.
the improved algorithm has the same global structure as the original algorithm, and in particular
satisfies \Lem{terminal}, \Prop{narrow}, and \Prop{good}.  Assumption 1 still holds with probability $1-n^{-\Omega(\log n)}$.
We can break down the proof of \Thm{improved correctness} into proving \Clm{primary} (for the primary error terms)
and \Clm{secondary} (for the secondary error terms).
Furthermore, as noted in the remark at the end of \Sec{primary}, 
we can reuse \Clm{primary} and its proof as is.

Hence, it only remains to prove  \Clm{secondary} for our setting, which we restate below
for convenience.
%
%
\begin{eqnarray}
\nonumber
|\errorterm_1|+\sum_{\cT \in \vec{\cT}} (\errorterm_2(\cT) + 5 \errorterm_3(\cT)^+ + 5\errorterm_4(\cT)^+ +5 \errorterm_5(\cT)^+)& \leq & \delta_1\width(\cB)
= \frac{\width(\cB)}{\errorcont}.
\end{eqnarray}
%
%
%
We can use \Clm{et1+et2} and \Clm{et3} and their proofs from the basic analysis.
Using the new values $\taint=1/10\errorcont$ and $\gridprec=1/\errorcont^5$ gives a bound of $\width/5\errorcont$
on each of the contributions of $\errorterm_1$, $\errorterm_2$ and $\errorterm_3$.

To prove the bound on  $\errorterm_4(\cT) = |\cL^{out}(\cT)| - \mu_t \cdot \outside(\vec{\cE}(\cT))$, 
we use the
following variant of \Clm{et4}. 

\begin{claim}
\label{clm:et4imp}
\begin{eqnarray*}
\sum_{\cT \in \vec{\cT}} \errorterm_4 (\cT)& \leq & \width(\cB)/5\errorcont.
\end{eqnarray*}
\end{claim}

\begin{proof}
Recall that \Lem{et4} was proved in enough generality to apply to the present situation and gave us:
\begin{eqnarray}
\label{eq:et4}
\sum_{\cT \in \vec{\cT}} \errorterm_4(\cT) & \leq & 2\gridprec \width(\cB) + 4\sum_{P \in \cP^{\circ}(\vec{\cT})}\gamma(P)\width(\cH(P)).
\end{eqnarray}

To analyze the final sum we consider for each index $x \in X(\cB)$, the root-to-leaf path $\Pi(x)$ in the $t$-splitter tree.  Let $\cP(x)$ denote the
set of splitters encountered along that path.  Recall that for $j \geq 0$, $d_j(x)$ is the number of splitters in $\cP(x)$
that were found in phase $j$. Define 
$\tilde{\gamma}(x)=\sum_{P \in \cP(x)} \gamma(P)=\sum_{j \leq \phase(x)} \gamma_jd_j(x)$.  
The summation on the righthand side
of \Eqn{et4} is equal to $\sum_{x \in X(\cB)} \tilde{\gamma}(x)$.

Using the inequalities $d_j(x) \leq (1/\rho_j)^2$ and $\rho_j \leq 2\gridprec^{1/4}$, and using the
fact that $\rho_j$ is proportional to $2^j$,
\Prop{phase} we have:

\begin{equation*}
\tilde{\gamma}(x) \leq \sum_{j = 0}^{\phase(x)} \gamma_j/(\rho_j)^2=\sum_{j=0}^{\phase(x)} (\rho_j )^2 \leq 2 (\rho_{\phase(x)})^2 \leq 8\sqrt{\gridprec}.
\end{equation*}

Summing over $x \in X(\cB)$  and substituting into the above bound yields the upper bound 
$(2\gridprec+32\sqrt{\gridprec})\width(\cB) \leq 34\sqrt{\gridprec}\width(\cB) \leq \width(\cB)/5\errorcont$.
\end{proof}

Finally, we bound $\sum_{\cT \in \vec{\cT}}\errorterm_5(\cT) = \sum_{\cT \in \vec{\cT}}
[\mu_t\cdot\aloss_{t-1}(\vec{\cE}(\cT)) - (1-\mu_t)\cdot \outside(\vec{\cE}(\cT))]$.
%

\begin{claim} \label{clm:et5i}
\begin{eqnarray*}
\sum_{\cT \in \vec{\cT}} \errorterm_5(\cT) & \leq & \width(\cB)/5\errorcont.
\end{eqnarray*}
\end{claim}

This claim is analogous to \Clm{et5} but the proof has a crucial difference.
In \Clm{et5}, we obtained a bound for $\errorterm_5(\cT)$ for each $\cT$ separately and summed the bound.  In the analysis of the improved algorithm we no longer are able to separately bound each of the terms $\errorterm_5(\cT)$, rather we must
look at the entire sum and bound it.   This difference represents an important subtlety in the improved algorithm.

We explain what goes wrong if we try to follow the proof of \Clm{et5}.
Focus on the case when $\cT$ is wide.
The dichotomy lemma, \Lem{dichotomy}, gives the bound $(1-\mu_t)\cdot \outside(\vec{\cE}(\cT)) \geq \mu_t|U|$, where $U$ is (roughly speaking) the set of unsafe points in  $\vec{\cE}(\cT)$.
For this to be a useful bound, we must ensure that $|U|$ is small.
In the basic algorithm, $\findsplitter_t$ had been called on $\cT$ and failed, and this
ensures that there are few safe points in $\cT$.
In the improved algorithm, a failure of $\findsplitter_t$ does not terminate the procedure, but rather 
leads to a new phase
The terminal box $\cT$ is  not chosen
because $\findsplitter_t$ failed, but rather
because the threshold $\theta$ rises above $\width(\cT)$.   
 
When $\cT$ was declared terminal, the value of $\theta$ was equal to $\gamma_j \width(\cH)/\gridprec$ where $\cH$ is an
ancestor of $\cT$ in the splitter tree, such that $\findsplitter_t$ failed on $\cH$ in phase $j$.
This failure implies that the number of safe splitters with respect to 
$\strip{\cH}{\cB}$ is small compared to $\width(\cH)$. We cannot conclude that the number of safe
splitters with respect to $\strip{\cT}{\cB}$ is a small fraction of $\width(\cT)$, which is required
to bound $|U|$.
%

To overcome this problem, we need to take a more gloabal view.  Rather than
apply the dichotomy lemma to  each grid chain $\vec{\cD}(\cT)$, we   apply the dichotomy lemma 
to a single chain $\vec{\cD}$ that is obtained by piecing together all the grid chains $\vec{\cD}(\cT)$ (as defined in \Sec{stripnot}). 

\begin{proof} By \Prop{et5},
$$ \sum_{\cT \in \vec{\cT}} \errorterm_5(\cT) \leq 
\mu_t\cdot \sum_{\cT \in \vec{\cT}} \aloss_{t-1}(\vec{\cD}(\cT))
- (1-\mu_t)\cdot \sum_{\cT \in \vec{\cT}}\outside(\vec{\cD}(\cT)).$$

For each terminal box $\cT$, the grid $\Gamma(\cT)$ defines a strip decomposition $\vec{\cS}(\cT)$ of $\strip{\cT}{\cB}$ and $\vec{\cD}(\cT)$
is a chain compatible with $\vec{\cS}(\cT)$.  Let $\vec{\cS}$ denote the concatenation of $\vec{\cS}(\cT)$ (in left-to-right order), which
gives a strip decomposition of $\cB$, and let $\vec{\cD}$ denote the concatenation of $\vec{\cD}(\cT)$, which is a chain compatiable
with $\vec{\cS}$.  
We apply the dichotomy lemma to $\cB$ with strip decomposition $\vec{\cS}$ and get:
$$ (1-\mu_t)\cdot \sum_{\cT \in \vec{\cT}}\outside(\vec{\cD}(\cT)) \geq
\mu_t \cdot \sum_{\cT \in \vec{\cT}} |u(\cT)|,$$
%
%
%
where $u(\cT)$ is the set of $(\mu_t,\vec{\cD})$-unsafe indices $x$ satisfying $F(x) \in \vec{\cD}^{\cup} \cap \cT$. Thus,
$$ \sum_{\cT \in \vec{\cT}} \errorterm_5(\cT) \leq 
\mu_t\cdot \sum_{\cT \in \vec{\cT}} [\aloss_{t-1}(\vec{\cD}(\cT))
- |u(\cT)|].$$

%

Partition $\vec{\cT}$ into $\vec{\cT}^{narrow}$ and $\vec{\cT}^{wide}$ depending on
whether $\width(\cT) \leq 1/\gridprec$ or $\width(\cT) > 1/\gridprec$.
For a summand $\cT \in \vec{\cT}^{narrow}$  we have (as in the proof of \Clm{et5})
$\aloss_{t-1}(\vec{\cD}(\cT)) = 0$
and the summand is nonpositive.

For $\cT \in \vec{\cT}^{wide}$, each summand is (trivially) at most  
$|\vec{\cD}^{\cup} \cap \cT| - |u(\cT)| = |s(\cT)|$,
where $s(\cT)$ is the set of $(\mu_t,\vec{\cD})$-\emph{safe} indices $x$ satisfying $F(x) \in \vec{\cD}^{\cup} \cap \cT$.
Hence, it suffices
to bound $\sum_{\cT \in \vec{\cT}^{wide}}|s(\cT)|$.
In the proof of \Clm{et5}, we used the fact that $\findsplitter$ failed on $\cT$.   
For the improved algorithm, we will argue that $\findsplitter$ failed on an ancestor $\cH$ of $\cT$ in
the $t$-splitter tree and that every point of $s(\cT)$ is safe for $\strip{\cH}{\cB}$.

Recall that a box $\cH$  is a multiphase box if the path of boxes
in the $t$-execution tree corresponding to $\cH$ has more than one box.
\begin{prop}
\label{prop:bound-s}
For any wide terminal box $\cT$, there is a multiphase box $\cH$ in the $t$-splitter tree such that for $j=\phase(\cH)$, 
$(\cH,j-1)$ is an unsplittable pair along the $t$-execution tree path of $\cT$
and every index $s \in s(\cT)$ is $(\mu_t,\gamma_{j-1}\width(\cT))$-safe for $\strip{\cH}{\cB}$.
Furthermore if $s$ is nondegenerate (with respect to $\cT$) then it is also non-degenerate with respect to $\strip{\cH}{\cB}$.
\end{prop}
\begin{proof}
By \Prop{wide} there is an ancestor $\cH$ in the $t$-splitter tree such that $\gridprec \width(\cT) \leq \gamma_{j-1}\width(\cH)$.
Let $x \in s(\cT)$,  and let $\cD$ be $\vec{\cD}[x]$ (the box of $\vec{\cD}$ with $x \in X(\cD)$). 
We claim that $x$ is
$(\mu_t,\gamma_{j-1}\width(\cT))$-safe  for $\strip{\cH}{\cB}$.  We are given that $x$ is
$(\mu_t,\vec{\cD})$-safe, which means that it is $(\mu_t,\width(\cD))$-safe for $\cB$ and is
therefore $(\mu_t,\width(\cD))$-safe for the substrip $\strip{\cH}{\cB}$ of $\cB$.
Since $\width(\cD) \leq \gridprec \width(\cT) \leq \gamma_{j-1}\width(\cH)$, we have the desired conclusion.

If $s$ is nondegenerate with respect to $\cT$ then it is not equal to $x_R(\cT)$.  Since $\cT$ is a subbox of $\cH$,
$s \neq x_R(\cH)$.
\end{proof}
Since each terminal box $\cT$ has at most one degenerate splitter, we can bound the number of degenerate
splitters in a wide terminal box by $s(\cT) \leq \gridprec \width(\cT)$ (since any wide
box has width at least $\omega = 1/\gridprec$). The sum over all wide $\cT$ yields a bound of $\gridprec \width(\cB)$.

To account for the non-degenerate splitters, \Prop{bound-s} allows us to sum over all safe points in multiphase boxes.
For multiphase box $\cH$, let $s'(\cH)$ be the set of nondegenerate indices that are $(\mu_t,\gamma_{\phase(\cH)-1})$-safe for $\strip{\cH}{\cB}$.
Let $\vec{\cH}^j$ denote the sequence of multiphase boxes with $\phase(\cH)=j$ and $\width(\cH) > \omega$. 
Combining with the simple bound for degenerate splitters,
$$ \sum_{\cT \in \vec{\cT}^{wide}} s(\cT) \leq \gridprec \width(\cB) +\sum_{j} \sum_{\cH \in \vec{\cH}^j} |s'(\cH)| $$

Since $(\cH,\phase(\cH)-1)$ is unsplittable, by Assumption 1 and \Cor{fails}, $|s'(\cH)| \leq 3\rho_{\phase(\cH)-1}\width(\cH)$.

By \Prop{phase-box}, all boxes in $\vec{\cH}^j$ have disjoint index intervals, so $\sum_{\cH \in \vec{\cH}^j} \leq \width(\cB)$.
We put it all together and note that $\rho_j$ is a geometric progression. We use $j_{\max}$ to denote the largest
possible value of $\phase(\cH)$ and bound $\rho_{j_{\max}} \leq 2\gridprec^{1/4}$ (\Prop{phase}).
\begin{eqnarray*}
\sum_{\cT \in \vec{\cT}^{wide}} s(\cT) & \leq & \gridprec \width(\cB)+ \sum_{j} \sum_{\cH \in \vec{\cH}^j} 3\rho_{j-1}\width(\cH)\\
& \leq & \gridprec \width(\cB)+ \sum_{j} 3\rho_{j-1} \width(\cB)\\
& \leq &  \gridprec \width(\cB) + 6 \rho_{j_{\max}-1}  \width(\cB) \\
& \leq & (\gridprec + 3 \rho_{j_{\max}}) \width(\cB)\\ 
& \leq &  7\gridprec^{1/4}\width(\cB) \leq \width(\cB)/5\errorcont,
\end{eqnarray*}
useing the assumption that $C_2$ is sufficiently large.
\end{proof}  

This completes the proof of the approximation error bound for the improved algorithm.

\section{Acknowledgements} The second author would like to thank Robi Krauthgamer and David Woodruff
for useful discussions.

\bibliographystyle{alpha}
\bibliography{lis}

\end{document}